\documentclass[a4paper,11pt]{article}

\usepackage{amsthm,amsmath,amssymb,amsfonts,tikz}
\usepackage{color}
\usepackage[pdftex, bookmarksnumbered=true]{hyperref}

\usepackage{setspace}
\usetikzlibrary{arrows,positioning,cd,calc}
\usetikzlibrary{decorations.pathreplacing,decorations.markings,math}
\usepackage{pgffor}
\usepackage{upgreek}
\tikzset{
	on each segment/.style={
		decorate,
		decoration={
			show path construction,
			moveto code={},
			lineto code={
				\path [#1]
				(\tikzinputsegmentfirst) -- (\tikzinputsegmentlast);
			},
			curveto code={
				\path [#1] (\tikzinputsegmentfirst)
				.. controls
				(\tikzinputsegmentsupporta) and (\tikzinputsegmentsupportb)
				..
				(\tikzinputsegmentlast);
			},
			closepath code={
				\path [#1]
				(\tikzinputsegmentfirst) -- (\tikzinputsegmentlast);
			},
		},
	},
	mid arrow/.style={postaction={decorate,decoration={
				markings,
				mark=at position .5 with {\arrow[#1]{stealth}}
	}}},
}
\def\be{\begin{equation}\begin{gathered}}
\def\ee{\end{gathered}\end{equation}}
\newcommand{\rf}[1]{(\ref{#1})}
\def\stackreb#1#2{\mathrel{\mathop{#2}\limits_{#1}}}

\numberwithin{equation}{section}

\newtheorem{thm}{Theorem}[section]

\newtheorem{lemma}[thm]{Lemma}
\newtheorem{conj}[thm]{Conjecture}

\theoremstyle{definition}
\newtheorem{Remark}{Remark}[section]
\newtheorem{Example}[Remark]{Example}
\newtheorem{defin}[Remark]{Definition}

\newcommand{\eq}[1]{\begin{equation}\begin{gathered}
#1\end{gathered}\end{equation}}
\newcommand{\eqs}[1]{\begin{equation}\begin{gathered}\begin{split}
#1\end{split}\end{gathered}\end{equation}}

\def\lb{\left(}
\def\rb{\right)}

\def\d{\partial}
\def\mc{\mathcal}

\title{
Cluster Toda chains and Nekrasov functions}
\author{M. Bershtein, P. Gavrylenko, A. Marshakov}

\date{}

\begin{document}

\maketitle

\begin{abstract}\vspace*{2pt}
\noindent
In this paper the relation between the cluster integrable systems and $q$-difference equations is extended
beyond the Painlev\'e case.

We consider the class of hyperelliptic curves when the Newton polygons contain only four boundary points. The corresponding cluster integrable Toda systems are presented, and their discrete automorphisms are identified with certain reductions of the Hirota difference equation. We also construct non-autonomous versions of these equations and find that their solutions are expressed in terms of 5d Nekrasov functions with the Chern-Simons contributions, while in the autonomous case these equations are solved in terms of the Riemann theta-functions.
\end{abstract}

\tableofcontents

\newpage

%\dedicatory
\begin{flushright}
\textit{To the memory of L.D. Faddeev}
\end{flushright}

\section{Introduction}

In this paper, following \cite{BGM}, we continue the study of the relation between cluster integrable systems, $q$-difference equations and Nekrasov partition functions for 5d gauge theories, and extend it to the class of
theories with the higher rank gauge groups.

Recall the main conjecture of \cite{BGM}. First, to any \textit{Newton polygon $\Delta$} one can assign the cluster integrable system \cite{GK}, \cite{FM:2014}. The phase space of this system is \textit{$X$-cluster} variety $\mathcal{X}_\mathcal{Q}$ with the Poisson bracket defined by the quiver~$\mathcal{Q}$. The group $\mathcal{G}_\mathcal{Q}$ of discrete automorphisms  acts on $\mathcal{X}_\mathcal{Q}$, preserving the integrals of motion of the cluster integrable system. After \textit{deautonomization} the action $\mathcal{G}_\mathcal{Q}$ leads to $q$-difference equations, which are equations of $q$-isomonodromic deformations. Finally, these equations can be explicitly solved using Nekrasov functions of 5d supersymmetric gauge theory or topological strings amplitudes for the toric Calabi-Yau $CY_\Delta$. The Seiberg-Witten curve for corresponding supersymmetric gauge theory and corresponding toric Calabi-Yau manifold are constructed from the same Newton polygon~$\Delta$.

The statement about solutions to the $q$-difference equations is in fact a generalization of the ($q$-deformed) Isomonodromy/CFT correspondence \cite{GIL1207}, \cite{GIso}. Moreover, it has been recently proposed in \cite{BGM} that after quantization of the Poisson variety $\mathcal{X}_\mathcal{Q}$ the corresponding \emph{quantum} $q$-difference equations are solved using the \emph{refined} topological strings partition functions,
depending also on multiplicative quantum parameter in addition to the parameter of $q$-deformation. In terms of the Isomonodromy/CFT correspondence this quantization leads to generalization for the case of arbitrary central charge.

This proposal has been verified in \cite{BGM} for the class of polygons $\Delta$ with a single interior integer point. The corresponding $q$-difference equations are well-known discrete Painlev\'e equations \cite{SakaiCMP}. Note that in this case the Poisson bracket on $\mathcal{X}_\mathcal{Q}$ has rank two, and integrable system is almost trivial (any integral of motion is function of the Hamiltonian), however the group $\mathcal{G}_\mathcal{Q}$ can be already very nontrivial. In terms of combinatorics of $\Delta$ the rank of the Poisson bracket on $\mathcal{X}_\mathcal{Q}$ is twice the number of the interior integer points in $\Delta$ (equal to the number of independent integrals of motion), while number of commuting discrete flows (or the rank of corresponding Abelian subgroup in $\mathcal{G}_\mathcal{Q}$) is related to the number $B$ of the integer points on the border of $\Delta$ (literarly equals to $B-3$).

In this paper we consider the ``opposite'' case where polygon $\Delta$ has only four integer points on the border, then $\mathcal{G}_\mathcal{Q}$ contains only  rank one integer lattice. Generator of this lattice shifts Casimir variable $z\mapsto qz$. Moreover, we restrict ourselves to the subclass of polygons, when all interior integer points in $\Delta$ belong to the same line, and denote the number of such interior points by $N-1$. The spectral curves corresponding to such Newton polygons are always hyperelliptic. We classify the corresponding Newton polygons in Theorem~\ref{th:polygons}: they belong either to $Y^{N,k}$, with $0\leq k \leq N$, or to $L^{1,2N-1,2}$ families (here in notations we follow \cite{Hanany1}, \cite{Hanany2}).
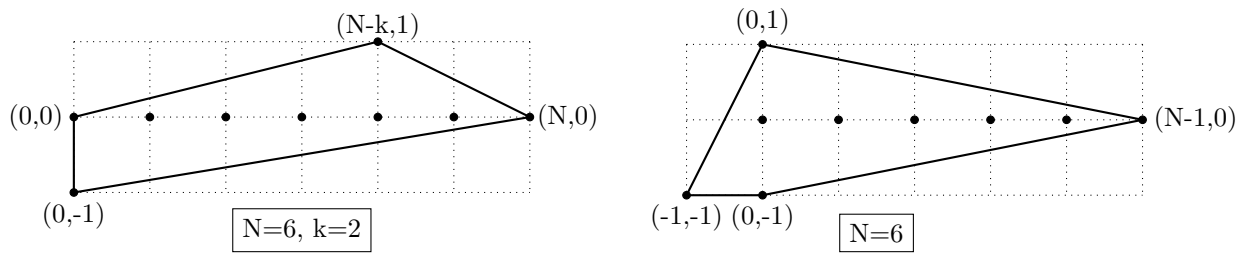
\begin{figure}[h]

\begin{center}
\begin{tabular}{cc}
\begin{tikzpicture}[scale=1, font = \small, rotate=-90]

\tikzmath{\N=6;\k=2;\Nk=int(\N-\k);}

\draw[dotted] (-1,0) grid (1,\N);

\begin{scope}[%decoration={markings,mark=at position 0.5 with {\arrow{Stealth[scale=1.5]}}},
thick]
\draw[postaction={decorate}] (0,0)--(1,0);\draw[postaction={decorate}] (1,0)--(0,\N);
\draw[postaction={decorate}] (0,\N)--(-1,\N-\k);\draw[postaction={decorate}](-1,\N-\k)--(0,0);
\end{scope}

\foreach \x in {0,...,\N} \draw[fill] (0,\x) circle[radius=0.05];
\draw[fill] (1,0) circle[radius=0.05];
\draw[fill] (-1,\N-\k) circle[radius=0.05];

\node at (0,\N+0.5) {(N,0)};
\node at (1.3,0) {(0,-1)};
\node at (0,-0.5) {(0,0)};
\node at (-1.2,\N-\k) {(N-k,1)};

\node[draw] at (1.5,\N/2) {N=\N, k=\k};

\end{tikzpicture}
&
\begin{tikzpicture}[scale=1, font = \small, rotate=-90]

\tikzmath{\Ni=6;\N=\Ni-1;}

\draw[dotted] (-1,-1) grid (1,\N);

\begin{scope}[%decoration={markings,mark=at position 0.5 with {\arrow{Stealth[scale=1.5]}}},
thick]
\draw[postaction={decorate}] (-1,0)--(1,-1)--(1,0);\draw (1,0)--(0,\N);
\draw (0,\N)--(-1,0);
\end{scope}

\foreach \x in {0,...,\N} \draw[fill] (0,\x) circle[radius=0.05];
\draw[fill] (1,0) circle[radius=0.05];
\draw[fill] (1,-1) circle[radius=0.05];
\draw[fill] (-1,0) circle[radius=0.05];

\node at (0,\N+0.7) {(N-1,0)};
\node at (1.3,0) {(0,-1)};
\node at (1.3,-1) {(-1,-1)};
\node at (-1.3,0) {(0,1)};

\node[draw] at (1.5,\N/2-1) {N=\Ni};

\end{tikzpicture}
\end{tabular}
\end{center}
	\caption{Examples of $Y^{N,k}$ polygon for $N=6$, $k=2$ and $L^{1,2N-1,2}$ with $N=6$.}\label{fig:YNk}
\end{figure}
Integrable systems, corresponding to $Y^{N,k}$ polygons, were studied e.g. in \cite{Brini}, \cite{Eager:2011}. In the case of $Y^{N,0}$ it is standard affine relativistic Toda chain with $N$ particles (see details e.g. in \cite{AMJGP}), for other $Y^{N,k}$ they can be viewed as different affinizations of the same open Toda chain.

Below we show that deautonomization of discrete flows of these integrable systems can be written in a form of bilinear Hirota equations. For $Y^{N,k}$-case these equations are
\eq{\label{eq:taulat}
	\tau_{(n,m+1)}\tau_{(n,m-1)}=\tau_{(n,m)}^2+ z_0^{1/N} q^{\frac{kn-Nm}{N^2}}
\tau_{(n+1,m)}\tau_{(n-1,m)}
}
with the boundary conditions $\tau_{(n+k,m+N)}=\tau_{(n,m)}$.
One can also rewrite equations \eqref{eq:taulat} as difference equations in the variable $z=z_0q^{\frac{kn-Nm}N}$
\eq{\label{eq:tauJ}
	\tau_{j}\lb q z\rb \tau_{j}\lb q^{-1} z\rb=\tau_j(z)^2+z^{1/N} \tau_{j+1}\lb q^{k/N}z\rb \tau_{j-1}\lb q^{-k/N}z\rb\,,
}
on $N$ tau-functions $\{\tau_{j}(z)|j\in \mathbb{Z}/N\mathbb{Z}\}$. Similar difference equation is derived for the case of $L^{1,2N-1,2}$.

We propose generic solution of the difference equations \rf{eq:tauJ} in Conjecture~\ref{conj:bilinTau}. The result is given in terms of topological string amplitudes for the $Y^{N,k}$ geometry, which in this case are equal to Nekrasov partition functions for 5d pure $SU(N)$ supersymmetric gauge theory with Chern-Simons term at level $k$, \cite{IKP02}, \cite{Eguchi:2003}, \cite{Ta}. Substituting our solution into equations \eqref{eq:tauJ} reduces them to bilinear relations on Nekrasov partition functions, similar to the blow-up equations of \cite{GNY} (but for another geometry --- blow-up of $\mathbb{C}^2/\mathbb{Z}_2$).

We also present solution of the autonomous version of the equations \eqref{eq:taulat}, and their analog for the $L^{1,2N-1,2}$ geometry, for $q=1$. The corresponding tau functions are essentially given by the Riemann theta-functions on the Jacobians of hyperelliptic curves with the Newton polygons $Y^{N,k}$: then bilinear relations \eqref{eq:taulat} reduce to the Fay identity. We conclude our discussion with some remarks about 4d limit and present few comments about quantization.

\section{Cluster integrable systems and Toda chains}

\label{sec:Todas}

\subsection{Newton polygons} \label{ssec:polygons}

For a \textit{convex lattice} polygon $\Delta$ in $\mathbb{R}^2$ (with all vertices in  $\mathbb{Z}^2\subset\mathbb{R}^2$), 
one can write the Laurent polynomial $f_\Delta(\lambda,\mu)$, so that equation
\be
\label{SC}
f_\Delta(\lambda,\mu) = \sum_{(a,b)\in\Delta}\lambda^a\mu^b f_{a,b}=0.
\ee
defines a plane (noncompact) spectral curve in $\mathbb{C}^\times\times \mathbb{C}^\times$. Polygons related by the action of  $SA(2,\mathbb{Z})=SL(2,\mathbb{Z})\ltimes \mathbb{Z}^2$ lead to the same integrable systems.

\begin{defin}
By the \emph{Toda family curves} we call the curves \rf{SC} that are hyperelliptic, and their Newton polygons have 4 boundary integer points.
\end{defin}

\noindent
It turns out that all such curves can be
classified. First, it is well-known that there are only three $N=2$ Toda family curves with a single internal point (see $4_a, 4_b, 4_c$ in \cite{BGM}). Here we present classification for generic $N>2$ case.

\begin{lemma} All internal points of the Newton polygon $\Delta$ for a hyperelliptic curve belong to a single straight line. \label{hyperellipticity}
\end{lemma}
\begin{proof}
Suppose that there are three points inside the Newton polygon $\Delta$ with the coordinates $(a,b), (a+1,b), (a,b+1)$~\footnote{Or some $SL(2,\mathbb Z)$ image of this triple.}. Then the ratios of corresponding holomorphic 1-forms $dv_{a,b} = \frac{\lambda^{a-1}\mu^{b-1}d\mu}{\d f_\Delta/\d\lambda}$
\eq{
\frac{dv_{(a+1,b)}}{dv_{(a,b)}}=\lambda,\quad \frac{dv_{(a,b+1)}}{dv_{(a,b)}}=\mu
}
give just two coordinate functions $(\lambda,\mu)$. It means that the canonical map:
${\mathcal C}\to {\mathbb P}\lb H^{1,0}(\mathcal{C},\mathbb C)\rb$ is non-degenerate, what contradicts to the well-known fact that canonical map is
degenerate iff the curve is hyperelliptic.
\end{proof}

\begin{lemma} All points of the Newton polygon $\Delta$ of Toda family curve can be placed on three consecutive  horizontal lines.
\end{lemma}
\begin{proof}
Using the action of $SA(2,\mathbb{Z})$ one can move all internal integer points to $(1,0)\ldots (N-1,0)$. Assume that there is at least one point with coordinates $(a,2)$ in $\Delta$. Then the triangle with the vertices  $\{(a,2),(1,0),(2,0)\}$ belongs to $\Delta$, this triangle contains integer boundary point $(\frac{a+1}2,1)$ or $(\frac{a+2}2,1)$, depending of parity of $a$. This point should be internal since $(1,0)$ and $(2,0)$ are internal, and this contradicts to
Lemma~\ref{hyperellipticity}. Therefore all points of $\Delta$ belong to three horizontal lines with $y=1$, $y=0$ and $y=-1$.
\end{proof}
Sometimes, for convenience, we rotate a polygon from the Toda family, so that it lie on three consecutive vertical lines.

\begin{thm}\label{th:polygons}
	The Newton polygon of a Toda family curve is $SA(2,\mathbb{Z})$ equivalent to one of the following polygons:
\eqs{
Y^{N,k}&=\{(0,0),(0,-1),(N-k,1),(N,0)\},\quad k=0,\ldots,N\\
L^{1,2N-1,2}&=\{(-1,-1),(0,-1),(0,1),(N-1,0)\}.
}
Here we specify a convex polygon by listing its boundary integer points.
\end{thm}
\begin{proof}
Notice first that since all four boundary integer points should be placed on three lines, two of them necessarily belong to the same line. If this line is the line
of internal integer points, then two remaining points should lie above and below --- this leads to the $Y^{N,k}$ class. If this is the bottom (or top) line, the distance between these two points should be unit,
just to avoid extra boundary points, and this leads to the $L^{1,2N-1,2}$ class.
\end{proof}

According to \cite{GK}, \cite{FM:2014} any convex Newton polygon $\Delta$ defines a cluster integrable system on a Poisson $X$-cluster variety $\mathcal{X}$ of dimension $\dim\mathcal{X} = 2S$, where $S$ is an area of the polygon $\Delta$. The Poisson structure can be encoded by the quiver (oriented graph) $\mathcal{Q}$ with $2S$ vertices and
antisymmetric exchange matrix $\epsilon$, where $\epsilon_{ij}$ equals to the number of arrows from $i$-th to $j$-th vertex of $\mathcal{Q}$ minus number of arrows from $j$-th to $i$-th vertex of $\mathcal{Q}$.

Below we construct these quivers for all Newton polygons of the Toda family curves, following the general algorithm of \cite{GK}, i.e. for each polygon we present a Thurston diagram, a bipartite graph on torus, and a quiver.

\subsection{Thurston diagram and quivers for $Y^{N,k}$ systems} \label{ssec:Thurst}

According to the Goncharov-Kenyon algorithm, one starts with orienting all boundaries of the Newton polygon counterclockwise and considering
them as closed loops on torus $\mathbb{R}^2/\mathbb{Z}^2$. Then one has to deform these loops to certain smooth curves in order to get only
triple intersections with alternating orientations of the incoming curves.

\tikzset{
thurston1/.pic=
{\tikzset{
scale=0.5, font = \small,%thick,
styleFill/.style={fill,gray,opacity=0.2},
styleArrow/.style={postaction={decorate},decoration={markings,mark=at position 0.5 with {\arrow{Stealth[scale=1.5]}}}},
blackCircle/.style={fill=black,thick,radius=0.1,inner sep=0},
whiteCircle/.style={fill=white,thick,radius=0.1,inner sep=0},
}

%\draw[dotted] (0,0) grid (6,6);
\draw[dotted,scale=6] (0,0) grid (1,1);

\path[styleFill] (0,4)--(2,2)--(4,2)--(0,6);
\path[styleFill] (6,3) to [in=70,out=180] (4,2)--(6,2);
\path[styleFill] (4,2) to [in=-15,out=250] (2,2) -- (4,0)--(6,0);
\path[styleFill] (0,2) to (2,2) to [in=0,out=165] (0,3);
\path[styleFill] (6,4) to (6,6) to (4,6);

\draw[blue,styleArrow] (6,0)--(0,6);
\draw[red,styleArrow] (0,2)--(6,2);
\draw[olive,postaction=decorate,decoration={markings,mark=at position 0.65 with {\arrow{Stealth[scale=1.5]}}}] (0,4)--(4,0);
\draw[olive,styleArrow] (4,6)--(6,4);
\draw[orange,styleArrow] (6,3) to [in=70,out=180] (4,2) to [in=-15,out=250] (2,2) to [in=0,out=165] (0,3);

\draw[rounded corners=5,inner sep=0] (0,3.5) -- (0.2,3.5) --
(2,2) node (b1){}
(b1)--
(3.5,1.82) node (w1) {}
(w1)--
(4,2) node (b2) {}
(b2)--
(2,5) node (w2){}
(w2)--
(5.5,3.5)--(6,3.5)

(b2)--(5.5,1.5)--(6,1)
(0,1)--(1,0)
(1,6)--(w2)

(b1)--(2,0)
(2,6)--(w2);

\draw[whiteCircle] (w1) circle;
\draw[whiteCircle] (w2) circle;
\draw[blackCircle] (b1) circle;
\draw[blackCircle] (b2) circle;

}
}

\tikzset{thurston2/.pic={
\tikzset{
scale=0.5, font = \small,%thick,
styleFill/.style={fill,gray,opacity=0.2},
styleArrow/.style={postaction={decorate},decoration={markings,mark=at position 0.45 with {\arrow{Stealth[scale=1.5]}}}},
blackCircle/.style={fill=black,thick,radius=0.1,inner sep=0},
whiteCircle/.style={fill=white,thick,radius=0.1,inner sep=0}
}

%\draw[dotted] (0,0) grid (6,6);
\draw[dotted,scale=6] (0,0) grid (1,1);

\path[styleFill] (6,3) to [out=180,in=70] (4,2) to (6,2);
\path[styleFill] (6,0) to (4,2) to [in=45,out=-120](3,0);
\path[styleFill] (0,2) to (4,2) to (2,4) to [in=0,out=-90](0,3);
\path[styleFill] (2,4) to [in=135,out=15](6,4) to (6,6) to (3,6) to [out=-135,in=90] (2,4);
\path[styleFill] (2,4) to (0,6) to (0,4) to [in=195,out=-45] (2,4);

\draw[blue,styleArrow] (6,0)--(0,6);
\draw[red,styleArrow] (0,2)--(6,2);
\draw[olive,styleArrow] (0,4) to [in=195,out=-45] (2,4) to [in=135,out=15](6,4);
\draw[orange,postaction={decorate},decoration={markings,mark=at position 0.4 with {\arrow{Stealth[scale=1.5]}}}] (3,6) to [out=-135,in=90] (2,4) to [in=0,out=-90](0,3);
\draw[orange,styleArrow] (6,3) to [out=180,in=70] (4,2) to [in=45,out=-120](3,0);

\draw[rounded corners=5,inner sep=0] (0,3.5) -- (1.5,3.5) --
(2,4) node (b1){}
(b1)--
(4,3.5) node (w1) {}
(w1)--
(4,2) node (b2) {}
(b2)--
(1.5,1) node (w2){}
(w2)--(1.5,0)

(w1)--(6,3.5)
(0,1)--(0.2,0.8)--(w2)
(b2)--(5.5,1.5)--(6,1)

(b1)--(1.5,5.5)--(1.5,6);

\draw[whiteCircle] (w1) circle;
\draw[whiteCircle] (w2) circle;
\draw[blackCircle] (b1) circle;
\draw[blackCircle] (b2) circle;
}
}

\tikzset{thurston3/.pic={

\tikzset{
scale=0.5, font = \small,%thick,
styleFill/.style={fill,gray,opacity=0.2},
styleArrow/.style={postaction={decorate},decoration={markings,mark=at position 0.45 with {\arrow{Stealth[scale=1.5]}}}},
blackCircle/.style={fill=black,thick,radius=0.1,inner sep=0},
whiteCircle/.style={fill=white,thick,radius=0.1,inner sep=0},
}

%\draw[dotted] (0,0) grid (6,6);
\draw[dotted,scale=6] (0,0) grid (1,1);

\path[styleFill] (0,4) to [out=-45,in=90] (1.5,2) to (4,2) to (0,6);
\path[styleFill] (0,3) to [out=0] (1.5,2) to (0,2);
\path[styleFill] (1.5,2) to [out=-90,in=135] (2,0) to (3.5,0);
\path[styleFill] (4,2) to[out=-90,in=135] (4.5,0) to (6,0);
\path[styleFill] (2,6) to[in=90,out=-45] (4,2) to (6,2) to (6,3) to [out=180, in=-45] (3.5,6);
\path[styleFill] (4.5,6) to[out=-45,in=135] (6,4) to (6,6);

\draw[blue,styleArrow] (6,0)--(0,6);
\draw[red,styleArrow] (0,2)--(6,2);
\draw[orange,postaction={decorate},decoration={markings,mark=at position 0.65 with {\arrowreversed{Stealth[scale=1.5]}}}] (0,3) to [out=0] (1.5,2) to (3.5,0);
\draw[orange,postaction={decorate},decoration={markings,mark=at position 0.45 with {\arrowreversed{Stealth[scale=1.5]}}}] (3.5,6) to[in=180, out=-45] (6,3);

\draw[olive,styleArrow] (0,4) to [out=-45,in=90] (1.5,2) to [out=-90,in=135] (2,0);
\draw[olive,styleArrow] (2,6) to[in=90,out=-45] (4,2)
(4,2) to[out=-90,in=135] (4.5,0);
\draw[olive,styleArrow] (4.5,6) to[out=-45,in=135] (6,4);

\draw[inner sep=0] (0,3.5) to[out=0,in=140] (0.5,3.25) to [out=320,in=120] (1.5,2) node (b1) {}
to (3,1.5) node (w2) {}
to (4,2) node (b2){}
to[in=135,out=-22.5] (6,1)
(w2) to[in=135,out=-70] (4,0)
(4,6) to[out=-45,in=180] (6,3.5)
(1,6) to[out=-45,in=112.5] (b2)
(b1) to (0.5,0.5) node (w1) {}
(1,0) to (0,1);

\draw[whiteCircle] (w1) circle;
\draw[whiteCircle] (w2) circle;
\draw[blackCircle] (b1) circle;
\draw[blackCircle] (b2) circle;

}
}

\noindent\begin{figure}[h]

\noindent\begin{tabular}{cccc}

\begin{tikzpicture}[scale=0.7, font = \scriptsize]

\tikzmath{\N=6;\k=2;\Nk=int(\N-\k);}

\draw[dotted] (-1,0) grid (1,\N);

\begin{scope}[decoration={markings,mark=at position 0.55 with {\arrow{Stealth[scale=1.5]}}},thick]
\draw[red,postaction={decorate}] (0,0)--(1,0);\draw[blue,postaction={decorate}] (1,0)--(0,\N);
\draw[orange,postaction={decorate}] (0,\N)--(-1,\N-\k);\draw[olive,postaction={decorate}](-1,\N-\k)--(0,0);
\end{scope}

\foreach \x in {2,...,\N} \draw[fill] (0,\x-1) circle[radius=0.05];

\node at (0,\N+0.3) {(0,N)};
\node at (1.5,0) {(1,0)};
\node at (0,-0.4) {(0,0)};
\node at (-1.8,\N-\k) {(-1,N-k)};

\node[draw] at (0,-1) {N=\N, k=\k};

\end{tikzpicture}

&

\begin{tikzpicture}
\draw (0,0) (0,0.5) pic[scale=1.3]{thurston1};

\end{tikzpicture}

&

\begin{tikzpicture}
\draw (0,0) (0,0.5) pic[scale=1.3]{thurston2};

\end{tikzpicture}

&

\begin{tikzpicture}
\draw (0,0) (0,0.5) pic[scale=1.3]{thurston3};

\end{tikzpicture}

\end{tabular}
\caption{Newton polygon and building blocks for the Thurston diagrams: type 0, $\mathrm{I}$ and $-\mathrm{I}$ blocks.}
\label{fig:Thurston1}
\end{figure}
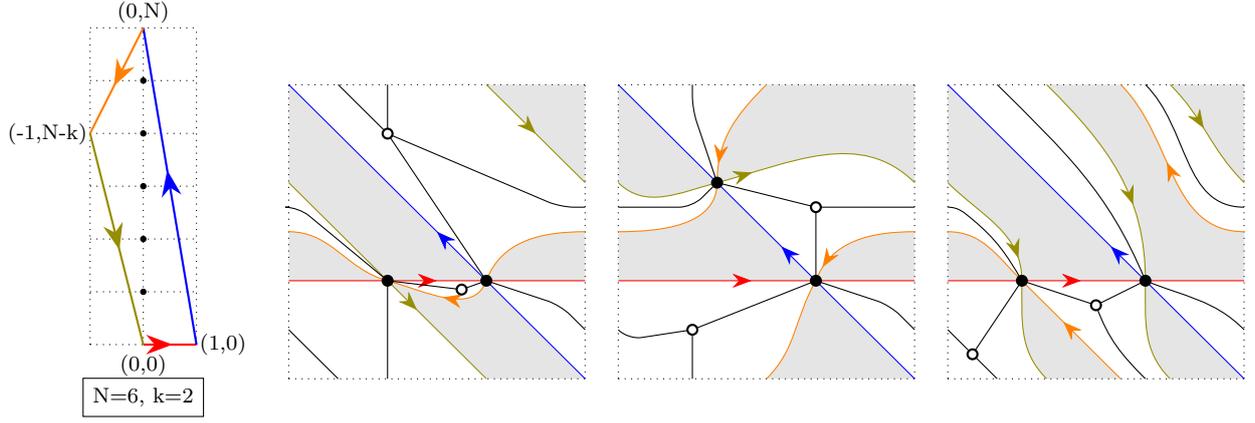
\begin{lemma}\label{Lem:YNkThurst}
The Thurston diagram on torus for $Y^{N,k}$ polygon can be given by $N_0$ type 0 blocks, $N_1$ type
$\mathrm{I}$ blocks and $N_{-1}$ type $-\mathrm{I}$ blocks from Fig.~\ref{fig:Thurston1}, where $N=N_0+N_1+N_{-1}$, $k=N_1-N_{-1}$. Heights of these blocks equal to $1$,
whereas their widths equal to $1/N$. The order of blocks can be arbitrary.
\end{lemma}
\begin{proof}
First let us list the homology classes of closed cycles (see Fig.~\ref{fig:Thurston1}):
\eq{
	C^{N,k}=\{\color{blue}(-1,N),\color{orange}(-1,-k),\color{olive}(1,k-N),\color{red}(1,0)\color{black}\}
}
and compare them with homology classes corresponding to the building blocks, $c^0$, $c^\mathrm{I}$ and $c^{-\mathrm{I}}$:
\eqs{
	 c^0&=\{\color{blue}(-\frac1N,1),\color{orange}(-\frac1N,0),\color{olive}(\frac1N,-1),\color{red}(\frac1N,0)\color{black}\}\,,\\
	 c^{\mathrm{I}}&=\{\color{blue}(-\frac1N,1),\color{orange}(-\frac1N,-1),\color{olive}(\frac1N,0),\color{red}(\frac1N,0)\color{black}\}\,,\\
	c^{-\mathrm{I}}&=\{\color{blue}(-\frac1N,1),\color{orange}(-\frac1N,	 1),\color{olive}(\frac1N,-2),\color{red}(\frac1N,0)\color{black}\}\,.
}
One obviously gets $N_0c^0+N_1c^{\mathrm{I}}+N_{-1}c^{-\mathrm{I}}=C^{N,k}$. Since each of the blocks satisfies all requirements to the Thurston diagrams it is enough to check this equality in homologies.\end{proof}

\begin{Remark}
	If $N_{-1}>N_{1}$, one gets the Thurston diagram for $Y^{N,k}$ with negative $k$. But polygon $Y^{N,k}$ is $SA(2,\mathbb{Z})$ equivalent to $Y^{N,-k}$, therefore we restrict ourselves below to the case of non-negative $k$.
\end{Remark}

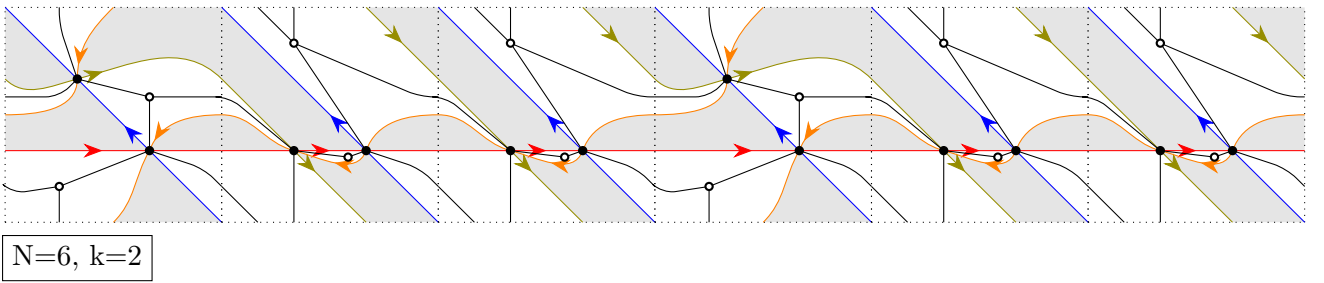
\begin{figure}[h!]
\begin{center}
\begin{tikzpicture}
\tikzmath{
int \N,\k,\Nk; \N=6; \k=2; \Nk=\N-\k;
\Lx=\N;\Ly=1;
\scale=0.95;}

\begin{scope}[scale=\scale]

\tikzmath{
for \z in {0,...,\Ly-1}{
\j=0;
for \i in {-1,...,\Lx-2}{
  if ((\i+1)*\k/\N)>=\j then{
    \j=\j+1;
    {\draw  (\i*3,\z*3) pic[scale=\scale] {thurston2};};
%    if mod(\i+1,\N)==0 then{{\draw (\i*3,0)--(\i*3,\Ly*3);};};
  } else{
    {\draw  (\i*3,\z*3) pic[scale=\scale] {thurston1};};};
  };
};
}

\node[draw] at (-2,-0.5) {N=\N, k=\k};

\end{scope}
\end{tikzpicture}
\end{center}
\caption{Example of Thurston diagram for $(N,k)=(6,2)$. Notice that interested reader can
download the source file from \texttt{arXiv} and get more examples just modifying the definitions of the variables \textbackslash N and \textbackslash k in the Ti\emph{k}Z code.}
\label{Thurston2}
\end{figure}

The next step is to construct a bipartite graph. It can be drawn in the following way: the white vertices are placed in the white-colored faces (with the clockwise orientation), the black vertices are placed in the triple intersections, and the edges are drawn between the neighboring white and black vertices (belonging to the same white face). After the graph is drawn one may erase the valence two vertices and shrink corresponding edges.

As well as Thurston diagrams, the bipartite graphs can be constructed from the building blocks from
Fig.~\ref{fig:Thurston1}: one example of gluing is shown in Fig.~\ref{Thurston2}. The gluing rules are more transparently shown at Fig.~\ref{fig:blocks1}, where
the Thurston diagrams are omitted, but instead the faces are shown --- together with the arrows of the quiver and labels ``$\times$'' or ``$+$''. The arrows encode Poisson structure for the cluster variables attached to
faces of bipartite graph on torus. The corresponding quiver is constructed in the following way: one draws $n$ arrows clockwise around valence $n$ black vertices and then counts the total
number of arrows connecting all pairs of different faces. The labels ``$\times$'' or ``$+$'' are put for
correct gluing, giving rise to a bipartite graph --- one should just connect the labels of the same type.

\tikzset{cross/.pic={
\draw[line width=2.5*#1,scale=#1](-0.1,-0.1)--(0.1,0.1)
(-0.1,0.1)--(0.1,-0.1);
}
}

\tikzset{dimer1/.pic={
\tikzset{
scale=0.5, font = \small,thick,
styleArrow/.style={postaction={decorate},decoration={markings,mark=at position 0.45 with {\arrow{Stealth[scale=1.5]}}}},
blackCircle/.style={fill=black,thick,radius=0.2,inner sep=0},
whiteCircle/.style={fill=white,thick,radius=0.2,inner sep=0}
}

\begin{scope}[yshift=-3cm,thick]
%\draw[dotted] grid[yscale=1] (4,8);
%\draw[dotted] grid[yscale=8,xscale=4] (1,1);

\draw [rounded corners=0,inner sep=0] (2,0)--(2,3) node(b){}--(2,7) node(w){}--(2,8)
(0,5)--%(0.3,5)--
(b)--%(3.7,1)--
(4,1)
(1,8)--(w)
(w)--(3,6)--%(3.7,5)--
(4,5)
(0,1) to %[in=135,out=0]
(1,0);

\draw[whiteCircle,fill=white] (w) circle;
\draw[blackCircle] (b) circle;
\end{scope}

}
}

\tikzset{dimer2/.pic={

\tikzset{
scale=0.5,
font = \small,
styleArrow/.style={postaction={decorate},decoration={markings,mark=at position 0.45 with {\arrow{Stealth[scale=1.5]}}}},
blackCircle/.style={fill=black,thick,radius=0.2,inner sep=0},
whiteCircle/.style={fill=white,thick,radius=0.2,inner sep=0},
}
\begin{scope}[yshift=-1cm,thick]

%\draw[dotted] grid (4,8);
%\draw[dotted] grid[yscale=8,xscale=4] (1,1);

\draw [rounded corners=5,inner sep=0]
(0,3)--(1,2) node(w1){}--(3,0) node(b1){}--(4,-1)
(0,7)--(1,6) node(b2){}--(3,4) node(w2){}--(4,3)
(1,0)--(w1)
(b1)--(w2)
(b2)--(1,8);

\draw[whiteCircle] (w1) circle;
\draw[blackCircle] (b1) circle;
\draw[whiteCircle] (w2) circle;
\draw[blackCircle] (b2) circle;
\end{scope}

}
}

\tikzset{dimer3/.pic={
		
		\tikzset{
			scale=0.5,
			font = \small,
			styleArrow/.style={postaction={decorate},decoration={markings,mark=at position 0.45 with {\arrow{Stealth[scale=1.5]}}}},
			blackCircle/.style={fill=black,thick,radius=0.2,inner sep=0},
			whiteCircle/.style={fill=white,thick,radius=0.2,inner sep=0},
		}
		\begin{scope}[yshift=-1cm,thick]
			
			%\draw[dotted] grid (4,8);
			%\draw[dotted] grid[yscale=8,xscale=4] (1,1);
			
			\draw [rounded corners=5,inner sep=0,yslant=-2]
			(0,3)--(-1,2) node(w1){}--(-3,0) node(b1){}--(-4,-1)
			(0,7)--(-1,6) node(b2){}--(-3,4) node(w2){}--(-4,3)
			(-1,0)--(w1)
			(b1)--(w2)
			(b2)--(-1,8);
			
			\draw[blackCircle] (w1) circle;
			\draw[whiteCircle] (b1) circle;
			\draw[blackCircle] (w2) circle;
			\draw[whiteCircle] (b2) circle;
		\end{scope}
		
	}
}

\tikzset{arrows1/.pic={
\tikzset{
scale=0.5, font = \small,
styleArrow/.style={blue,postaction={decorate},decoration={markings,mark=at position 0.6 with {\arrow{Stealth[scale=1.5]}}}},
styleLine/.style={blue}}
\tikzmath{\offset=0.25cm;}

\draw[styleArrow,shorten <=\offset,shorten >=\offset](0,0)..controls(0.2,2)..(0,4);
\draw[styleArrow,shorten <=\offset,shorten >=\offset](0,4)--(4,0);
\draw[styleArrow,shorten <=\offset](4,0) to [in=90,out=-100](3.9,-3);
\draw[styleLine,shorten >=\offset](3.9,5) to [out=-90,in=90](4,4);
\draw[styleLine,shorten <=\offset](4,4)--(3,5);
\draw[styleArrow,shorten >=\offset] (3,-3)--(0,0);

}
}

\tikzset{vertices1/.pic={
\tikzset{font = \small,
styleArrow/.style={blue,postaction={decorate},decoration={markings,mark=at position 0.6 with {\arrow{Stealth[scale=1.5]}}}},
styleLine/.style={blue}}
\tikzmath{\offset=0.25cm;}

\tikzset{scale=#1}

  \draw(0,0) pic[magenta,rotate=45]{cross=#1};
  \draw(0,2) pic[teal]{cross=#1};
  \draw(2,0) pic[magenta,rotate=45]{cross=#1};
  \draw(2,2) pic[teal]{cross=#1};

}
}

\tikzset{arrows2/.pic={

\begin{scope}[yshift=2cm]

\tikzset{
scale=0.5, font = \small,
styleArrow/.style={blue,postaction={decorate},decoration={markings,mark=at position 0.6 with {\arrow{Stealth[scale=1.5]}}},shorten <=\offset,shorten >=\offset},
styleLine/.style={blue}}
\tikzmath{\offset=0.25cm;}

\draw[styleArrow](0,0)..controls(0.2,2)..(0,4);
\draw[styleArrow](0,4)--(4,0);
\draw[styleArrow](4,0)..controls(2,-0.3)..(0,0);

\draw[styleArrow](4,-8)--(0,0);
\draw[styleArrow](4,-4)..controls(3.8,-6)..(4,-8);
\draw[styleArrow] (0,0)--(4,-4);

\end{scope}
}
}

\tikzset{vertices2/.pic={

\tikzset{
font = \small,
styleArrow/.style={blue,postaction={decorate},decoration={markings,mark=at position 0.6 with {\arrow{Stealth[scale=1.5]}}},shorten <=\offset,shorten >=\offset},
styleLine/.style={blue}}
\tikzmath{\offset=0.25cm;}

\tikzset{scale=#1}

\draw(0,0) pic[teal]{cross=#1};
\draw(0,2) pic[magenta,rotate=45]{cross=#1};
\draw(2,2) pic[teal]{cross=#1};
\draw(2,0) pic[magenta,rotate=45]{cross=#1};

}
}

\tikzset{vertices3/.pic={

		\tikzset{
			font = \small,
			styleArrow/.style={blue,postaction={decorate},decoration={markings,mark=at position 0.6 with {\arrow{Stealth[scale=1.5]}}},shorten <=\offset,shorten >=\offset},
			styleLine/.style={blue},
                        }
		\tikzmath{\offset=0.25cm;}

\tikzset{scale=#1}
		
		\draw[teal](0,4) pic{cross=#1};
		\draw[magenta](0,2) pic[rotate=45]{cross=#1};
		\draw(-2,2) pic[teal]{cross=#1};
		\draw(-2,4) pic[magenta,rotate=45]{cross=#1};

	}
}

\tikzset{arrows3/.pic={
		
		\begin{scope}[yshift=2cm]
			
			\tikzset{
				scale=0.5, font = \small,
				styleArrow/.style={blue,postaction={decorate},decoration={markings,mark=at position 0.6 with {\arrow{Stealth[scale=1.5]}}},shorten <=\offset,shorten >=\offset},
				styleLine/.style={blue},
                                yslant=-2}
			\tikzmath{\offset=0.25cm;}
			
			\draw[styleArrow](0,0)..controls(-0.2,2)..(-0,4);
			\draw[styleArrow](0,4)--(-4,0);
			\draw[styleArrow](-4,0)..controls(-2,-0.3)..(0,0);
			
			\draw[styleArrow](-4,-8)--(0,0);
			\draw[styleArrow](-4,-4)..controls(-3.8,-6)..(-4,-8);
			\draw[styleArrow] (0,0)--(-4,-4);
			
		\end{scope}
	}
}

\begin{figure}[h!]
\begin{center}
\begin{tikzpicture}

\draw (0,-2) pic[scale=0.8]{dimer1};
\draw (0,-2) pic[scale=0.8]{arrows1};
\draw[ultra thick] (0,-2) pic{vertices1=0.8};

\draw (4,-2) pic[scale=0.8]{dimer2};
\draw (4,-2) pic[scale=0.8]{arrows2};
\draw[ultra thick] (4,-2) pic{vertices2=0.8};

\draw (10,-3.6) pic[scale=0.8]{dimer3};
\draw (10,-3.6) pic[scale=0.8]{arrows3};

\draw[ultra thick] (10,-3.6) pic{vertices3=0.8};

\end{tikzpicture}
\end{center}

\caption{Building blocks for the bipartite graph: type 0, type $\mathrm{I}$ and type $-\mathrm{I}$.}
\label{fig:blocks1}
\end{figure}
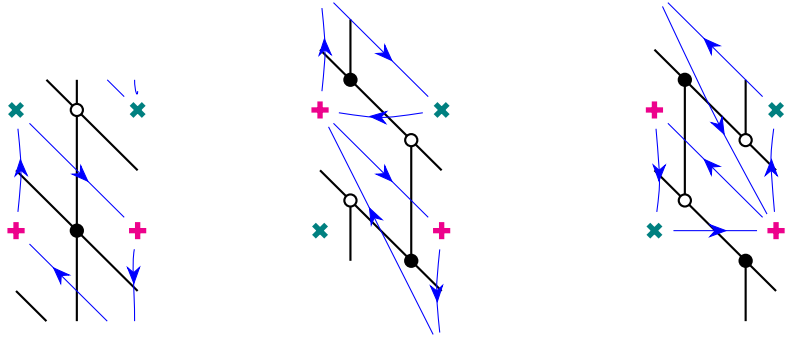

 Consider now Fig.~\ref{fig:blocks2} with the building blocks for the Poisson quiver, taken from Fig.~\ref{fig:blocks1}. It will be convenient below to draw vertices in the lattice $\mathbb{Z}^2\subset \mathbb{R}^2$, such that four vertices of block 0 form a square, while four vertices for the blocks $\mathrm{I}$ and $-\mathrm{I}$ in Fig.~\ref{fig:blocks2} form  parallelograms. Then sequence of $N$ blocks on a torus that we need belongs to a strip of (horizontal) size $N$, and we extend it to infinite $N$-periodic sequence on a whole plane. Both sets of ``$+$'' (top) and ``$\times$'' (bottom) vertices of a quiver form a path with the steps $(1,0)$, $(1,1)$ and $(1,-1)$, sometimes called grand Motzkin path. Since $N_0(1,0)+N_1(1,1)+N_{-1}(1,-1)=(N,k)$ (in notations of the Lemma~\ref{Lem:YNkThurst}), we see that the quiver corresponding to $Y^{N,k}$ polygon should be periodic with the period, given by vector $(N,k)$.
\tikzset{crossFilled/.pic={
    \tikzset{scale=0.1*#1}
		\fill[white] circle[radius=2];
		\draw[line width=2.5*#1](-1,-1)--(1,1)
		(-1,1)--(1,-1);
	}
}

\tikzset{quiver1/.pic={
		\tikzset{styleArrow/.style={blue,postaction={decorate},decoration={markings,mark=at position 0.5 with {\arrow{Stealth[scale=1.5*#1]}}}}}

                \tikzset{scale=#1}

		\draw [styleArrow,bend left] (0,2) to (0,0);
		\draw [styleArrow,bend right] (2,2) to (2,0);
		\draw [styleArrow] (0,0) to (2,2);
		\draw [styleArrow] (2,0) to (0,2);

		\draw (0,0) pic[teal]{crossFilled=#1}
		(2,0) pic[teal]{crossFilled=#1}
		(2,2) pic[rotate=45,magenta,pic actions]{crossFilled=#1}
		(0,2) pic[rotate=45,magenta,pic actions]{crossFilled=#1};
	}
}

\tikzset{quiver2/.pic={
		\tikzset{styleArrow/.style={blue,postaction={decorate},decoration={markings,mark=at position 0.5 with {\arrow{Stealth[scale=1.5*#1]}}}}}
		
                \tikzset{scale=#1}		

		\draw [styleArrow,bend left] (0,2) to (0,0);
		\draw [styleArrow,bend right] (2,4) to (2,2);
		\draw [styleArrow] (0,0) to (2,2);
		\draw [styleArrow,bend left] (2,2) to (0,2);
		\draw [styleArrow,bend right] (2,2) to (0,2);
		\draw [styleArrow] (0,2) to (2,4);

		\draw (0,0) pic[teal]{crossFilled=#1}
		(2,2) pic[teal]{crossFilled=#1}
		(2,4) pic[rotate=45,magenta,pic actions]{crossFilled=#1}
		(0,2) pic[rotate=45,magenta,pic actions]{crossFilled=#1};
		
	}
}

\tikzset{quiver3/.pic={
		\tikzset{styleArrow/.style={blue,postaction={decorate},decoration={markings,mark=at position 0.5 with {\arrow{Stealth[scale=1.5*#1]}}}}}
		
                \tikzset{scale=#1}
		
		\draw [styleArrow,bend left] (0,4) to (0,2);
		\draw [styleArrow,bend right] (2,2) to (2,0);
		\draw [styleArrow] (2,0) to (0,2);
		\draw [styleArrow,bend left] (0,2) to (2,2);
		\draw [styleArrow,bend right] (0,2) to (2,2);
		\draw [styleArrow] (2,2) to (0,4);
		\draw (0,2) pic[teal]{crossFilled=#1}
		(2,0) pic[teal]{crossFilled=#1}
		(2,2) pic[rotate=45,magenta]{crossFilled=#1}
		(0,4) pic[rotate=45,magenta]{crossFilled=#1};
		
	}
}

\begin{figure}[h!]
	\begin{center}
		
		\begin{tikzpicture}
		    \draw (0,0) pic{quiver1=0.8};
		    \draw (5,0) pic{quiver2=0.8};
		    \draw (10,0) pic{quiver3=0.8};
		\end{tikzpicture}
		
	\end{center}
	\caption{Building blocks for quiver: type 0, type $\mathrm{I}$ and type $-\mathrm{I}$.}
	\label{fig:blocks2}
\end{figure}
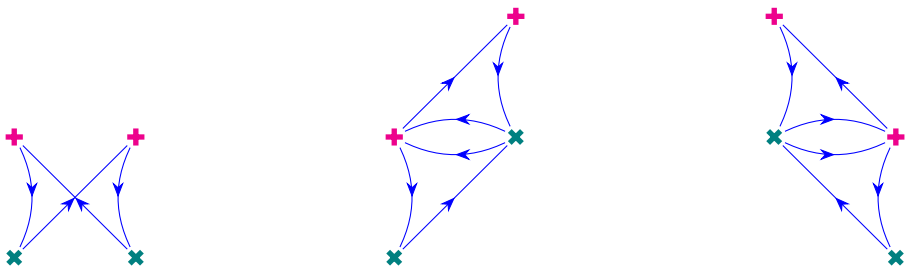

Denote by $\mathcal{Q}^{N,k}$ the quiver corresponding to $Y^{N,k}$. This quiver is not unique and depends on the concrete choice of the numbers $N_0, N_1, N_{-1}$ and order of the blocks. Our construction can be summarizes in the following:
\begin{lemma}\label{Lem:YNk:quiver}
	The vertices of the quiver $\mathcal{Q}^{N,k}$  form a polyline strip on the lattice $\mathbb{Z}^2$ with heigth $1$ and period $(N,k)$. This strip consist of the blocks presented at Fig.~\ref{fig:blocks2}.
\end{lemma}

Note that the same bipartite graphs on torus and quivers, corresponding to $Y^{N,k}$ polygons, were obtained in \cite{Hanany1} by different method.

\begin{Remark}
In the approach of \cite{FM:2014} the Poisson $X$-cluster varieties are realized as double Bruhat cells in the loop group $\widehat{SL}(M)$. The corresponding cells are labeled by elements of co-extedned  double affine Weyl group $(\mathbb{Z}/M\mathbb{Z})\ltimes \widehat{W}_M\times \widehat{W}_M$. It appears that $X$-cluster Poisson varieties, corresponding to quivers given in Lemma~\ref{Lem:YNk:quiver}, can be realized in the loop group $\widehat{SL}(2)$, so below we consider only the case $M=2$. This goes back to Faddeev-Takhtajan approach to the Toda chains~\cite{FT}. Another natural choice would be $M=N$.

As in \cite{FM:2014}, we denote by $s_0,s_1,s_{\overline 0}, s_{\overline 1}$ the generators of $\widehat{W}_2\times \widehat{W}_2$, and by  $\Uplambda$ --- the generator of $\mathbb{Z}/2\mathbb{Z}$ with relations
\be
s_is_{\overline j}=s_{\overline j}s_i,\quad \Uplambda s_i=s_{i+1}\Uplambda,\; \Uplambda s_{\overline i}=s_{\overline{i+1}}\Uplambda,\quad \Uplambda^2=s_i^2=s_{\overline i}^2=1\,.
\ee
There is a straightforward way to reconstruct $X$-cluster Poisson variety from an element of the co-extended  double affine Weyl group. The vertices of corresponding quiver are drawn on two parallel lines on cylinder. To each generator $s_i$ one assigns a triangle as at Fig.~\ref{fig:Weyl}, while the generator $\Uplambda$ corresponds to the entanglement of the lines.
\begin{figure}[h!]
	\begin{center}
		{\begin{tabular}{c c c c c}
				\begin{tikzpicture}[scale=1, font = \small]
				\draw[fill] (0,cos{30}) circle (1pt);
				\draw[fill] (-0.5,0) circle (1pt);
				\draw[fill] (0.5,0) circle (1pt);
				\path 	
				(-0.5,0) edge[mid arrow] (0.5,0)
				(0.5,0) edge[mid arrow] (0,cos{30})
				(0,cos{30}) edge[mid arrow] (-0.5,0);
				\draw[ultra thin] (-0.7,0) -- (0.7,0);
				\draw[ultra thin] (-0.7,cos{30}) -- (0.7,cos{30});	
				\end{tikzpicture}
				&
				\begin{tikzpicture}[scale=1, font = \small]
				\draw[fill] (0,cos{30}) circle (1pt);
				\draw[fill] (-0.5,0) circle (1pt);
				\draw[fill] (0.5,0) circle (1pt);
				\path 	
				(0.5,0) edge[mid arrow] (-0.5,0)
				(-0.5,0) edge[mid arrow] (0,cos{30})
				(0,cos{30}) edge[mid arrow] (0.5,0);
				\draw[ultra thin] (-0.7,0) -- (0.7,0);
				\draw[ultra thin] (-0.7,cos{30}) -- (0.7,cos{30});	
				\end{tikzpicture}
				&
				\begin{tikzpicture}[scale=1, font = \small]
				\draw[fill] (0,0) circle (1pt);
				\draw[fill] (-0.5,cos{30}) circle (1pt);
				\draw[fill] (0.5,cos{30}) circle (1pt);
				\path 	
				(-0.5,cos{30}) edge[mid arrow] (0.5,cos{30})
				(0.5,cos{30}) edge[mid arrow] (0,0)
				(0,0) edge[mid arrow] (-0.5,cos{30});
				\draw[ultra thin] (-0.7,0) -- (0.7,0);
				\draw[ultra thin] (-0.7,cos{30}) -- (0.7,cos{30});	
				\end{tikzpicture}
				&
				\begin{tikzpicture}[scale=1, font = \small]
				\draw[fill] (0,0) circle (1pt);
				\draw[fill] (-0.5,cos{30}) circle (1pt);
				\draw[fill] (0.5,cos{30}) circle (1pt);
				\path 	
				(0.5,cos{30}) edge[mid arrow] (-0.5,cos{30})
				(-0.5,cos{30}) edge[mid arrow] (0,0)
				(0,0) edge[mid arrow] (0.5,cos{30});
				\draw[ultra thin] (-0.7,0) -- (0.7,0);
				\draw[ultra thin] (-0.7,cos{30}) -- (0.7,cos{30});	
				\end{tikzpicture}
 		        &
\begin{tikzpicture}[scale=1, font = \small]
%				\draw[fill] (0,0) circle (1pt);
%				\draw[fill] (0,cos{30}) circle (1pt);
%				\draw[fill] (1,0) circle (1pt);
%				\draw[fill] (1,cos{30}) circle (1pt);		
%				\path 	
%				(0,cos{30}) edge[mid arrow,bend left=20] (1,0)
%(0,0) edge[mid arrow,bend right=20] (1,cos{30});
				\draw[ultra thin] (-0.2,0) to[in=180,out=0] (1.2,cos{30});
				\draw[ultra thin] (-0.2,cos{30}) to[in=180,out=0] (1.2,0);	
				\end{tikzpicture}
				\\
				$s_0$ & $s_{\bar{0}}$  & $s_1$ & $s_{\bar{1}}$ & $\Uplambda$
		\end{tabular}}
	\end{center}
	\begin{center}
		{\begin{tabular}{c c c c c c}
				\begin{tikzpicture}[scale=1, font = \small]
				\draw[fill] (0,0) circle (1pt);
				\draw[fill] (0,1) circle (1pt);
				\draw[fill] (1,0) circle (1pt);
				\draw[fill] (1,1) circle (1pt);		
				\path 	
				(0,0) edge[mid arrow] (1,0)
				(1,0) edge[mid arrow] (1,1)
				(1,1) edge[mid arrow] (0,1)
				(0,1) edge[mid arrow] (0,0);
				\draw[ultra thin] (-0.2,0) -- (1.2,0);
				\draw[ultra thin] (-0.2,1) -- (1.2,1);	
				\end{tikzpicture}
				&
				\begin{tikzpicture}[scale=1, font = \small]
				\draw[fill] (0,0) circle (1pt);
				\draw[fill] (0,1) circle (1pt);
				\draw[fill] (1,0) circle (1pt);
				\draw[fill] (1,1) circle (1pt);		
				\path 	
				(1,0) edge[mid arrow] (0,0)
				(1,1) edge[mid arrow] (1,0)
				(0,1) edge[mid arrow] (1,1)
				(0,0) edge[mid arrow] (0,1);
				\draw[ultra thin] (-0.2,0) -- (1.2,0);
				\draw[ultra thin] (-0.2,1) -- (1.2,1);	
				\end{tikzpicture}
				&
				\begin{tikzpicture}[scale=1, font = \small]
				\draw[fill] (0,0) circle (1pt);
				\draw[fill] (0,1) circle (1pt);
				\draw[fill] (1,0) circle (1pt);
				\draw[fill] (1,1) circle (1pt);		
				\path 	
				(0,0) edge[mid arrow] (1,0)
				(1,1) edge[mid arrow] (1,0)
				(0,1) edge[mid arrow] (1,1) edge[mid arrow] (0,0)
				(1,0) edge[mid arrow,bend right=20] (0,1) edge[mid arrow,bend left=20] (0,1);
				\draw[ultra thin] (-0.2,0) -- (1.2,0);
				\draw[ultra thin] (-0.2,1) -- (1.2,1);	
				\end{tikzpicture}
				&
				\begin{tikzpicture}[scale=1, font = \small]
				\draw[fill] (0,0) circle (1pt);
				\draw[fill] (0,1) circle (1pt);
				\draw[fill] (1,0) circle (1pt);
				\draw[fill] (1,1) circle (1pt);		
				\path 	
				(1,0) edge[mid arrow] (0,0) edge[mid arrow] (1,1)
				(0,0) edge[mid arrow] (0,1)
				(1,1) edge[mid arrow] (0,1)
				(0,1) edge[mid arrow,bend right=20] (1,0) edge[mid arrow,bend left=20] (1,0);
				\draw[ultra thin] (-0.2,0) -- (1.2,0);
				\draw[ultra thin] (-0.2,1) -- (1.2,1);	
				\end{tikzpicture}
				&
				\begin{tikzpicture}[scale=1, font = \small]
				\draw[fill] (0,0) circle (1pt);
				\draw[fill] (0,1) circle (1pt);
				\draw[fill] (1,0) circle (1pt);
				\draw[fill] (1,1) circle (1pt);		
				\path 	
				(0,0) edge[mid arrow] (1,0) edge[mid arrow] (0,1)
				(1,0) edge[mid arrow] (1,1)
				(0,1) edge[mid arrow] (1,1)
				(1,1) edge[mid arrow,bend right=20] (0,0) edge[mid arrow,bend left=20] (0,0);
				\draw[ultra thin] (-0.2,0) -- (1.2,0);
				\draw[ultra thin] (-0.2,1) -- (1.2,1);	
				\end{tikzpicture}
				&  		
				\begin{tikzpicture}[scale=1, font = \small]
				\draw[fill] (0,0) circle (1pt);
				\draw[fill] (0,1) circle (1pt);
				\draw[fill] (1,0) circle (1pt);
				\draw[fill] (1,1) circle (1pt);		
				\path 	
				(1,0) edge[mid arrow] (0,0)
				(1,1) edge[mid arrow] (0,1) edge[mid arrow] (1,0)
				(0,1) edge[mid arrow] (0,0)
				(0,0) edge[mid arrow,bend right=20] (1,1) edge[mid arrow,bend left=20] (1,1);		
				\draw[ultra thin] (-0.2,0) -- (1.2,0);
				\draw[ultra thin] (-0.2,1) -- (1.2,1);	
				\end{tikzpicture}		
				\\
				$s_0s_{\bar{1}}$ & $s_{\bar{0}}s_1$  & $s_0s_1$ & $s_{\bar{0}}s_{\bar{1}}$ & $s_1s_0$ & $s_{\bar{1}}s_{\bar{0}}$
		\end{tabular}}
	\end{center}
	\caption{Blocks for construction of quiver from the element of co-extended double affine Weyl group.}\label{fig:Weyl}
\end{figure}
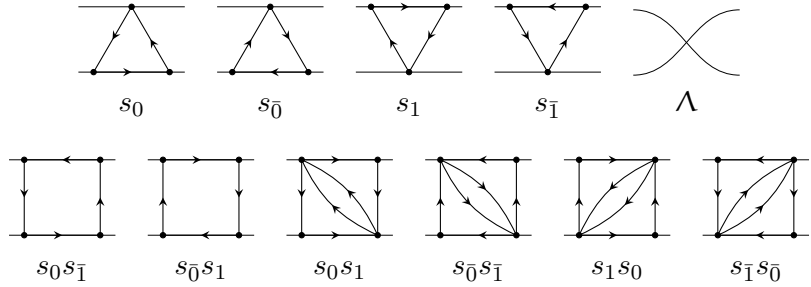

%\noindent
One finds from Fig.~\ref{fig:Weyl} that the products $s_0s_{\overline 1}$ and $s_{\overline 0}s_1$ correspond to the type 0 block  by clear gluing rules, the products $s_0s_{1}$ and $s_{1}s_0$ correspond to the block~$\mathrm{I}$, while the products $s_{\overline 0}s_{\overline 1}$ and $s_{\overline 0}s_{\overline 1}$ correspond to the block~$-\mathrm{I}$ (see Fig.~\ref{fig:blocks2}). Therefore the reduced word $u\in(\mathbb{Z}/2\mathbb{Z})\ltimes \widehat{W}_2\times \widehat{W}_2$, corresponding to $\mathcal{Q}^{N,k}$, consists of product of $2N$ reflections $s\in\widehat{W}_2$, and, in the case of odd $N+k$, an extra $\Uplambda$. The total number of the generators $s_0, s_{\overline 0}$ equals to the number of generators $s_1, s_{\overline 1}$  and is given by $N$, while the total number of the generators $s_{0}, s_{1}$ equals to $N+k$. For example, $k=0$ relativistic Toda chain~\footnote{Alternatively it can be always realized on a Poisson submanifold in affine $\widehat{SL(N)}$ being defined by the ``double-Coxeter'' word $u=s_0s_{\overline 0}s_1s_{\overline 1}\ldots s_{N-1}s_{\overline{N-1}}\in \widehat{W}_N\times \widehat{W}_N$, see e.g. \cite{AMJGP}.} corresponds to the word $u=(s_0s_{\overline 0}s_{1}s_{\overline{1}})^{\lfloor N/2\rfloor}(s_0s_{\overline{0}}\Uplambda)^{2\{N/2\}}\in (\mathbb{Z}/2\mathbb{Z})\ltimes\widehat{W}_2\times \widehat{W}_2$, while for the $k=N$ case one finds $u=(s_0s_1)^N\in \widehat{W}_2\times \widehat{W}_2$.

\end{Remark}

\subsection{Quiver mutations}

One has a freedom in choosing numbers $N_0, N_1, N_{-1}$ and order of blocks in the Lemma~\ref{Lem:YNkThurst}. This freedom remains in the subsequent constructions of bipartite graphs and quivers, but as proven in \cite{GK}, the resulting cluster variety $\mathcal{X}$ does not depend on this choice. Thurston diagrams, corresponding to different choices, are related by Thurston moves; corresponding bipartite graphs are related by spider moves~\footnote{also referred as ``square moves'' or ``urban renewal''}
and contractions of two-valent vertices (we recall the definition of the Thurston moves at Fig.~\ref{fig:Thurston:moves}, and the spider moves and contractions of bipartite graphs at Fig.~\ref{fig:spider:moves}), whereas the corresponding quivers are related by mutations.

\begin{Remark}
Since there are three types of blocks, one gets $3^N$ quivers (for a given $N$), to be compared with $3^N$ open relativistic Toda chains that appeared in \cite{GSV:2009},  see \cite[Sect. 11.7]{FT17} and \cite{GT:2018} for the quantum case. Note that in the case of open Toda chains their integrals of motion are in fact equivalent by certain rational transformations (the B\"acklund-Darboux transformations), see \cite[Theorem 6.1]{GSV:2009}. In terms of $X$-cluster varieties these rational transformations are compositions of mutations. In contrast to the open Toda cases, the $Y^{N,k}$ integrable systems correspond to affine Toda chains, when there is an additional parameter $k$, preserved under mutations. The corresponding integrals of motion are not equivalent, for example in the expressions for two $N=3$ and $k=1$ Hamiltonians
\be
H_1^{(k=1)} = {1+y_1+y_1x_1+y_1x_1y_2+y_1x_1y_2x_2+z^{-2/3}y_1^{1/3}x_1^{4/3}y_2^{2/3}x_2^{2/3}(1+x_2)\over (x_1y_1)^{2/3}(x_2y_2)^{1/3}}\,,
\\
H_2^{(k=1)} = {1+y_2+y_2x_2+y_2x_2y_1+y_2x_2y_1x_1+z^{-2/3}y_1^{1/3}x_1^{4/3}y_2^{2/3}x_2^{5/3}\over (x_1y_1)^{1/3}(x_2y_2)^{2/3}}
\ee
one can easily find the difference in the ``affine'' ($z$-dependent) terms with the usual $N=3$ and $k=0$ relativistic Toda expressions
\be
H_1^{(k=0)} = {1+y_1+y_1x_1+y_1x_1y_2+y_1x_1y_2x_2+z^{-1}x_1x_2\over (x_1y_1)^{2/3}(x_2y_2)^{1/3}},
\\
H_2^{(k=0)} = {1+y_2+y_2x_2+y_2x_2y_1+y_2x_2y_1x_1+z^{-1}x_1x_2\over (x_1y_1)^{1/3}(x_2y_2)^{2/3}}
\ee
Both pairs are in the involution $\{H_1^{(k)},H_2^{(k)}\}=0$ w.r.t. the standard Poisson bracket,
where the only nontrivial relations are $\{ y_i, x_j\} = C_{ij}y_ix_j$ for $i,j=1,\ldots,2$. Here $C_{ij}$ is the Cartan matrix for $SL(3)$,
the only nontrivial Casimir is denoted by $z$.
\end{Remark}

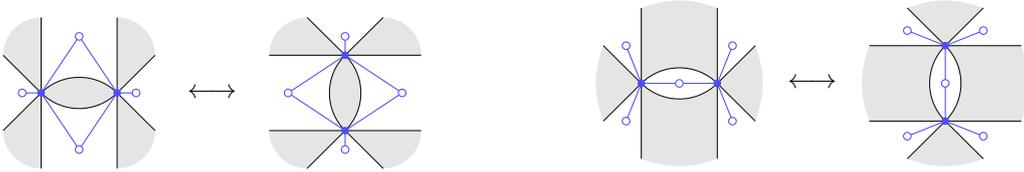
\begin{figure}[ht]
	\begin{center}
\begin{tikzpicture}[scale=0.5, font = \small,				styleFill/.style={fill,gray,opacity=0.2},
styleArrow/.style={postaction={decorate},decoration={markings,mark=at position 0.5 with {\arrow{Stealth[scale=1.5]}}}},
styleDimer/.style={blue!70}]					
	\draw (0,0) -- (0,4);	
	\draw (2,0) -- (2,4);
	\draw (-1,1) -- (0,2) to[out=45,in=135] (2,2) -- (3,1);
	\draw (3,3) -- (2,2) to[out=225,in=315] (0,2) -- (-1,3);

	\path[styleFill] (0,0)--(0,2) --  (-1,1)  to  [out=270,in=180] (0,0);
	\path[styleFill] (0,4)--(0,2) --  (-1,3) to  [out=90,in=180] (0,4);
	\path[styleFill] (0,2) to[out=45,in=135]  (2,2) to[out=225,in=315] (0,2);
	\path[styleFill] (2,0)--(2,2) --  (3,1)  to  [out=270,in=00] (2,0);
	\path[styleFill] (2,4)--(2,2) --  (3,3) to  [out=90,in=0] (2,4);

	\node at (4.5,2) {$\longleftrightarrow$};

	\draw (6,1) -- (10,1); \draw (6,3)--(10,3);	
	\draw (7,0) -- (8,1) to[out=45,in=315] (8,3) -- (7,4);
	\draw (9,0) -- (8,1) to[out=135,in=225] (8,3) -- (9,4);	

	\path[styleFill] (6,1) -- (8,1) -- (7,0) to[out=180,in=270] (6,1); 	
	\path[styleFill] (6,3) -- (8,3) -- (7,4) to[out=180,in=90] (6,3);
	\path[styleFill] (10,1) -- (8,1) -- (9,0) to[out=0,in=270] (10,1);
	\path[styleFill] (10,3) -- (8,3) -- (9,4) to[out=0,in=90] (10,3);	
	\path[styleFill] (8,1) to[out=45,in=315]  (8,3) to[out=225,in=135] (8,1);		

        \fill[styleDimer](0,2) circle(0.1)
        (2,2) circle(0.1);
        \draw[styleDimer](0,2)--(1,3.5)--(2,2)--(1,0.5)--cycle (-0.5,2)--(0,2) (2,2)--(2.5,2);
        \draw[styleDimer,fill=white]
        (1,3.5) circle(0.1)
        (1,0.5) circle(0.1)
        (-0.5,2) circle(0.1)
        (2.5,2) circle(0.1);

\begin{scope}[xshift=10cm,yshift=1cm,rotate=90]
         \fill[styleDimer](0,2) circle(0.1)
        (2,2) circle(0.1);
        \draw[styleDimer](0,2)--(1,3.5)--(2,2)--(1,0.5)--cycle (-0.5,2)--(0,2) (2,2)--(2.5,2);
        \draw[styleDimer,fill=white]
        (1,3.5) circle(0.1)
        (1,0.5) circle(0.1)
        (-0.5,2) circle(0.1)
        (2.5,2) circle(0.1);
\end{scope}

\end{tikzpicture}		
\hspace{2cm}
\begin{tikzpicture}[scale=0.5, font = \small,				styleFill/.style={fill,gray,opacity=0.2},
styleArrow/.style={postaction={decorate},decoration={markings,mark=at position 0.5 with {\arrow{Stealth[scale=1.5]}}}},
styleDimer/.style={blue!70}]				
	\draw (0,0) -- (0,4); \draw (2,0) -- (2,4);
	\draw (-1,1) -- (0,2) to[out=45,in=135] (2,2) -- (3,1);
	\draw (3,3) -- (2,2) to[out=225,in=315] (0,2) -- (-1,3);

	\path[styleFill] (-1,1) --(0,2) -- (-1,3) to[out=250,in=110] (-1,1);
	\path[styleFill] (3,1) --(2,2) -- (3,3) to[out=290,in=70] (3,1);	
	\path[styleFill] (0,0) -- (0,2) to[in=225,out=315] (2,2) -- (2,0) to[out=200,in=340] (0,0);
	\path[styleFill] (0,4) -- (0,2) to[out=45,in=135] (2,2) -- (2,4) to[out=160,in=20] (0,4);
	
	\node at (4.5,2) {$\longleftrightarrow$};

	\draw (6,1) -- (10,1); \draw (6,3)--(10,3);	
	\draw (7,0) -- (8,1) to[out=45,in=315] (8,3) -- (7,4);
	\draw (9,0) -- (8,1) to[out=135,in=225] (8,3) -- (9,4);	

	\path[styleFill] (6,1) -- (8,1) to[out=135,in=225] 	(8,3) --(6,3) to[out=250,in=110] (6,1);
	\path[styleFill] (10,1) -- (8,1) to[out=45,in=315] 	(8,3) --(10,3) to[out=290,in=70] (10,1);
	\path[styleFill] (7,0) -- (8,1) -- (9,0) to[out=200,in=340] 	(7,0);
	\path[styleFill] (7,4) -- (8,3) -- (9,4) to[out=160,in=20] 	(7,4);

        \fill[styleDimer](0,2) circle(0.1)
        (2,2) circle(0.1);
        \draw[styleDimer](0,2)--(2,2) (2.4,3)--(2,2)--(2.4,1) (-0.4,3)--(0,2)--(-0.4,1);
        \draw[styleDimer,fill=white]
        (1,2) circle(0.1)
        (-0.4,3) circle(0.1)
        (-0.4,1) circle(0.1)
        (2.4,3) circle(0.1)
        (2.4,1) circle(0.1);

\begin{scope}[xshift=10cm,yshift=1cm,rotate=90]
        \fill[styleDimer](0,2) circle(0.1)
        (2,2) circle(0.1);
        \draw[styleDimer](0,2)--(2,2) (2.4,3)--(2,2)--(2.4,1) (-0.4,3)--(0,2)--(-0.4,1);
        \draw[styleDimer,fill=white]
        (1,2) circle(0.1)
        (-0.4,3) circle(0.1)
        (-0.4,1) circle(0.1)
        (2.4,3) circle(0.1)
        (2.4,1) circle(0.1);
\end{scope}

\end{tikzpicture}		
	\end{center}
\caption{Grey and white Thurston moves.}\label{fig:Thurston:moves}
\end{figure}

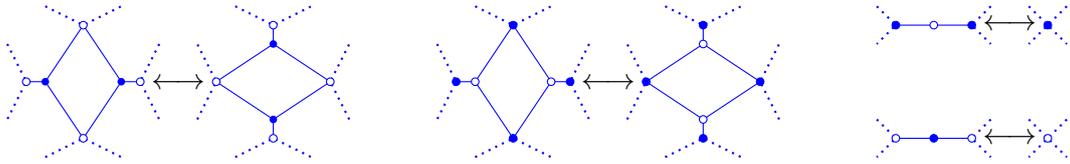
\begin{figure}[ht]
\begin{center}
\begin{tikzpicture}[styleDimer/.style={blue},scale=0.5]				

        \draw[dotted,thick,styleDimer](-1,3)--(-0.5,2)--(-1,1)
        (3,3)--(2.5,2)--(3,1)
        (0,4)--(1,3.5)--(2,4)
        (0,0)--(1,0.5)--(2,0);

        \fill[styleDimer](0,2) circle(0.1)
        (2,2) circle(0.1);
        \draw[styleDimer](0,2)--(1,3.5)--(2,2)--(1,0.5)--cycle (-0.5,2)--(0,2) (2,2)--(2.5,2);
        \draw[styleDimer,fill=white]
        (1,3.5) circle(0.1)
        (1,0.5) circle(0.1)
        (-0.5,2) circle(0.1)
        (2.5,2) circle(0.1);

\node at(3.5,1.95){$\longleftrightarrow$};

\begin{scope}[xshift=8cm,yshift=1cm,rotate=90]

        \draw[dotted,thick,styleDimer](-1,3)--(-0.5,2)--(-1,1)
        (3,3)--(2.5,2)--(3,1)
        (0,4)--(1,3.5)--(2,4)
        (0,0)--(1,0.5)--(2,0);

         \fill[styleDimer](0,2) circle(0.1)
        (2,2) circle(0.1);
        \draw[styleDimer](0,2)--(1,3.5)--(2,2)--(1,0.5)--cycle (-0.5,2)--(0,2) (2,2)--(2.5,2);
        \draw[styleDimer,fill=white]
        (1,3.5) circle(0.1)
        (1,0.5) circle(0.1)
        (-0.5,2) circle(0.1)
        (2.5,2) circle(0.1);
\end{scope}

\end{tikzpicture}\hspace{1cm}
\begin{tikzpicture}[styleDimer/.style={blue},scale=0.5]				

        \draw[dotted,thick,styleDimer](-1,3)--(-0.5,2)--(-1,1)
        (3,3)--(2.5,2)--(3,1)
        (0,4)--(1,3.5)--(2,4)
        (0,0)--(1,0.5)--(2,0);

        \draw[styleDimer](0,2)--(1,3.5)--(2,2)--(1,0.5)--cycle (-0.5,2)--(0,2) (2,2)--(2.5,2);
        \draw[styleDimer,fill=blue]
        (1,3.5) circle(0.1)
        (1,0.5) circle(0.1)
        (-0.5,2) circle(0.1)
        (2.5,2) circle(0.1);
         \draw[styleDimer,fill=white](0,2) circle(0.1)
        (2,2) circle(0.1);

\node at(3.5,1.95){$\longleftrightarrow$};

\begin{scope}[xshift=8cm,yshift=1cm,rotate=90]

        \draw[dotted,thick,styleDimer](-1,3)--(-0.5,2)--(-1,1)
        (3,3)--(2.5,2)--(3,1)
        (0,4)--(1,3.5)--(2,4)
        (0,0)--(1,0.5)--(2,0);

        \draw[styleDimer](0,2)--(1,3.5)--(2,2)--(1,0.5)--cycle (-0.5,2)--(0,2) (2,2)--(2.5,2);
        \draw[styleDimer,fill=blue]
        (1,3.5) circle(0.1)
        (1,0.5) circle(0.1)
        (-0.5,2) circle(0.1)
        (2.5,2) circle(0.1);
         \draw[styleDimer,fill=white](0,2) circle(0.1)
        (2,2) circle(0.1);

\end{scope}
\end{tikzpicture}
\hspace{1cm}
\begin{tikzpicture}[styleDimer/.style={blue},scale=0.5]				

        \draw[dotted,thick,styleDimer](2.5,0.5)--(2,0)--(2.5,-0.5)
        (-0.5,0.5)--(0,0)--(-0.5,-0.5)
        (4.5,0.5)--(4,0)--(4.5,-0.5)
        (3.5,0.5)--(4,0)--(3.5,-0.5);

        \draw[styleDimer] (0,0)--(2,0);
        \draw[styleDimer,fill=white]  (0,0) circle(0.1) (2,0) circle(0.1) (4,0) circle(0.1);
        \draw[styleDimer,fill=blue] (1,0) circle(0.1);
        \node at(3,0){$\longleftrightarrow$};

\begin{scope}[yshift=3cm]

        \draw[dotted,thick,styleDimer](2.5,0.5)--(2,0)--(2.5,-0.5)
        (-0.5,0.5)--(0,0)--(-0.5,-0.5)
        (4.5,0.5)--(4,0)--(4.5,-0.5)
        (3.5,0.5)--(4,0)--(3.5,-0.5);

        \draw[styleDimer] (0,0)--(2,0);
        \draw[styleDimer,fill=blue]  (0,0) circle(0.1) (2,0) circle(0.1) (4,0) circle(0.1);
        \draw[styleDimer,fill=white] (1,0) circle(0.1);
        \node at(3,0){$\longleftrightarrow$};

\end{scope}

\end{tikzpicture}

\end{center}

\caption{Spider moves and contraction of the edges. }\label{fig:spider:moves}
\end{figure}

Mutation of the quiver $\mathcal{Q}$ is defined as follows. Let $\epsilon_{ij}$ be the number of arrows from $i$-th to $j$-th vertex ($\epsilon_{ji} = -\epsilon_{ij}$) of $\mathcal{Q}$. The mutation at the vertex $j$ acts as
\be
\label{mute}
\mu_j\colon \epsilon_{ik}\mapsto -\epsilon_{ik}, \text{ if $i=j$ or $k=j$},\quad  \epsilon_{ik}\mapsto\epsilon_{ik}+\frac{\epsilon_{ij}|\epsilon_{jk}|+\epsilon_{jk}|\epsilon_{ij}|}{2}\ \  \text{ otherwise}.
\ee

\begin{lemma}\label{Lem:mut}
Mutations can transform the sequences of blocks as follows: $(0,0)\leftrightarrow (-\mathrm{I},\mathrm{I})$;\; $(\mathrm{I},-\mathrm{I})\leftrightarrow (0,0)$;\; $(0,-\mathrm{I}) \leftrightarrow (-\mathrm{I},0)$;\; $(\mathrm{I},0) \leftrightarrow (0,\mathrm{I})$.
\end{lemma}
\begin{proof}
We just present the corresponding transformations of quivers at Fig.~\ref{fig:mut:quiv}, where the circles mark the vertices, where mutation is performed~\footnote{We also notice that such quivers on the square lattice and their mutations appeared already in \cite[page 5]{DiFrancesco} during the study of similar integrable system.}. Note that after mutation we change the marks of the vertices such that label $+$ is above $\times$ on each vertical line.
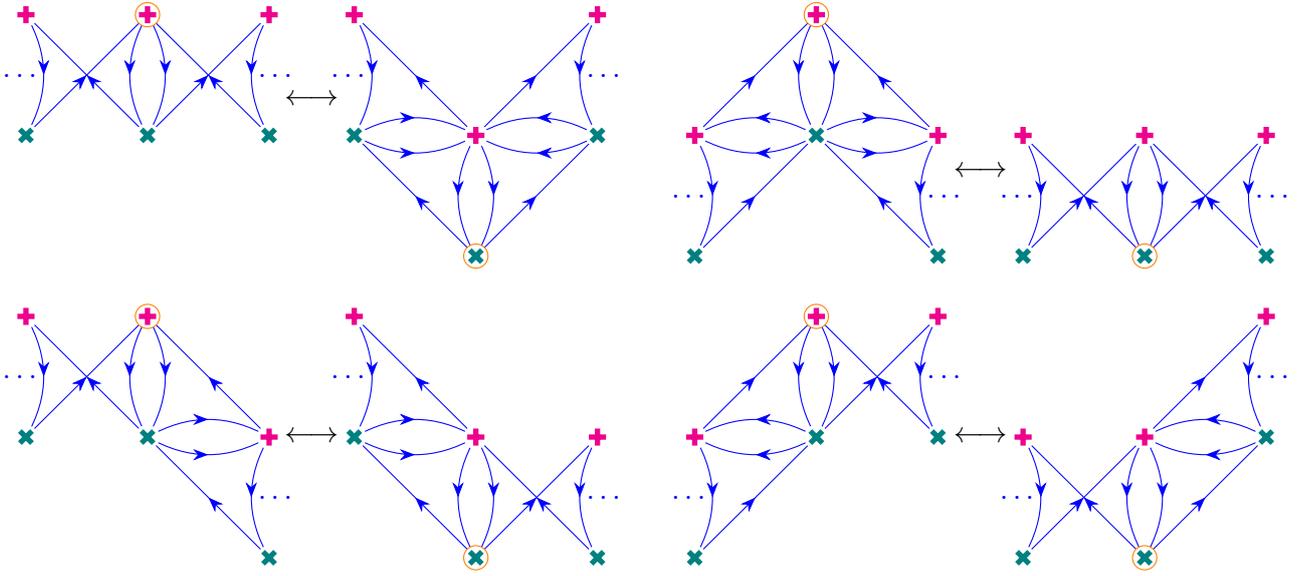
\begin{figure}[h!]
\begin{center}
\begin{tikzpicture}
\tikzmath{\scale=0.8;\radius=0.2;}
\draw[scale=\scale](0,0) pic{quiver1=\scale} (2,0) pic{quiver1=\scale} (4.7,0.6) node {$\longleftrightarrow$}
(5.4,-2) pic{quiver3=\scale} (7.4,-2) pic{quiver2=\scale}
(-0.1,1) node[blue]{$\ldots$} (4.1,1) node[blue]{$\ldots$} (5.3,1) node[blue]{$\ldots$} (9.5,1) node[blue]{$\ldots$};
\draw[orange,scale=\scale] (0+2,0+2) circle[radius=\radius] (5.4+2,-2) circle[radius=\radius];

\draw[scale=\scale,xshift=11cm](0,-2) pic{quiver2=\scale} (2,-2) pic{quiver3=\scale} (4.7,-0.6) node {$\longleftrightarrow$}
(5.4,-2) pic{quiver1=\scale} (7.4,-2) pic{quiver1=\scale}
(-0.1,-1) node[blue]{$\ldots$} (4.1,-1) node[blue]{$\ldots$} (5.3,-1) node[blue]{$\ldots$} (9.5,-1) node[blue]{$\ldots$}; \draw[orange,scale=\scale,xshift=11cm] (0+2,-2+4) circle[radius=\radius] (5.4+2,-2) circle[radius=\radius];

\draw[scale=\scale,yshift=-5cm](0,0) pic{quiver1=\scale} (2,-2) pic{quiver3=\scale} (4.7,0) node {$\longleftrightarrow$}
(5.4,-2) pic{quiver3=\scale} (7.4,-2) pic{quiver1=\scale}
(-0.1,1) node[blue]{$\ldots$} (4.1,-1) node[blue]{$\ldots$} (5.3,1) node[blue]{$\ldots$} (9.5,-1) node[blue]{$\ldots$};
\draw[orange,scale=\scale,yshift=-5cm] (0+2,-2+4) circle[radius=\radius] (5.4+2,-2) circle[radius=\radius];

\draw[scale=\scale,yshift=-5cm,xshift=11cm](0,-2) pic{quiver2=\scale} (2,0) pic{quiver1=\scale} (4.7,0) node {$\longleftrightarrow$}
(5.4,-2) pic{quiver1=\scale} (7.4,-2) pic{quiver2=\scale}
(-0.1,-1) node[blue]{$\ldots$} (4.1,1) node[blue]{$\ldots$} (5.3,-1) node[blue]{$\ldots$} (9.5,1) node[blue]{$\ldots$};
\draw[orange,scale=\scale,yshift=-5cm,xshift=11cm] (0+2,-2+4) circle[radius=\radius] (5.4+2,-2) circle[radius=\radius];

\end{tikzpicture}
\end{center}
\caption{Mutations.}\label{fig:mut:quiv}
\end{figure}
\end{proof}
Two remarks are now in order. First, note that mutations at Fig.~\ref{fig:mut:quiv} move the corresponding vertex of a quiver to another integer point on the same vertical line. Clearly, using such moves one can transform any sequence of $0, \mathrm{I}, -\mathrm{I}$ blocks to any other sequence with the same $N=N_0+N_1+N_{-1}$ and $k=N_1-N_{-1}$. In other words, one can transform grand Motzkin path of any shape to path of any other shape (but with the same $N,k$).

Second, generally speaking, not any mutation of a quiver corresponds to spider move of a bipartite graph (or Thurston move of corresponding Thurston diagram), but one can find such moves corresponding to the mutations from Fig.~\ref{fig:mut:quiv}. For example, the moves corresponding to $(0,0)\leftrightarrow (-\mathrm{I},\mathrm{I})$ and $(0,-\mathrm{I})$ $\leftrightarrow$ $(-\mathrm{I},0)$ mutations are given in Fig.~\ref{fig:moves}.

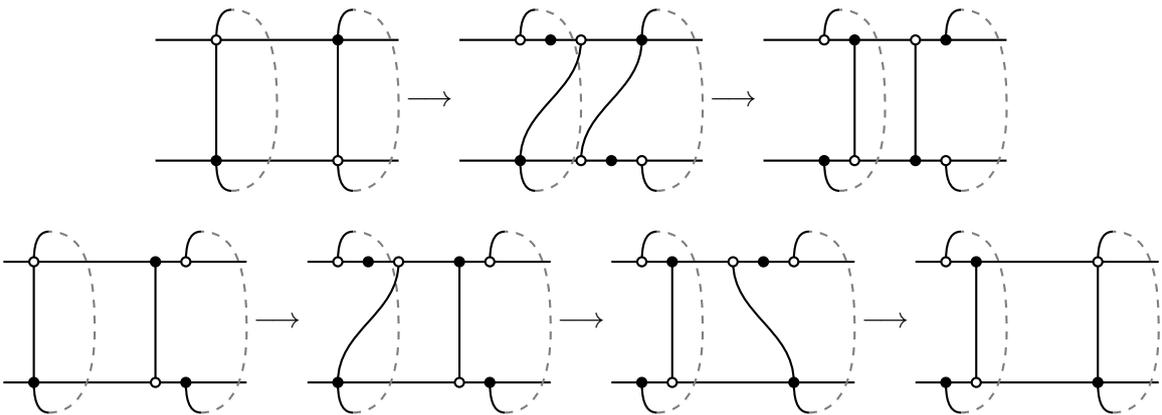
\begin{figure}[h!]
\begin{center}
\tikzset{pics/line1/.style 2 args={code={
    \draw (#1,0) to[out=90,in=-90] (#2,4);
  }}
}

\tikzset{line2/.pic={
    \draw (#1,0) to[out=-90,in=180]  (#1+0.5,-1) (#1+0.5,5) to[out=180,in=90] (#1,4) ;
    \draw[dashed,gray] (#1+0.5,-1) to[out=0,in=-90,dashed]  (#1+2,1.7) to [out=90,in=0] (#1+0.5,5);
  }
}

\begin{tikzpicture}

\tikzmath{\scale=0.4;}

\tikzset{scale=\scale,
blackCircle/.style={fill=black,thick,radius=0.15,inner sep=0},
whiteCircle/.style={fill=white,thick,radius=0.15,inner sep=0}
}

%\draw[dotted](0,0) grid(4,4);

\draw[thick]
  (-2,0)--(6,0)
  (-2,4)--(6,4)
  (0,0) pic[scale=\scale]{line2=0}
  (0,0) pic[scale=\scale]{line1=0 0}
  (0,0) pic[scale=\scale]{line2=4}
  (0,0) pic[scale=\scale]{line1=4 4}
  ;

\draw [whiteCircle] (0,4)  circle[] (4,0) circle[];
\draw [blackCircle] (0,0)  circle[] (4,4) circle[];

\node at (7,2) {\bf $\longrightarrow$};

\begin{scope}[xshift=10cm]

\draw[thick]
  (-2,0)--(6,0)
  (-2,4)--(6,4)
  (0,0) pic[scale=\scale]{line2=0}
  (0,0) pic[scale=\scale]{line2=4}
  (0,0) pic[scale=\scale]{line1=0 2}
  (0,0) pic[scale=\scale]{line1=2 4}
  ;

\draw [whiteCircle] (0,4)  circle[] (2,4) circle[] (4,0) circle[] (2,0) circle[];
\draw [blackCircle] (0,0)  circle[] (3,0) circle[] (4,4) circle[] (1,4) circle[];

\node at (7,2) {\bf $\longrightarrow$};

\end{scope}

\begin{scope}[xshift=20cm]

\draw[thick]
  (-2,0)--(6,0)
  (-2,4)--(6,4)
  (0,0) pic[scale=\scale]{line2=0}
  (0,0) pic[scale=\scale]{line2=4}
  (0,0) pic[scale=\scale]{line1=1 1}
  (0,0) pic[scale=\scale]{line1=3 3}
  ;

\draw [whiteCircle] (0,4)  circle[] (3,4) circle[] (4,0) circle[] (1,0) circle[];
\draw [blackCircle] (0,0)  circle[] (3,0) circle[] (4,4) circle[] (1,4) circle[];

%\node at (7,2) {\bf $\longrightarrow$};

\end{scope}

\end{tikzpicture}

\

\noindent
\begin{tikzpicture}

\tikzmath{\scale=0.4;}

\tikzset{scale=\scale,
blackCircle/.style={fill=black,thick,radius=0.15,inner sep=0},
whiteCircle/.style={fill=white,thick,radius=0.15,inner sep=0}
}

%\draw[dotted](0,0) grid(4,4);

\draw[thick]
  (-1,0)--(7,0)
  (-1,4)--(7,4)
  (0,0) pic[scale=\scale]{line2=0}
  (0,0) pic[scale=\scale]{line1=0 0}
  (0,0) pic[scale=\scale]{line2=5}
  (0,0) pic[scale=\scale]{line1=4 4}
  ;

\draw [whiteCircle] (0,4)  circle[] (4,0) circle[] (5,4) circle[];
\draw [blackCircle] (0,0)  circle[] (4,4) circle[] (5,0) circle[];

\node at (8,2) {\bf $\longrightarrow$};

\begin{scope}[xshift=10cm]

\draw[thick]
  (-1,0)--(7,0)
  (-1,4)--(7,4)
  (0,0) pic[scale=\scale]{line2=0}
  (0,0) pic[scale=\scale]{line1=0 2}
  (0,0) pic[scale=\scale]{line2=5}
  (0,0) pic[scale=\scale]{line1=4 4}
  ;

\draw [whiteCircle] (0,4)  circle[] (2,4) circle[] (4,0) circle[] (5,4) circle[];
\draw [blackCircle] (0,0)  circle[] (4,4) circle[] (1,4) circle[] (5,0) circle[];

\node at (8,2) {\bf $\longrightarrow$};

\end{scope}

\begin{scope}[xshift=20cm]

\draw[thick]
  (-1,0)--(7,0)
  (-1,4)--(7,4)
  (0,0) pic[scale=\scale]{line2=0}
  (0,0) pic[scale=\scale]{line1=1 1}
  (0,0) pic[scale=\scale]{line2=5}
  (0,0) pic[scale=\scale]{line1=5 3}
  ;

\draw [whiteCircle] (0,4)  circle[] (3,4) circle[] (1,0) circle[] (5,4) circle[];
\draw [blackCircle] (0,0)  circle[] (4,4) circle[] (1,4) circle[] (5,0) circle[];

\node at (8,2) {\bf $\longrightarrow$};

\end{scope}

\begin{scope}[xshift=30cm]

\draw[thick]
  (-1,0)--(7,0)
  (-1,4)--(7,4)
  (0,0) pic[scale=\scale]{line2=0}
  (0,0) pic[scale=\scale]{line1=1 1}
  (0,0) pic[scale=\scale]{line2=5}
  (0,0) pic[scale=\scale]{line1=5 5}
  ;

\draw [whiteCircle] (0,4)  circle[] (1,0) circle[] (5,4) circle[];
\draw [blackCircle] (0,0)  circle[] (1,4) circle[] (5,0) circle[];

%\node at (8,2) {\bf $\longrightarrow$};

\end{scope}
\end{tikzpicture}

\end{center}
\caption{Moves of the dimer lattice.}\label{fig:moves}
\end{figure}

\subsection{$X$-cluster variety and $Y$-system}\label{ssec:bezout}

The Poisson bracket on $X$-cluster variety is defined in terms of $X$-cluster variables, to be denoted by $\{x_i\}$. The bracket is logarithmically constant and has the form
\begin{equation}\label{eq:Poisson:x}
	\{x_i,x_j\}=\epsilon_{ij}x_ix_j
	\end{equation}
determined by the exchange matrix of quiver. The mutations $\mu_j$ act on such variables by the formula
\be
\label{mutx}
	\mu_j: x_j\mapsto x_j^{-1},\ \ \ \ \
	x_i\mapsto x_i\left(1+ x_j^{{\rm sgn}\epsilon_{ij}}\right)^{\epsilon_{ij}},\quad i\neq j
\ee
Formulas \rf{mutx} and \rf{mute} define simultaneous transformation of $\{x_i\}$ and $\epsilon_{ij}$ that preserves the form of the bracket
 \eqref{eq:Poisson:x}.

From now on assume first that $0\leq k<N$, the case $k=N$ will be discussed separately in Sect.~\ref{ssec:YNN} below. We draw quivers $\mathcal{Q}^{N,k}$ in the plane $\mathbb{R}^2$, and it is natural to assign $X$-variables to the integer points $(n,m)\in \mathbb{Z}^2\subset\mathbb{R}^2$. However, since each point can belong to several grand Motzkin paths (equivalently, to different quivers), one can assign to this point several different $X$-variables. We describe this correspondence below.

First, we consider $\{x_{(n,m)}\}$ corresponding to integer points $(n,m)\in \mathbb{Z}^2$ subject to relations of a sort of a Y-system
\begin{equation}\label{eq:Ysystem}
    \frac{x_{(n,m+1)}x_{(n,m-1)}}{x_{(n,m)}^2}=\frac{(1+x_{(n+1,m)})(1+x_{(n-1,m)})}{(1+x_{(n,m)})^2}
\end{equation}
with the boundary conditions
\begin{equation}\label{eq:x:periodicity}
    x_{(n,m)}=x_{(n+N,m+k)}.
\end{equation}
As initial data one can take $x_{(n,m)}$ in all points of the polyline strip from Lemma~\ref{Lem:YNk:quiver} (with periodicity \eqref{eq:x:periodicity}), and then uniquely determine $x_{(n,m)}$ for all integer points of the plane, using equation \eqref{eq:Ysystem}. Here we actually use the restriction $0\leq k<N$ --- in such case there is always at least one of the following consecutive block's sequences $(0,0)$, $(\mathrm{I},-\mathrm{I})$, $(0,-\mathrm{I})$ or $(\mathrm{I},0)$. Therefore, using a mutation from Fig.~\ref{fig:mut:quiv}, we can determine one $x_{(n,m)}$ below the polyline strip, and then continue inductively. Analogously, one can always recover $x_{(n,m)}$ above the initial polyline strip.

Now, for any given polyline strip we define at each integer point $(n,m)$ of this strip the corresponding cluster variable $x^{***}_{(n,m)}$ by the following rule (which depends both on the shape of the polyline and $+$ or $\times$ type of the point):
\begin{itemize}

    \item $x^{+00}_{(n,m)}=x_{(n,m)}$ if the polyline has the form
    \begin{tikzpicture}[scale=0.5]
    \draw (0,0) -- (1,0) pic[rotate=45,magenta]{crossFilled=1} --
     (2,0) ;
    \end{tikzpicture}
    or
    \begin{tikzpicture}[scale=0.5]
    \draw (0,0) -- (1,0) pic[rotate=45,magenta]{crossFilled=1} --
     (2,-1) ;
    \end{tikzpicture}
    or
    \begin{tikzpicture}[scale=0.5]
    \draw (0,-1) -- (1,0) pic[rotate=45,magenta]{crossFilled=1} -- (2,0) ;
    \end{tikzpicture}
    or
    \begin{tikzpicture}[scale=0.5]
    \draw (0,-1) -- (1,0) pic[rotate=45,magenta]{crossFilled=1} -- (2,-1) ;
    \end{tikzpicture}.

    \item $x^{+10}_{(n,m)}=\dfrac{x_{(n,m)}}{1+x_{(n-1,m+1)}}$ if the polyline has the form
    \begin{tikzpicture}[scale=0.5]
    \draw (0,1) -- (1,0) pic[rotate=45,magenta]{crossFilled=1} --
     (2,-1) ;
    \end{tikzpicture}
    or
    \begin{tikzpicture}[scale=0.5]
    \draw (0,1) -- (1,0) pic[rotate=45,magenta]{crossFilled=1} --
     (2,0) ;
    \end{tikzpicture}.

    \item $x^{+01}_{(n,m)}=\dfrac{x_{(n,m)}}{1+x_{(n+1,m+1)}}$ if the polyline has the form
    \begin{tikzpicture}[scale=0.5]
    \draw (0,-1) -- (1,0) pic[rotate=45,magenta]{crossFilled=1} -- (2,1) ;
    \end{tikzpicture}
    or
    \begin{tikzpicture}[scale=0.5]
    \draw (0,0) -- (1,0) pic[rotate=45,magenta]{crossFilled=1} --
     (2,1) ;
    \end{tikzpicture}.

    \item $x^{+11}_{(n,m)}=\dfrac{x_{(n,m)}}{(1+x_{(n-1,m+1)})(1+x_{(n+1,m+1)})}$ if the polyline has the form
    \begin{tikzpicture}[scale=0.5]
    \draw (0,1) -- (1,0) pic[rotate=45,magenta]{crossFilled=1} --   (2,1) ;
    \end{tikzpicture}.

    \item $x^{\times00}_{(n,m)}=\dfrac{1}{x_{(n,m+2)}}$ if the polyline has the form
    \begin{tikzpicture}[scale=0.5]
    \draw (0,0) -- (1,0) pic[teal]{crossFilled=1} --
     (2,0) ;
    \end{tikzpicture}
    or
    \begin{tikzpicture}[scale=0.5]
    \draw (0,0) -- (1,0) pic[teal]{crossFilled=1} --
     (2,1) ;
    \end{tikzpicture}
    or
    \begin{tikzpicture}[scale=0.5]
    \draw (0,1) -- (1,0) pic[teal]{crossFilled=1} --
     (2,0) ;

    \end{tikzpicture}
    or
    \begin{tikzpicture}[scale=0.5]
    \draw (0,1) -- (1,0) pic[teal]{crossFilled=1} --
     (2,1) ;
    \end{tikzpicture}.

    \item $x^{\times10}_{(n,m)}=\dfrac{1+x_{(n-1,m+1)}}{x_{(n,m+2)}}$ if the polyline has the form
    \begin{tikzpicture}[scale=0.5]
    \draw (0,-1) -- (1,0) pic[teal]{crossFilled=1} --
     (2,0) ;
    \end{tikzpicture}
    or
    \begin{tikzpicture}[scale=0.5]
    \draw (0,-1) -- (1,0) pic[teal]{crossFilled=1} --
     (2,1) ;
    \end{tikzpicture}.

    \item $x^{\times01}_{(n,m)}=\dfrac{1+x_{(n+1,m+1)}}{x_{(n,m+2)}}$ if the polyline has the form
    \begin{tikzpicture}[scale=0.5]
    \draw (0,0) -- (1,0) pic[teal]{crossFilled=1} --
     (2,-1) ;
    \end{tikzpicture}
    or
    \begin{tikzpicture}[scale=0.5]
    \draw (0,1) -- (1,0) pic[teal]{crossFilled=1} --
     (2,-1) ;
    \end{tikzpicture}.

    \item $x^{\times11}_{(n,m)}=\dfrac{(1+x_{(n-1,m+1)})(1+x_{(n+1,m+1)})}{x_{(n,m+2)}}$ if the polyline has the form
    \begin{tikzpicture}[scale=0.5]
    \draw (0,-1) -- (1,0) pic[teal]{crossFilled=1} --
     (2,-1) ;
    \end{tikzpicture}.
\end{itemize}
All these expressions can be actually written uniformly as
\eq{
    x^{+\alpha\beta}_{(n,m)}=\frac{x_{(n,m)}}{(1+x_{(n-1,m+1)})^\alpha(1+x_{(n+1,m+1)})^\beta},\quad
    x^{\times\alpha\beta}_{(n,m)}=\frac{(1+x_{(n-1,m+1)})^\alpha(1+x_{(n+1,m+1)})^\beta}{x_{(n,m+2)}}.
    \label{eq:x_alpha_beta}
}

\begin{thm}
If a quiver $\mathcal{Q}^{N,k}$ is transformed to $\tilde{\mathcal{Q}}^{N,k}$ by mutation from Fig.~\ref{fig:mut:quiv}, the corresponding cluster variables $x_{(n,m)}^{***}$ defined by above rule transform into cluster variables $\tilde{x}_{(n,m)}^{***}$.
\end{thm}
\begin{proof}
The proof just follows from the case by case direct check of the collection of relations, that should be satisfied by $x_{(n,m)}^{***}$ in order to obey conditions of the theorem:
\eq{
    x_{(n,m)}^{+*0}=x_{(n,m)}^{+*1}(1+x_{(n+1,m+1)}),\quad x_{(n,m)}^{+0*}=x_{(n,m)}^{+1*}(1+x_{(n-1,m+1)}),\\
    x_{(n,m)}^{\times*1}=x_{(n,m)}^{\times*0}(1+x_{(n+1,m+1)}),\quad x_{(n,m)}^{\times1*}=x_{(n,m)}^{\times0*}(1+x_{(n-1,m+1)}),\\
    x_{(n,m)}^{+\alpha\beta}=x_{(n,m)}^{\times\alpha'\beta'}(1+x^{-1}_{(n,m)})^{-2},
}
where $\alpha+\alpha'=\beta+\beta'=1$.
Indeed, these relations immediately follow from \eqref{eq:x_alpha_beta} and \eqref{eq:Ysystem}.
\end{proof}

\subsection{Hirota equation from cluster mutations}\label{ssec:Hirota}

Here and below we denote $l_{n,m}=kn-Nm$, this linear function is invariant under our periodic shift $(n,m)\rightarrow (n+N,m+k)$.

\begin{lemma} \label{lem:tau:x}
Let $\tau_{(n,m)}$ satisfy the non-autonomous version of discrete Hirota bilinear equation
\eq{
		\tau_{(n,m+1)}\tau_{(n,m-1)}=\tau_{(n,m)}^2+z_0^{1/N} q^{l_{n,m}/N^2}\tau_{(n+1,m)}\tau_{(n-1,m)}\,,
		\label{eq:toda_deaut1}
	}
together with the boundary condtitons $\tau_{(n,m)}=\tau_{(n+N,m+k)}$. Then $x_{(n,m)}$ defined by
\begin{equation}\label{eq:x:through:tau}
    x_{(n,m)}=z_0^{1/N}q^{(l_{n,m}+N)/N^2} \tau_{(n-1,m-1)}\tau_{(n+1,m-1)}\tau_{(n,m-1)}^{-2}
\end{equation}
satisfy \eqref{eq:Ysystem} and \eqref{eq:x:periodicity}.
\end{lemma}
\begin{proof}
The proof of the lemma is straightforward.
\end{proof}

Let us now rederive the Hirota equations \eqref{eq:toda_deaut1} from the mutations rules of  $\tau$-variables
on $A$-cluster variety. For any vertex $i\in \mathcal{Q}$ of a quiver assign the variable $\tau_i$, so that mutation $\mu_j$ at $j$-th vertex has the form
\be
\label{eq:muttau}
\mu_j:\ \  \tau_j\mapsto \frac{\prod_{\epsilon_{ij}>0} \tau_i^{\epsilon_{ij}}+\prod_{\epsilon_{ij}<0} \tau_i^{-\epsilon_{ij}}}{\tau_j},
\quad\quad
\tau_i\mapsto \tau_i,\ \ \ \ \ i\neq j
\ee
again being supplemented by transformation \rf{mute} of a quiver. Note that the formula \eqref{eq:muttau} is the simplest, coefficient free, case of mutation. For inclusion of coefficients see the formula \eqref{mutay} below.

When quivers $\mathcal{Q}^{N,k}$ are drawn on the plane, it becomes natural to assign $\tau$-variables to all integer points $(n,m)\in \mathbb{Z}^2\subset \mathbb{R}^2$. Let us immediately point out, that since the $A$-cluster mutation \eqref{eq:muttau} changes only the variable at mutation vertex (in contrast to \eqref{mutx}, all variables in other vertices remain intact), the variables $\tau_{(n,m)}$ will be uniquely determined by \rf{eq:toda_deaut1}, not depending on the shape of the polyline strip (in contrast to multiple choices for $x_{(n,m)}^{***}$).

These $\tau$-variables can be constructed as follows: for a quiver $\mathcal{Q}^{N,k}$, described in
Lemma~\ref{Lem:YNk:quiver} and  drawn on a plane, we assign first the variables $\tau_{(n,m)}$ to the integer points $(n,m)$ in the polyline strip (taking into account periodicity). For any choice of $\mathcal{Q}^{N,k}$
(with $0\leq k<N$) one can always make one of the mutations from Fig.~\ref{fig:mut:quiv}, like we already discussed above in the construction for the $X$-variables. There is always a mutation from Fig.~\ref{fig:mut:quiv}, which moves a vertex in the middle down, and at least one mutation moves the vertex up.

Consider mutation at the point $(n,m+1)$, which moves a vertex down, then new polyline strip is obtained by removing the point $(n,m+1)$ and adding the point $(n,m-1)$. For any mutation from Fig.~\ref{fig:mut:quiv} new $\tau$-variable is given by
\begin{equation}\label{eq:mu:taunm}
\tau_{(n,m-1)}=\mu_{(n,m+1)}(\tau_{(n,m+1)})=\frac{\tau_{(n,m)}^2+\tau_{(n+1,m)}\tau_{(n-1,m)}}{\tau_{(n,m+1)}}.
\end{equation}
Performing these mutations one assigns $\tau$-variables to each integer point of the plane. Since equation \eqref{eq:mu:taunm} does not depend on the concrete type of the mutation, such assignment is unambiguously determined. In other words, we obtain in such a way the solutions of the Hirota type bilinear equation
(see e.g. review~\cite{Zabrodin_Hir} and references therein)
\eq{
	\tau_{(n,m+1)}\tau_{(n,m-1)}=\tau_{(n,m)}^2+\tau_{(n+1,m)}\tau_{(n-1,m)}
	\label{eq:toda_aut1}
}
with the periodicity condition
\eq{
	\tau_{(n+N,m+k)}=\tau_{(n,m)}\,.
	\label{eq:boundary_condition1}
}
The values of $\tau$-variables on initial polyline strip correspond to
the initial conditions for the discrete Hirota equation \eqref{eq:toda_aut1}.

Hence, we have obtained equation \eqref{eq:toda_aut1}, which is a special case of \eqref{eq:toda_deaut1} with constant coefficients, or $q=z_0=1$. In this simplest case the relation between the $\tau$-variables and $x$-variables, satisfying \rf{eq:Ysystem}, is given by $x_i=\prod_j{\tau_j^{\epsilon_{ji}}}$, which together with formula for mutations \eqref{eq:muttau} reproduces the mutation rules \eqref{mutx}. Consider however the Casimir or central elements of the Poisson algebra \eqref{eq:Poisson:x}: it is easy to see that they correspond to the kernel of matrix $\epsilon$. Expressing them by $x_i=\prod_j{\tau_j^{\epsilon_{ji}}}$ through the $\tau$-variables, one gets very restricted values of the Casimirs (e.g. for the monomial ones --- just unities), i.e. we obtain a very special symplectic leaf in the Poisson $X$-cluster variety.

  In order to consider generic situation one should modify the relation between $x$- and $\tau$-variables into $x_i=y_i\prod_j{\tau_j^{\epsilon_{ji}}}$ by introducing coefficients $\{y_i\}$~\footnote{Or, equivalently, by adding frozen vertices, corresponding to the Casimirs, to a quiver.}.
	In our case for the exchange matrix we have ${\rm corank}(\epsilon)=2$, therefore the bracket \eqref{eq:Poisson:x} has two independent Casimir functions, which can be chosen as
	\begin{equation}
	q=\prod x_{(n,m)}^{***}=\prod y_{(n,m)}^{***}\,,\qquad z=\prod \Big(x_{(n,m)}^{***}\Big)^{l_{n,m}}=\prod \Big(y_{(n,m)}^{***}\Big)^{l_{n,m}}\,.
	\end{equation}
The product here is taken over the vertices in a fundamental domain (under $(N,k)$-translation) of the polyline strip. Without coefficients these Casimirs turn into unities. It is therefore natural to express the coefficients $\{ y_i\}$ through two Casimir variables.

One can start with generic coefficients $\{y_i\}$ in the tropical semifield \cite{FZ:2006}: $y_{(n,m)}\in {\rm Trop}(z_0,q)$,
where the tropical operations on ${\rm Trop}(z_0,q)$ are
\eqs{
	z_0^a q^b\odot z_0^c q^d&=z_0^{a+c} q^{b+d}\,,\qquad
	z_0^a q^b\oplus z_0^c q^d&=z_0^{\min(a,c)} q^{\min(b,d)}\,.
\label{trop}
}
In this case mutations \rf{eq:muttau} are modified by coefficients
\eqs{
	&\mu_j(\tau_j)=\frac{y_j\prod\nolimits_{\epsilon_{ij}>0}\tau_i + \prod\nolimits_{\epsilon_{ij}<0}\tau_i }{(1\oplus y_j)\tau_j}\,,&\qquad 	&\mu_j(\tau_i)=\tau_i,\quad i\neq j\,,\\
	&\mu_j(y_j)=y_j^{-1}\,,&\qquad
	&\mu_j(y_i)=y_i(1\oplus y_j^{{\rm sgn}\epsilon_{ij}})^{\epsilon_{ij}},\quad i\neq j\,,
\label{mutay}
}
to be supplemented by mutation rules for the coefficients $\{y_j\}$ themselves, which are the same as mutations of the $x$-variables \eqref{mutx}, up to replacement $+$ by tropical $\oplus$ from \rf{trop}. The coefficient's assignment
should be consistent with tropical version of \eqref{eq:Ysystem}.

Set, for example $y_{(n,m)}=z_0^{1/N} q^{(l_{n,m}+N)/N^2}$, then in the region $l_{n,m}>0$ we have an obvious equality
\begin{equation}\label{eq:Ysystem:y}
    \frac{y_{(n,m+1)}y_{(n,m-1)}}{y_{(n,m)}^2}=\frac{(1\oplus y_{(n+1,m)})(1\oplus y_{(n-1,m)})}{(1\oplus y_{(n,m)})^2}\,.
\end{equation}
Due to tropical addition \rf{trop} the corresponding $y_{(n,m)}^{***}$ depend only on the vertex type, and do not depend on the shape of the polyline strip, i.e.
\begin{equation}
    y^{ +**}_{(n,m)}=z_0^{1/N} q^{(l_{n,m}+N)/N^2},\quad
	y^{\times**}_{(n,m)}=z_0^{-1/N} q^{(-l_{n,m}+N)/N^2}\,.
	\label{eq:coefficients1}
\end{equation}
For a mutation from \rf{mutay} one now gets
\eq{\label{eq:mut:tau:y}
	\tau_{(n,m-1)}=\mu_{(n,m+1)}\left(\tau_{(n,m+1)}\right)=\frac{\tau_{(n,m)}^2+z_0^{1/N} q^{l_{n,m}/N^2}\tau_{(n+1,m)}\tau_{(n-1,m)}}{\tau_{(n,m+1)}}\,,
}
and in this way we obtain the generic Hirota equation \eqref{eq:toda_deaut1} from cluster mutations.

\begin{Remark}
    Strictly speaking, we have obtained equation \eqref{eq:mut:tau:y} only in the region $l_{n,m}>0$.
    Certainly, this asymmetry comes from the tropical summation $\oplus$ in \rf{trop}. However, for any given $(n,m)$ on can redefine $z_0\rightarrow z_0q^{-c}$ for sufficiently large $c$ and obtain a positive
    $q$-exponent in $y_{(n,m)}$, then one returns to \eqref{eq:mut:tau:y}. This can be considered as an explanation why Lemma~\ref{lem:tau:x} does not have any restrictions on $(n,m)$.

    On the other side, one can ask what happens with the continuation of the solution $y_{(n,m)}=z_0^{1/N} q^{(l_{n,m}+N)/N^2}$ of the equation \eqref{eq:Ysystem:y} to the region $l_{n,m}<0$. It appears that after crossing the line for $l_{n,m}=0$  the behaviour of the solution changes drastically. We are going to discuss this issue elsewhere.
\end{Remark}

It is also convenient to rewrite \eqref{eq:toda_deaut1}, using the variables $j=n\ {\rm mod}\ N$ and $l = l_{n,m}=kn-Nm$: it takes the form
\eq{
	\tau_{j,l+N}\tau_{j,l-N}=\tau_{j,l}^2+z_0^{1/N} q^{l/N^2} \tau_{j+1,l+k}\tau_{j-1,l-k}\,.
\label{eq:toda_deaut2}
}
Then in the case $\gcd(N,k)=1$ the $\tau$-variables actually depend only on $l$, since the index $j$ is determined by $l$ and can be dropped. For generic case we have effectively $d=\gcd(N,k)$ auxiliary indices $ i\in \mathbb{Z}/d\mathbb{Z}$, for example for $k=0$ one gets $i\in \mathbb{Z}/N\mathbb{Z}$. However, it is convenient (for any $N,k$) to arrange $\tau$-variables into $N$ tau-functions by the formula $\tau_{j,l}=\tau_j(z_0 q^{l/N})$. Then equation \eqref{eq:toda_deaut2} can be rewritten as a $q$-difference equations in the variable $z=z_0 q^{l/N}$:
\eq{
	\tau_{j}\lb q z\rb \tau_{j}\lb q^{-1} z\rb=\tau_j(z)^2+z^{1/N} \tau_{j+1}\lb q^{k/N}z\rb \tau_{j-1}\lb q^{-k/N}z\rb\,.
	\label{eq:toda_deaut3}
}
Equations \eqref{eq:Ysystem}, \eqref{eq:toda_deaut1} and \eqref{eq:toda_deaut3} are the main results of Sect.~\ref{sec:Todas}. In Sect.~\ref{sec:sol} we present solutions of these equations.		
Recall that these equation have been derived for $k=0,\ldots, N-1$. It turns out however, that in $Y^{N,N}$ case one gets the same equations: we derive them below in Sect.~\ref{ssec:YNN}.

\subsection{The ``Uniform'' quiver}

\noindent
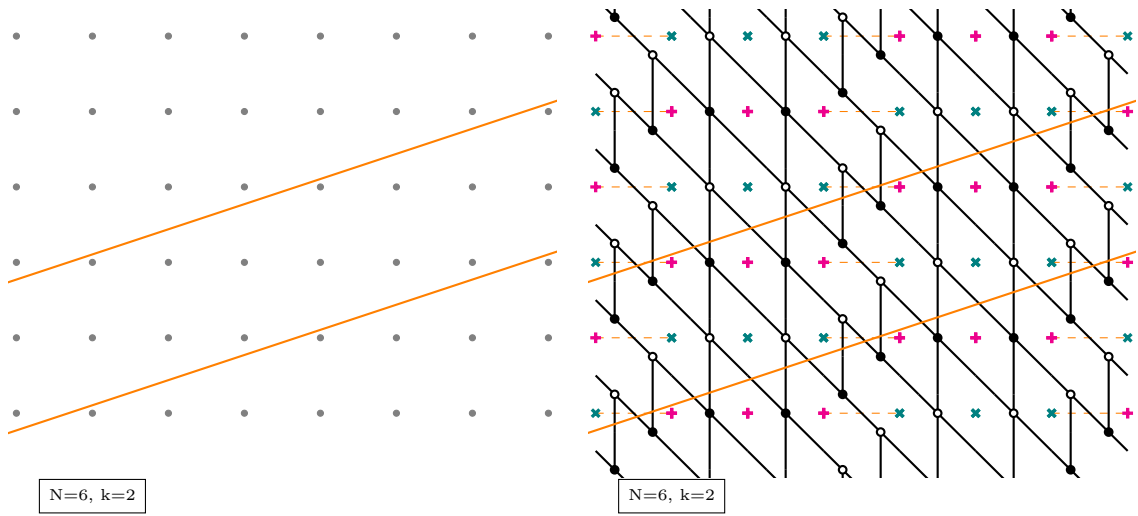
\begin{figure}[h!]
\begin{center}
\begin{tabular}{cc}

\begin{tikzpicture}

\tikzmath{int \N,\k,\Nk,\Lx,\Ly,\d,\Nd,\kd;
\N=6; \k=2;
\Lx=\N+1; \Ly=3;
\Nk=\N-\k;
\d=gcd(\N,\k); \Nd=\N/\d; \kd=\k/\d;
\scale=0.5;}

\tikzset{
style1/.style={dashed},
color1/.style={red},
color2/.style={green},
color3/.style={orange},
}

\begin{scope}[scale=\scale,font=\tiny]

\begin{scope}

\clip ($(-2,-1.5)-(0.2,0.2)$) rectangle ($(2*\Lx-2,4*\Ly-1.5)+(0.2,0.2)$);

\begin{scope}[xshift=-2cm]

\tikzmath{
for \i in {0,...,\Lx}{
 for \j in {0,...,2*\Ly-1}{
   {\draw[gray,fill=gray](2*\i,2*\j) circle(0.08);};
 };
};
}

\end{scope}

\tikzmath{
  \shift=0.2*sqrt(\N*\N+\k*\k)/\k;
}
\coordinate (slope) at ($(1,\k/\N)$);
\coordinate (start) at (-\shift,0);
\draw[color3,thick]($(start)-2*(slope)$)--+($2*\Lx*(slope)+3*(slope)$);
\coordinate(X1) at ($($(start)-2*(slope)$)!(0.7,0)!($(start)+2*\Lx*(slope)$)$);

\coordinate (start) at (-\shift,4);
\draw[color3,thick]($(start)-2*(slope)$)--+($2*\Lx*(slope)+3*(slope)$);

\tikzmath{
   \shift=0.2*sqrt(\N*\N+\k*\k)/\k-2/\kd;
}

\coordinate (start) at (-\shift,0);
%\draw[color3,thick,densely dotted]($(start)-4*(slope)$)--+($2*\Lx*(slope)+3*(slope)$);
\coordinate (X2) at ($($(start)-3*(slope)$)!(0.7,0)!($(start)+2*\Lx*(slope)$)$);

%\draw[color3,thick,densely dotted]($(start)-4*(slope)+(0,4)$)--+($2*\Lx*(slope)+3*(slope)$);

% \draw[color3,->,thick] ($(X1)+(0,4)$)--($(X2)+(0,4)$);
% \draw[color3,->,thick] ($(X1)+(0,0)$)--($(X2)+(0,0)$);

\end{scope}

\node[draw] at (0,-2.2) {N=\N, k=\k};

\end{scope}

\end{tikzpicture}

&

\begin{tikzpicture}

\tikzmath{int \N,\k,\Nk,\Lx,\Ly,\d,\Nd,\kd;
\N=6; \k=2;
\Lx=\N+1; \Ly=3;
\Nk=\N-\k;
\d=gcd(\N,\k); \Nd=\N/\d; \kd=\k/\d;
\scale=0.5;}

\tikzset{
style1/.style={dashed},
color1/.style={red},
color2/.style={green},
color3/.style={orange},
}

\begin{scope}[scale=\scale,font=\tiny]

\begin{scope}

\clip ($(-2,-1.5)-(0.2,0.2)$) rectangle ($(2*\Lx-2,4*\Ly-1.5)+(0.2,0.2)$);

\tikzmath{
for \i in {-1,...,\Lx-1}{
%  {\draw[style1](2*\i,-1.5)--(2*\i,4*\Ly-1.5);};
};
}

\tikzmath{
for \z in {-1,...,\Ly}{
\j=0;\l=-1;
for \i in {-1,...,\Lx-2}{
  if ((\i+1)*\k/\N)>=\j then{
    \j=\j+1;
    {
     \draw[color3,style1] (\i*2,4*\z)--(\i*2+2,4*\z) (\i*2,4*\z+2)--(\i*2+2,4*\z+2);
     \draw (\i*2,\z*4-1-\l) pic[scale=\scale] {dimer2};
%     \draw (\i*2,\z*4-1-\l) pic[scale=\scale] {arrows2};
     \draw[gray,thick] (\i*2,\z*4-1-\l) pic {vertices2=\scale};
      };
%     if Mod(\k*(\i+1)-\N*\z,\N)==0 then {{\draw[color3] (\i*2+2,\z*4+1+\l) circle[radius=0.2];};};
%     if Mod(\kd*(\i+1)-\Nd*\z,\Nd)==1 then {{\draw[color3] ($(\i*2+2,\z*4+1+\l)-(0.2,0.2)$) rectangle ($(\i*2+2,\z*4+1+\l)+(0.2,0.2)$);};};
     \l=-\l;
  } else{
    {
%     \draw[style1] (\i*2,4*\z)--(\i*2+2,4*\z) (\i*2,4*\z+2)--(\i*2+2,4*\z+2);
     \draw  (\i*2,\z*4-1+\l) pic[scale=\scale] {dimer1};
%     \draw  (\i*2,\z*4-1+\l) pic[scale=\scale] {arrows1};
     \draw[gray,thick]  (\i*2,\z*4-1+\l) pic {vertices1=\scale};
      };};
%     if Mod(\kd*(\i+1)-\Nd*\z,\Nd)==1 then {{\draw[color3] ($(\i*2+2,\z*4+1+\l)-(0.2,0.2)-(0,2)$) rectangle ($(\i*2+2,\z*4+1+\l)+(0.2,0.2)-(0,2)$);};};
  };
};
}

\tikzmath{
  \shift=0.2*sqrt(\N*\N+\k*\k)/\k;
}
\coordinate (slope) at ($(1,\k/\N)$);
\coordinate (start) at (-\shift,0);
\draw[color3,thick]($(start)-2*(slope)$)--+($2*\Lx*(slope)+3*(slope)$);
\coordinate(X1) at ($($(start)-2*(slope)$)!(0.7,0)!($(start)+2*\Lx*(slope)$)$);

\coordinate (start) at (-\shift,4);
\draw[color3,thick]($(start)-2*(slope)$)--+($2*\Lx*(slope)+3*(slope)$);

\tikzmath{
   \shift=0.2*sqrt(\N*\N+\k*\k)/\k-2/\kd;
}

\coordinate (start) at (-\shift,0);
%\draw[color3,thick,densely dotted]($(start)-4*(slope)$)--+($2*\Lx*(slope)+3*(slope)$);
\coordinate (X2) at ($($(start)-3*(slope)$)!(0.7,0)!($(start)+2*\Lx*(slope)$)$);

%\draw[color3,thick,densely dotted]($(start)-4*(slope)+(0,4)$)--+($2*\Lx*(slope)+3*(slope)$);

% \draw[color3,->,thick] ($(X1)+(0,4)$)--($(X2)+(0,4)$);
% \draw[color3,->,thick] ($(X1)+(0,0)$)--($(X2)+(0,0)$);

\end{scope}

\node[draw] at (0,-2.2) {N=\N, k=\k};

\end{scope}

\end{tikzpicture}

\end{tabular}
\end{center}
\caption{Construction of the "Uniform" dimer lattice.}
\label{fig:uniform_lattice}
\end{figure}

\begin{figure}[h!]
\begin{center}
\noindent
\begin{tabular}{cc}

\begin{tikzpicture}

\tikzmath{int \N,\k,\Nk,\Lx,\Ly,\d,\Nd,\kd;
\N=6; \k=2;
\Lx=\N+1; \Ly=3;
\Nk=\N-\k;
\d=gcd(\N,\k); \Nd=\N/\d; \kd=\k/\d;
\scale=0.5;}

\tikzset{
style1/.style={dashed},
color1/.style={red},
color2/.style={green},
color3/.style={orange},
}

\begin{scope}[scale=\scale,font=\tiny]

\begin{scope}

\clip ($(-2,-1.5)-(0.2,0.2)$) rectangle ($(2*\Lx-2,4*\Ly-1.5)+(0.2,0.2)$);

\begin{scope}[xshift=-2cm]

\tikzmath{
for \i in {0,...,\Lx}{
 for \j in {0,...,2*\Ly-1}{
   {\draw[gray,fill=gray](2*\i,2*\j) circle(0.08);};
 };
};
}

\tikzmath{
\j=0;
for \i in {0,...,\N}{
 if (\i*\k/\N)>=\j then{
     {\draw(2*\i,2*\j) pic{quiver2=\scale};};
     \j=\j+1;
   }else{
     {\draw(2*\i,2*\j) pic{quiver1=\scale};};
     };
 };
}

\tikzmath{
\j=0;\l=0;
for \i in {0,...,\N}{
 if (\i*\k/\N)>=\j then{
     if Mod(\k*\i,\N)==0 then {
       {\draw[orange](2*\i+2,2*\j+4) circle[radius=0.2];
        \draw(2*\i+2,2*\j) pic[gray,rotate=0]{crossFilled=\scale};
        };};
     \j=\j+1;
     };
 };
if \k==\N then {{\draw[orange] (0,2) circle[radius=0.2];};};
}

\end{scope}

\tikzmath{
  \shift=0.2*sqrt(\N*\N+\k*\k)/\k;
}
\coordinate (slope) at ($(1,\k/\N)$);
\coordinate (start) at (-\shift,0);
\draw[color3,thick]($(start)-2*(slope)$)--+($2*\Lx*(slope)+3*(slope)$);
\coordinate(X1) at ($($(start)-2*(slope)$)!(0.7,0)!($(start)+2*\Lx*(slope)$)$);

\coordinate (start) at (-\shift,4);
\draw[color3,thick]($(start)-2*(slope)$)--+($2*\Lx*(slope)+3*(slope)$);

\tikzmath{
   \shift=0.2*sqrt(\N*\N+\k*\k)/\k-2/\kd;
}

\coordinate (start) at (-\shift,0);
\draw[color3,thick,densely dotted]($(start)-4*(slope)$)--+($2*\Lx*(slope)+3*(slope)$);
\coordinate (X2) at ($($(start)-3*(slope)$)!(0.7,0)!($(start)+2*\Lx*(slope)$)$);

\draw[color3,thick,densely dotted]($(start)-4*(slope)+(0,4)$)--+($2*\Lx*(slope)+3*(slope)$);

\draw[color3,->,thick] ($(X1)+(0,4)$)--($(X2)+(0,4)$);
\draw[color3,->,thick] ($(X1)+(0,0)$)--($(X2)+(0,0)$);

\fill[gray,opacity=0.2] ($(-2,-1.5)-(0.2,0.2)$) rectangle ($(-1,4*\Ly-1.5)+(0,0.2)$)
($(2*\N-1,-1.5)-(0,0.2)$) rectangle ($(2*\Lx-2,4*\Ly-1.5)+(0.2,0.2)$);

\end{scope}

\node[draw] at (0,-2.2) {N=\N, k=\k};

\end{scope}

\end{tikzpicture}
&

\begin{tikzpicture}

\tikzmath{int \N,\k,\Nk,\Lx,\Ly,\d,\Nd,\kd;
\N=6; \k=2;
\Lx=\N+1; \Ly=3;
\Nk=\N-\k;
\d=gcd(\N,\k); \Nd=\N/\d; \kd=\k/\d;
\scale=0.5;}

\tikzset{
style1/.style={dashed},
color1/.style={red},
color2/.style={green},
color3/.style={orange},
}

\begin{scope}[scale=\scale,font=\tiny]

\begin{scope}

\clip ($(-2,-1.5)-(0.2,0.2)$) rectangle ($(2*\Lx-2,4*\Ly-1.5)+(0.2,0.2)$);

\begin{scope}[xshift=-2cm]

\tikzmath{
for \i in {0,...,\Lx}{
 for \j in {0,...,2*\Ly-1}{
   {\draw[gray,fill=gray](2*\i,2*\j) circle(0.08);};
 };
};
}

\tikzmath{
\j=0;
for \i in {0,...,\N}{
 if (\i*\k/\N-1/\Nd)>=\j then{
     {\draw(2*\i,2*\j) pic{quiver2=\scale};};
     \j=\j+1;
   }else{
     {\draw(2*\i,2*\j) pic{quiver1=\scale};};
     };
 };
}

\tikzmath{
\j=0;\l=0;
for \i in {0,...,\N}{
 if (\i*\k/\N)>=\j then{
     if Mod(\k*\i,\N)==0 then {
       {\draw[orange](2*\i+2,2*\j+4) pic[gray,rotate=45]{crossFilled=\scale};
        \draw[orange](2*\i+2,2*\j) circle[radius=0.2];
        };};
     \j=\j+1;
     };
 };
if \k==\N then {{\draw[orange] (0,2) circle[radius=0.2];};};
}

\end{scope}

\tikzmath{
  \shift=0.2*sqrt(\N*\N+\k*\k)/\k;
}
\coordinate (slope) at ($(1,\k/\N)$);
\coordinate (start) at (-\shift,0);
\draw[color3,thick,densely dotted]($(start)-2*(slope)$)--+($2*\Lx*(slope)+3*(slope)$);
\coordinate(X1) at ($($(start)-2*(slope)$)!(0.7,0)!($(start)+2*\Lx*(slope)$)$);

\coordinate (start) at (-\shift,4);
\draw[color3,thick,densely dotted]($(start)-2*(slope)$)--+($2*\Lx*(slope)+3*(slope)$);

\tikzmath{
   \shift=0.2*sqrt(\N*\N+\k*\k)/\k-2/\kd;
}

\coordinate (start) at (-\shift,0);
\draw[color3,thick]($(start)-4*(slope)$)--+($2*\Lx*(slope)+3*(slope)$);
\coordinate (X2) at ($($(start)-3*(slope)$)!(0.7,0)!($(start)+2*\Lx*(slope)$)$);

\draw[color3,thick]($(start)-4*(slope)+(0,4)$)--+($2*\Lx*(slope)+3*(slope)$);

\draw[color3,<-,thick] ($(X1)+(0,4)$)--($(X2)+(0,4)$);
\draw[color3,<-,thick] ($(X1)+(0,0)$)--($(X2)+(0,0)$);

\fill[gray,opacity=0.2] ($(-2,-1.5)-(0.2,0.2)$) rectangle ($(-1,4*\Ly-1.5)+(0,0.2)$)
($(2*\N-1,-1.5)-(0,0.2)$) rectangle ($(2*\Lx-2,4*\Ly-1.5)+(0.2,0.2)$);

\end{scope}

\node[draw] at (0,-2.2) {N=\N, k=\k};

\end{scope}

\end{tikzpicture}

\end{tabular}
\end{center}
\caption{Mutations of the "Uniform" quiver.}
\label{fig:uniform_mutations}
\end{figure}
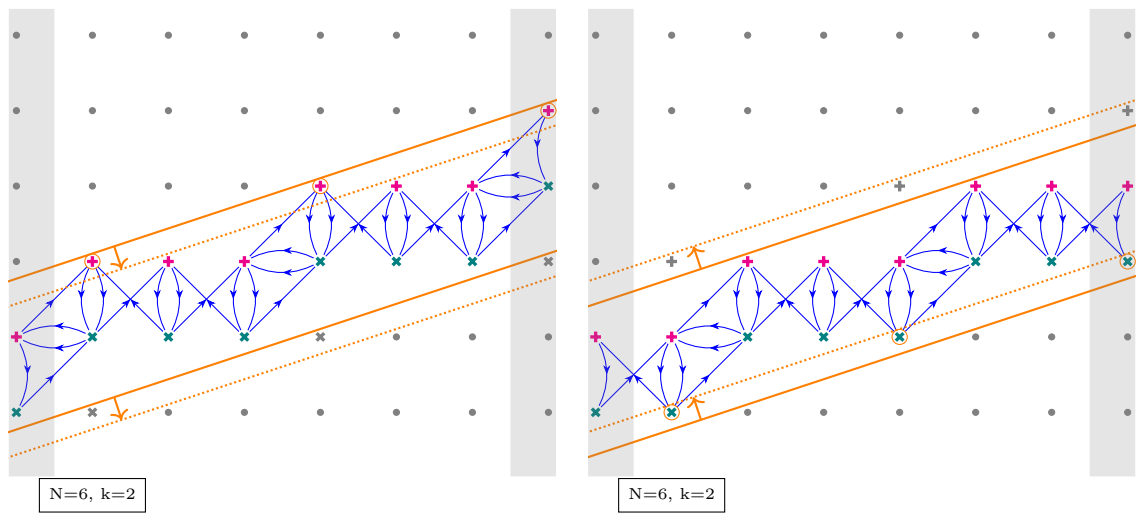

\noindent
The main object in the paper \cite{BGM} was the group $\mathcal{G}_\mathcal{Q}$. For a given quiver $\mathcal{Q}$ this group consists of compositions of mutations and permutations of the vertices, which preserve the quiver $\mathcal{Q}$. An element $T\in\mathcal{G}_\mathcal{Q}$ of infinite order generates a discrete flow.

As explained above, mutations correspond to changing the shape of the polyline strip (see Lemma~\ref{Lem:YNk:quiver} and Lemma~\ref{Lem:mut}), group $\mathcal{G}_\mathcal{Q}$ consists of the elements which transform the polyline strip to another strip of the same shape~\footnote{Actually we have not proven that any element of $\mathcal{G}_\mathcal{Q}$ can be given as a composition of mutations in Fig.~\ref{fig:mut:quiv} and permutations, but we believe that this is indeed true.}.
Such transformations exist for any initial shape of the polyline strip for $0<k<N$, but there exists ``the best'' (or ``uniform'') shape of such strip and corresponding ``uniform''  quiver $\mathcal{Q}^{N,k}_{\mathrm{u}}$~\footnote{Like the strictly horizontal shape for $k=0$. This case stands a little bit aside, but one can naturally identify $\mc{Q}^{N,0}_{\rm u}$ with a sequence of blocks of type 0.}.
For the ``uniform''  quiver $\mathcal{Q}^{N,k}_{\mathrm{u}}$ the discrete flow element $T\in\mathcal{G}_{\mathcal{Q}^{N,k}}$ can be presented as a composition of $d=\gcd(N,k)$ mutations and permutation. The construction of such ``uniform'' $\mathcal{Q}^{N,k}_{\mathrm{u}}$ is not used in other parts of the paper and can be viewed as a kind of elementary olympiad problem, this construction appeared, for example, in \cite{Ko}.

The construction goes as follows. Draw two slanting lines of slope $k/N$: $(x,y)\in \{(0,\epsilon+n)+t(N,k)|t\in \mathbb R, n=0,2\}$, where $\epsilon$ is a sufficiently small positive number. The integer points between these two lines form the polyline strip of width 2, this strip can be filled by the blocks of type $0$ and $\mathrm{I}$. Therefore, by Lemma~\ref{Lem:YNk:quiver} these integer points can be viewed as the vertices of the quiver of $\mathcal{Q}^{N,k}_{\mathrm{u}}$.

The transformation  $T\in\mathcal{G}_{\mathcal{Q}^{N,k}_{\mathrm{u}}}$ can be realized as a shift of the slanting lines down by $\gcd(N,k)/N$. During such shift each of these lines goes though $d=\gcd(N,k)$ integer points (up to $(N,k)$ periodicity), they are represented by circles in the figure below. The transformation of the quiver from Lemma~\ref{Lem:mut} can be given by mutations $(\mathrm{I},0)\leftrightarrow (0,\mathrm{I})$ in these integer points. We see that the integer points between shifted lines form the strip of the same shape as between original ones. Therefore, the resulting quiver coincides with $\mathcal{Q}^{N,k}_{\mathrm{u}}$ up to permutation of the vertices.

\subsection{Discrete flows for $Y^{N,N}$ systems}
\label{ssec:YNN}

The Thurston diagram, bipartite graph and quiver $\mathcal{Q}^{N,N}$ for the $Y^{N,N}$-triangle were described in Sect.~\ref{ssec:Thurst}. By Lemma~\ref{Lem:YNkThurst} they consist of $N$ blocks of type $\mathrm{I}$. Therefore no allowed mutations from Lemma~\ref{Lem:mut} can be performed in this case. However, there exist rather nontrivial compositions of mutations which preserve such quiver~\footnote{Note that there are no spider moves for such bipartite graphs, since it comes out of hexagonal bi-partite lattice, see also \cite{ILP} where the discrete flows for arbitrary hexagonal lattices were studied.}.

The polyline strip, which consist of the vertices of the quiver, can be now drawn between two lines with unit slope, for example $y=x-1+m+\epsilon$ and $y=x+1+m+\epsilon$, where $\epsilon$ is a small real number. Let us label by $j\in \mathbb{Z}/N\mathbb{Z}$ the ``upper vertices'' (with the coordinates $(j,j+m+1)$) and by $j'\in \mathbb{Z}/N\mathbb{Z}$ --- the ``lower vertices'' (with the coordinates $(j,j+m)$). Then discrete flow is given by the quiver automorphism of the form
\eq{\label{eq:T:YNN}
	T=(1 3' 3 4' 4 \ldots N' N 1')(2' 2)\cdot \mu_N \mu_{N-1} \ldots \mu_2 \mu_1 \cdot \mu_3 \mu_4 \ldots \mu_N\,.
}
On a plane this automorphism of the quiver can be interpreted as a unit shift down
\begin{equation}
\tau_{(j,j+m)}=T(\tau_{(j,j+m+1)})\,,\qquad \tau_{(j,j+m-1)}=T(\tau_{(j,j+m)})\,,
\end{equation}
so that tau-functions in the ``lower'' vertices, which become ``upper'' vertices of the quiver after this shift down, remain intact, but one gets nontrivial formulas for the ``new lower'' transformed tau-functions $\{\tau_{(j,j+m-1)}\}$. As before, these formulas are equivalent to bilinear relations. As in Sect.~\ref{ssec:Hirota} let us introduce $N$ tau-functions by collecting $\tau_{j,j+m}=\tau_j(q^m z_0)$, and introduce nontrivial coefficients $\{y_{j,m}\}$, related with the Casimir functions; then the action of the generator $T$ is equivalent to bilinear equations
\eq{\label{eq:bilin:YNN}
	\tau_{j}\lb q z\rb \tau_{j}\lb q^{-1} z\rb=(1-z)\tau_j(z)^2+z^{1/N} \tau_{j+1}\lb qz\rb \tau_{j-1}\lb q^{-1}z\rb, \quad j\in \mathbb{Z}/N\mathbb{Z}\,,\quad z=z_0q^m.
}
We do not present here complete (rather cumbersome) proof of \eqref{eq:bilin:YNN}, but illustrate it on the example below.
Note immediately, that in the case $k=N$ (unlike $k<N$ case) the bilinear relations \rf{eq:bilin:YNN} are not recurrent formulas for evolution in discrete time, since they contain simultaneously $\tau_j(qz)$ in the l.h.s. and $\tau_{j+1}(qz)$ in the r.h.s. As is illustrated in the example, the recurrent relations themselves are more complicated.
\begin{Example}
	Consider the case $Y^{3,3}$, the action of $T$ defined in \eqref{eq:T:YNN} gives here the following trilinear relations:
	\eq{\label{eq:Y33:sol}
\tau_1(qz)\tau_1(q^{-1}z)\tau_2(q^{-1}z)=\tau_2(q^{-1}z)\tau_1(z)^2+z^{1/3}\tau_3(q^{-1}z)\tau_2(z)^2+z^{2/3} \tau_1(q^{-1}z)\tau_3(z)^2\,,\\	 \tau_2(qz)\tau_3(q^{-1}z)\tau_3(q^{-1}z)=z^{2/3}\tau_2(q^{-1}z)\tau_1(z)^2+\tau_3(q^{-1}z)\tau_2(z)^2+z^{1/3} \tau_1(q^{-1}z)\tau_3(z)^2\,,\\		 \tau_3(qz)\tau_3(q^{-1}z)\tau_1(q^{-1}z)=z^{1/3}\tau_2(q^{-1}z)\tau_1(z)^2+z^{2/3} \tau_3(q^{-1}z)\tau_2(z)^2+\tau_1(q^{-1}z)\tau_3(z)^2\,.
	}
	On the other hand the bilinear equations  \eqref{eq:bilin:YNN} form a system of three linear equations for $\{\tau_j(qz)\}$
\eqs{
\begin{pmatrix}
\tau_1(q^{-1}z)&  -z^{1/3}\tau_3(q^{-1}z)& 0 \\
0 &\tau_2(q^{-1}z) &-z^{1/3}\tau_1(q^{-1}z) \\
-z^{1/3}\tau_2(q^{-1}z)& 0 &\tau_3(q^{-1}z)
\end{pmatrix}
\begin{pmatrix}
\tau_1(qz)\\
\tau_2(qz)\\
\tau_3(qz)
\end{pmatrix}
=(1-z)
\begin{pmatrix}
\tau_1(z)^2\\	
\tau_2(z)^2\\	
\tau_3(z)^2	
\end{pmatrix}.
 	}
 Solving this system using Cramer's rule one gets \eqref{eq:Y33:sol}.
\end{Example}

In order to relate the bilinear form \eqref{eq:bilin:YNN} to to \eqref{eq:toda_deaut2} we introduce new variables $\hat\tau_j(z)$ by the formula
\eq{\label{eq:hattau}
\hat\tau_j(z)=(qz;q,q)_\infty \tau_j(z)\,,
}
where $q$-Pochhammer symbol $(qz;q,q)_\infty = \prod_{i,j=0}^\infty(1-zq^{1+i+j})=\prod_{n=1}^\infty(1-zq^{n})^n$ (see also more generic formula \eqref{eq:Poch} below). Then renormalized functions $\hat\tau_j(z)$ solve exactly \eqref{eq:toda_deaut2}, which in this case has the form
\eq{
\hat \tau_{j}\lb q z\rb \hat \tau_{j}\lb q^{-1} z\rb=\hat \tau_j(z)^2+z^{1/N} \hat \tau_{j+1}\lb qz\rb \hat \tau_{j-1}\lb q^{-1}z\rb\,.
}
We will also need autonomous verstion of these bilinear relations. Note, however, that for $q=1$ the substitution \eqref{eq:hattau} is ill defined. Therefore, in order to remove the factor $(1-z)$ in \eqref{eq:bilin:YNN}, we use another substitution, namely
\eq{
\tau_{j,m}=(1-z)^{\frac12(m-m^2)}\tilde\tau_{j,m}\,,
}
so that the variables $\tilde\tau_{j,m}$ in the r.h.s. satisfy equation \eqref{eq:toda_deaut1}.

\subsection{Discrete flows for $L^{1,2N-1,2}$ systems}
\label{ssec:L12n-12}

The algorith for the $L^{1,2N-1,2}$ polygon is again the same: one constructs the Thurston diagram, bipartite graph on torus, and the quiver. We omit here the details, since they are similar to previous cases, and moreover, there is no clear interpretation of the
intermediate stages~\footnote{See also \cite{Hanany2}, where using another method the case of generic $L^{a,b,c}$ was discussed.}. Hence, we just present the final form of the quiver on Fig.~\ref{fig:quiverL}.

\tikzset{quiver2NN/.pic={
\tikzset{styleArrow/.style={blue,postaction={decorate},decoration={markings,mark=at position 0.5 with {\arrow{Stealth[scale=1.5]}}}}}

\draw [styleArrow,bend left] (0,2) to (0,0);
\draw [styleArrow,bend right] (2,2) to (2,0);
\draw [styleArrow,bend left=20] (2,0) to (0,2);
\draw [styleArrow,bend right=20] (2,0) to (0,2);
\draw [styleArrow] (0,2) to (2,2);
\draw [styleArrow] (0,0) to (2,0);

\draw (0,0) pic[magenta]{crossFilled=1}
(2,0) pic[magenta]{crossFilled=1}
(2,2) pic[magenta]{crossFilled=1}
(0,2) pic[magenta]{crossFilled=1};

}
}

\tikzset{quiver3NN/.pic={
\tikzset{styleArrow/.style={blue,postaction={decorate},decoration={markings,mark=at position 0.5 with {\arrow{Stealth[scale=1.5]}}}}}

\draw [styleArrow] (0,0) to (2,2);
\draw [styleArrow] (0,2) to (2,0);
\draw [styleArrow] (2,2) to (0,2);
\draw [styleArrow] (2,0) to (0,0);

\draw (0,0) pic[magenta]{crossFilled=1}
(2,0) pic[magenta]{crossFilled=1}
(2,2) pic[magenta]{crossFilled=1}
(0,2) pic[magenta]{crossFilled=1};

}
}

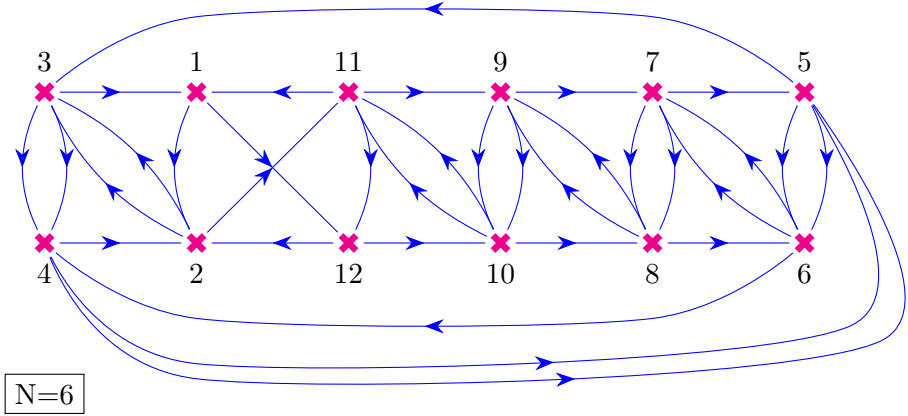
\begin{figure}[!h]
\begin{center}
\begin{tikzpicture}
\tikzset{styleArrow/.style={blue,postaction={decorate},decoration={markings,mark=at position 0.5 with {\arrow{Stealth[scale=1.5]}}}}}
\tikzmath{\Ni=6;
\N=\Ni;}

\draw[styleArrow,bend right] (0,2) to (0,0);
\draw[styleArrow,bend left] (2*\N-2,2)to (2*\N-2,0);
\draw[styleArrow] plot[smooth] coordinates {(2*\N-2,0) (2*\N-4,-1) (2,-1) (0,0) };
\draw[styleArrow] plot[smooth] coordinates { (2*\N-2,2) (2*\N-4,3) (2,3) (0,2)};
\draw[styleArrow] plot[smooth] coordinates {(0,0) (2,-1.8) (2*\N-2+1,-1.3) (2*\N-2,2)};
\draw[styleArrow] plot[smooth] coordinates {(0,0) (2,-1.6) (2*\N-2+0.6,-1.1) (2*\N-2,2)};

\draw(0,0) pic{quiver2NN};
\draw(2,0) pic{quiver3NN};
\tikzmath{
for \i in {2,...,\N-2}{{\draw(2*\i,0) pic{quiver2NN};};};
\k=1; int \i;
for \i in {1,...,2*\N}{\x=Mod(-\i+4,2*\N)-Mod(\i,2);\y=2.8*Mod(\i,2)-0.4; {\draw(\x,\y) node {\i};};};
}

\node[draw] at (0,-2){N=\Ni};

\end{tikzpicture}

\end{center}
\caption{Quiver for $L^{1,2N-1,2}$ system.}
\label{fig:quiverL}
\end{figure}

\noindent
The quiver automorphism is now given by
\eq{
T=\lb(2N)(2N-1)\ldots 4321\rb\cdot\mu_1\,,
}
and for the cluster algebra coefficients, belonging to the same semifield as before, one can take
\be
y_k =q^\frac1{2N-1},\quad k\neq 1,2N\,,\\
y_1 =z_0 q^{\frac{n}{2N-1}}\,,\quad
y_{2N} =z_0^{-1} q^{\frac{1-n}{2N-1}}\,.
\ee
This automorphism acts on the tau-functions as
\eqs{
T(\tau_i)&=\tau_{i+1},\ \ \ \ i=1,\ldots,2N-1,
\\
T(\tau_{2N})&=\frac{\tau_{2N}\tau_2+ z_0 q^{\frac{n}{2N-1}}\tau_{2N-1}\tau_3}{\tau_1}\,.
}
Defining $T^l(\tau_m)=\tau_{m+l}=\tau( q^{\frac{m+n+l}{2N-1}} z_0)=\tau\lb  q^{\frac{m+l}{2N-1}} z\rb$, we get
\eq{
\tau\lb q^{\frac{N}{2N-1}} z\rb \tau\lb q^{-\frac{N}{2N-1}} z\rb=\tau\lb  q^{\frac{N-1}{2N-1}} z\rb\tau\lb q^{-\frac{N-1}{2N-1}} z\rb
+z\cdot \tau\lb q^{\frac{N-2}{2N-1}} z\rb \tau\lb q^{\frac{2-N}{2N-1}} z\rb\,,
\label{eq:Labc_equation}
}
the Hirota bilinear equation for this case.

\begin{Remark} It is easy to notice an accidental coincidence of two Newton polygons from different families: $L^{1,3,2}=Y^{2,1}$, so there is a question about relation of corresponding difference equations. Let us rewrite \eqref{eq:Labc_equation} as
\eq{
\tau_{l+2}\tau_{l-2}=\tau_{l+1}\tau_{l-1}+z_0q^{l/3}\tau_l^2
}
and define $\tau_l=\hat\tau_{-l}z_0^{\frac{l^2}8}q^{\frac{l^3}{72}}$ so that
\eq{
\hat\tau_{l+2}\hat\tau_{l-2}=z_0^{-\frac34}q^{\frac94l}\hat\tau_{l+1}\hat\tau_{l-1}+\hat\tau_l^2}\,.
After rescaling the variables $z_0=\hat z_0^{-2/3}$, $q=\hat q^{1/9}$  we obtain
\eq{
\hat\tau_{l+2}\hat\tau_{l-2}=\hat z_0^{1/2}\hat q^{l/4}\hat\tau_{l+1}\hat\tau_{l-1}+\hat\tau_l^2
}
which actually coincides with \eqref{eq:toda_deaut2}.
\end{Remark}

\section{Solutions of Toda systems} \label{sec:sol}

\subsection{Solutions and Nekrasov functions}
\label{ssec:Nek}

In order to write the solutions let us prepare first some special functions. Recall that infinite multiple $q$-deformed Pochhammer symbol is defined by
\begin{equation}\label{eq:Poch}
(x;t_1,\ldots t_N)_{\infty}=\prod_{i_1,\ldots i_N=0}^{\infty}\left(1-x\prod_{k=1}^Nt_k^{i_k}\right)
=\exp\left(-\sum_{m=1}^{\infty}\frac{x^m}m\prod_{k=1}^N\frac1{1-t_k^m}\right).
\end{equation}
The product exists if all $|t_k|<1$, but the exponent is meaningfull in the region $|t_k|\neq 1$, so that the function $(x;t_1,\ldots t_N)_{\infty}$, defined by the second expression, satisfies
\be
(x;t_1^{-1},t_2,\ldots ,t_N)_{\infty}=(xt_1;t_1,t_2,\ldots ,t_N)^{-1}_{\infty} \label{qtrans}\,.
\ee
By $Z^{N,k}(\vec{u};q_1,q_2|z)$ we denote Nekrasov partition function for the $5d$ pure $SU(N)$ gauge theory with the Chern-Simons level $k$. Here $\vec{u}=(u_1,\ldots,u_N)$ with $u_j=e^{R a_j}$,  $q_1=e^{R\epsilon_1}, q_2=e^{R\epsilon_2}$,  where
$\vec{a}$ is the condensate of scalar field, $\epsilon_1, \epsilon_2$ are parameters of $\Omega$-background, and $R$ is the radius of 5-th compact dimension. For $SU(N)$ gauge group $\sum_j a_j=0$, and therefore $\prod_j u_j=1$.

The formulas for Nekrasov functions are given e.g. in \cite{GNY} (following \cite{IKP02,Ta}), and have the form
\begin{equation} \label{eq:ZNekkCs}
Z^{N,k}(\vec{u};q_1,q_2|z)=Z^{N,k}_{\rm cl}(\vec{u};q_1,q_2|z)\cdot Z^{N}_{\rm 1-loop}(\vec{u};q_1,q_2)\cdot Z^{N,k}_{\rm inst}(\vec{u};q_1,q_2|z)\,,
\end{equation}
where
\eqs{
\label{eq:ZNekkCs:inst}
Z^{N,k}_{\rm cl}(\vec{u};q_1,q_2|z)&=\exp\left(\log z\frac{\sum\left(\log u_i \right)^2}{-2 \log q_1 \log q_2}+k\frac{\sum\left(\log u_i \right)^3}{-6\log q_1 \log q_2}\right),\\
Z^{N}_{\rm 1-loop}(\vec{u};q_1,q_2)&=\prod_{1\leq i\neq j\leq N}({u_i}/{u_j};q_1,q_2)_\infty\,,\\
Z^{N,k}_{\rm inst}(\vec{u};q_1,q_2|z)&=\sum_{\vec{\lambda}} \frac{z^{|\vec{\lambda}| } \prod_{i=1}^N (\mathsf{T}_{\lambda^{(i)}}(u_i;q_1,q_2))^k}{\prod_{i,j=1}^N \mathsf N_{\lambda^{(i)},\lambda^{(j)}}(u_i/u_j;q_1,q_2)}\,,\\
\vec{\lambda}&=(\lambda^{(1)},\ldots,\lambda^{(N)}),\;\; |\vec{\lambda}|=\sum |\lambda^{(i)}|,\;\;
 |\lambda|=\sum \lambda_j,\;\\
\mathsf N_{\lambda,\mu}(u,q_1,q_2)
&=\prod_{s\in \lambda}
(1-u q_2^{-a_\mu(s)-1}q_1^{\ell_\lambda(s)}) \cdot
\prod_{s \in \mu}
(1-u q_2^{a_\lambda(s)}q_1^{-\ell_\mu(s)-1})\,,\\
\mathsf{T}_\lambda(u;q_1,q_2)&=u^{|\lambda|}q_1^{\frac12(\lVert\lambda'\rVert-|\lambda'|)}
q_2^{\frac12(\lVert\lambda\rVert-|\lambda|)}=\prod_{(i,j)\in \lambda} uq_1^{i-1}q_2^{j-1},  \quad \lVert\lambda\rVert=\sum \lambda_j^2\,.
}
Here $\{\lambda^{(i)}|i=1,\ldots,N\}$ are $N$-tuples of partitions (or corresponding Young diagrams), $\lambda'$ denotes the partition transposed to $\lambda$, $a_{\lambda}(s), l_{\lambda}(s)$ denote the lengths of arms and legs for the box $s$ in the Young diagram $\lambda$~\footnote{Our conventions for the Chern-Simons term differ from those of \cite{GNY} by sign of the level: $k$ here is $-k$ in \cite{GNY}.}.
Below we consider only the case $q_1=q$, $q_2=q^{-1}$, the infinite products $(x;q,q^{-1})_\infty$ should be understood, using the formula in the r.h.s. of \eqref{eq:Poch} or \eqref{qtrans}. The case of arbitrary $q_1$ and $q_2$ corresponds to the quantum deautonomized system with the multiplicative quantum parameter $p=q_1q_2$, see \cite{BGM}. We are going to return to the higher-rank quantum cluster systems elsewhere.

\begin{Remark}\label{rem:0kN}
	If the Chern-Simons level is  restricted to the region $0\leq k \leq N$, then the Nekrasov functions $Z^{N,k}(\vec{u};q,t|z)$ are equivalent to the topological string amplitudes on 3d toric Calabi-Yau, defined by the polygons $Y^{N,k}$ from Fig.~\ref{fig:YNk} (see also Theorem~\ref{th:polygons}). These toric Calabi-Yau correspond to resolved $A_{N-1}$ singularity, fibered over $\mathbb{CP}^1$, or minimal resolution of the $Y^{N,k}$ singularity.
	
	We do not discuss here the convergence of the series expansions for $Z^{N,k}_{\rm inst}$. It has been proven for $N=2,k=0$ (and $q_1=q$, $q_2=q^{-1}$) in \cite{BS:2016:1}, this proof works for any $N$ and $k=0$, other regions for the parameters $q_1,q_2$ were studied in \cite{Felder:2017}. Numeric experiments suggest that $Z^{N,k}_{\rm inst}$ converge for $-N\leq k\leq N$, but diverge for $|k|>N$.
	
	Below we restrict ourselves to the region $0\leq k \leq N$, as in Sect.~\ref{sec:Todas} where it comes from the fact that for $k>N$ the Newton polygons for $Y^{N,k}$ (see Fig.~\ref{fig:YNk}) become non-convex. Note also, that the difference equations \eqref{eq:toda_deaut3} for $k>N$ have ``higher order''.
\end{Remark}

We identify the root lattice of $A_{N-1}$ with the set $Q_{N-1}=\{(n_1,\ldots,n_N)\in \mathbb{Z}^{N}|\sum_i n_i=0\}$. The fundamental weights are then
\be
\omega_j=\underbrace{\left(\frac{N-j}{N},\ldots,\frac{N-j}{N}\right.}_{j\ \text{times}},\left.\frac{-j}{N}, \ldots, \frac{-j}{N}\right)\in\mathbb{Q}^N, \text{ where } 1 \leq j\le N-1,,
\ee
and below we also use notation $\omega_0=0$. For any $\vec{\Lambda}=(\Lambda_1,\ldots,\Lambda_N)$ denote $s^{\Lambda}=s_1^{\Lambda_1}\cdot\ldots\cdot s_N^{\Lambda_N}$ and $q^{\vec{\Lambda}}=(q^{\Lambda_1},\ldots,q^{\Lambda_N})$. If $\sum_j \Lambda_j=0$, then $s^{\Lambda}$ is invariant under the symmetry $(s_1,\ldots,s_n)\rightarrow (ts_1,\ldots,ts_n)$, so the actual number of dual parameters $\{s_j\}$ is $N-1$ and coincides with the number of parameters $\{u_j\}$ constrained by $\prod_j u_j=1$.

Define the Fourier transformed Nekrasov functions by
\begin{equation}\label{eq:TauGen}
\mathcal{T}^{N,k}_j(\vec{u},\vec{s};q|z)=\sum_{\vec{\Lambda} \in Q_{N-1}+\omega_j} s^\Lambda Z^{N,k}(\vec{u}q^{\vec{\Lambda}};q^{-1},q|z),\ \ \ \ \ j\in \mathbb{Z}/N\mathbb{Z}\,.
\end{equation}
Sometimes we omit $\vec{u},\vec{s}$ (parameters of solution) and $q$ (lattice step) below for brevity.

\begin{conj}\label{conj:bilinTau}
	The functions \eqref{eq:TauGen} satisfy the bilinear relations
\begin{equation}\label{eq:Tau:bilin:YNk}
\mathcal{T}^{N,k}_j(qz)	 \mathcal{T}^{N,k}_j(q^{-1}z)=\mathcal{T}^{N,k}_j(z)^2-z^{1/N}\mathcal{T}^{N,k}_{j+1}(q^{k/N}z)\mathcal{T}^{N,k}_{j-1}(q^{-k/N}z)\,.
\end{equation}
\end{conj}
In other words, the Fourier transformed Nekrasov functions $\{\mathcal{T}^{N,k}_j\}$ are the tau-functions, solving the bilinear relations \eqref{eq:toda_deaut2} (up to redefinition $z^{1/N}\leftrightarrow -z^{1/N}$).

We have checked this conjecture by expansion in $z$ for many cases with $N\leq 3$, this relation for $N=2,  k=0$ was already proposed in \cite{BS:2016:1}. Note that at the level of formal series in $z$ the relation  \eqref{eq:Tau:bilin:YNk} holds for any value of $k$, though we need them only for $0\leq k \leq N$, see Remark~\ref{rem:0kN}. These relations resemble the blow-up equations, conjectured in \cite{GNY} and proven in \cite{NY}.

\begin{Example}
	For $N=1$ there is no sum in \rf{eq:TauGen}, and tau-functions just coincide with Nekrasov functions themselves, having the form of the products
\be
\label{TU1}
\mathcal{T}^{1,0}(z)=Z^{1,0}(z)=(qz;q,q)_\infty=\prod_{n=1}^\infty(1-zq^{n})^n\,,
\\
\mathcal{T}^{1,1}(z)=Z^{1,1}(z)=\frac1{(-qz;q,q)_\infty}=\prod_{n=1}^\infty(1+zq^{n})^{-n}\,,
\ee
which trivially depend on omitted index $j$ and parameters $u,s$. These formulas fit the bilinear equation \eqref{eq:Tau:bilin:YNk}, which here turns into
\begin{equation}
\label{EU1}
\mathcal{T}^{1,0}(qz)\mathcal{T}^{1,0}(q^{-1}z)=(1-z)\mathcal{T}^{1,0}(z)^2,\quad (1+z)\mathcal{T}^{1,1}(qz)\mathcal{T}^{1,1}(q^{-1}z)=\mathcal{T}^{1,1}(z)^2.
\end{equation}
\end{Example}

\begin{Example}
Let $k=0$, i.e. consider the case of standard relativistic $N$-particle Toda integrable system and pure 5d $SU(N)$ gauge theory (without Chern-Simons term). The bilinear relations \eqref{eq:Tau:bilin:YNk} take the form
\begin{equation}\label{eq:Tau:bilin:YN0}
\mathcal{T}^{N,0}_j(qz)\mathcal{T}^{N,0}_j(q^{-1}z)=\mathcal{T}^{N,0}_j(z)^2-z^{{1}/{N}}\mathcal{T}^{N,0}_{j+1}(z)\mathcal{T}^{N,0}_{j-1}(z)\,.
\end{equation}
There is a special solution to \eqref{eq:Tau:bilin:YN0}
\begin{equation}\label{eq:TauN0:twisted}
\mathcal{T}^{N,0}_0(z)=\ldots=\mathcal{T}^{N,0}_{N-1}(q)=(q^{1/N}z^{1/N};q^{1/N},q^{1/N})_\infty
=\prod_{n=1}^\infty\left(1-z^{1/N}q^{n/N}\right)^n\,,
\end{equation}
basically obtained from the first expressions in \rf{TU1} by substitution $z\to z^{1/N}$, $q\to q^{1/N}$. This solution corresponds to the $q$-deformed twisted fields in the intermediate channel, see formula \eqref{eq:twist_4d} below and related discussion. In this case charge in the intermediate channel equals to $\vec u=(q^{\frac{1-N}{2N}},q^{\frac{3-N}{2N}},\ldots,q^{\frac{N-1}{2N}})$.

In the opposite case $k=N$ bilinear equations \eqref{eq:Tau:bilin:YNk} have the form \begin{equation}\label{eq:Tau:bilin:YNN}
\mathcal{T}^{N,N}_j(qz)\mathcal{T}^{N,N}_j(q^{-1}z)=\mathcal{T}^{N,N}_j(z)^2-z^{{1}/{N}}\mathcal{T}^{N,N}_{j+1}(qz)\mathcal{T}^{N,N}_{j-1}(q^{-1}z)\,.
\end{equation}
These equations also have special solution
\begin{equation}
\mathcal{T}^{N,N}_0(z)=\ldots=\mathcal{T}^{N,N}_{N-1}(q)=\frac1{(-q^{1/N}z^{1/N};q^{1/N},q^{1/N})_\infty}
=\prod_{n=1}^\infty\left(1+z^{1/N}q^{n/N}\right)^{-n}\,,
\end{equation}
corresponding to $q$-deformed twist fields in the intermediate channel, see~\cite{MNTT}. Moreover, in the case of arbitrary parameters $s$ and special parameters $u$, corresponding to twist fields, relation between the equation \eqref{eq:Tau:bilin:YNN} and formula \eqref{eq:TauGen} was shown in the paper~\cite{Takasaki09}.

\end{Example}
\begin{Example}
Let $N=2$, $k=2$. Then, since $\mathcal{T}^{2,2}_{j+1}=\mathcal{T}^{2,2}_{j-1}$ it follows from \eqref{eq:Tau:bilin:YNk} that
\be
\mathcal{T}^{2,2}_j(qz)	 \mathcal{T}^{2,2}_j(q^{-1}z)=\mathcal{T}^{2,2}_j(z)^2-z^{1/2}\mathcal{T}^{2,2}_{j+1}(qz)
\mathcal{T}^{2,2}_{j+1}(q^{-1}z) =
\\ =\mathcal{T}^{2,2}_j(z)^2-z^{1/2}\mathcal{T}^{2,2}_{j+1}(z)^2+z\mathcal{T}^{2,2}_j(qz)	 \mathcal{T}^{2,2}_j(q^{-1}z)\,.
\ee
This equation is actually equivalent to that of $k=0$ example \eqref{eq:Tau:bilin:YN0} for $N=2$ after the substitution
\begin{equation}
\mathcal{T}^{2,0}_j(z)=(qz;q,q)_\infty \mathcal{T}^{2,2}_j(z)\,,
\end{equation}	
which follows from relation for Nekrasov partiton functions $Z^{2,0}(z)=(qz;q,q)_\infty Z^{2,2}(z)$. Equivalence of the $Y^{2,2}$ and $Y^{2,0}$ geometries is certainly well-known \cite{IKP02}, note also that the corresponding quivers coincide.
\end{Example}
\begin{Example}
	Let now $N=2$ and $k=1$, denote for brevity
\[\mathcal{T}=\mathcal{T}^{2,1}_0(\vec{u},\vec{s}|z),\quad \overline{\mathcal{T}}=\mathcal{T}^{2,1}_0(q^{\omega_1}\vec{u},\vec{s}|q^{1/2}z)=
\mathcal{T}^{2,1}_1(\vec{u},\vec{s}|q^{1/2}z),\quad \underline{\mathcal{T}}=\mathcal{T}^{2,1}_0(q^{-\omega_1}\vec{u},\vec{s}|q^{-1/2}z)\]\,
	then by induction
\[\overline{\overline{\mathcal{T}}}=\mathcal{T}^{2,1}_1(q^{\omega_1}\vec{u},\vec{s}|qz)=\mathcal{T}^{2,1}_0(\vec{u},\vec{s}|qz)\]
	and equation \eqref{eq:Tau:bilin:YNk} can be rewritten in the form
	\begin{equation} \overline{\overline{\mathcal{T}}}\underline{\underline{\mathcal{T}}}=
\mathcal{T}^2-z^{1/2}\overline{\mathcal{T}}\underline{\mathcal{T}}\,.
	\end{equation}
	This is bilinear form of the $q$-difference Painlev\'e equations of the surface type $A_7^{(1)}$, relation between the Newton polygon $Y^{2,1]}$ and this Painlev\'e equation was proposed in \cite{BGM}. In order to compare with the standard form of this equation (see e.g. \cite[eq. (2.44)]{SakaiLax}), let $g=z^{1/2}\overline{\mathcal{T}}\underline{\mathcal{T}}\mathcal{T}^{-2}$, then for the function $g$ one gets
	\begin{equation}
	\overline{g}g^2\underline{g}=z^{-2}(g-1)\,.
	\end{equation}
Note also, that function $g$ is nothing but the $x_{(n,m)}^{***}$-varible, used in Sect.~\ref{ssec:bezout}.
\end{Example}

\subsection{Solutions in the autonomous limit}
\label{ssec:autlim}

Let us now turn to the solutions of the autonomous versions of equations \eqref{eq:toda_deaut1}.
One can certainly derive these solutions as a $q\to 1$ limit of generic non-autonomous solutions, given by \eqref{eq:TauGen}.

The limiting procedure looks as follows (for simplicity we consider the case $k=0$). Denote $\epsilon=\epsilon_1=-\epsilon_2$, $R=1$, then $q=e^{\epsilon}$, recall also that $u_i=e^{\epsilon a_i}$ and introduce similarly $s_i=e^{\eta_i/\epsilon}$. By \cite{NO:2003} we have the following limiting behavior at $\epsilon\rightarrow 0$:
\begin{equation}
\log Z^{N,0}(\vec{u};q^{-1},q|z)\ \stackreb{\epsilon\to 0}{ =}\  \frac{1}{\epsilon^2}\mathcal{F}^{N,0}(\vec{a},z)+\mathcal{F}_1^{N,0}(\vec{a},z) + o(\epsilon)\,.
\end{equation}
The computation of tau-function \eqref{eq:TauGen} should be done in this case by a kind of improved saddle point
approximation for the Fourier series~\footnote{The same technique was applied to matrix models, when studying the breakdown of genus expansion for the multi-cut solutions (see \cite{BDE} for details), where the Fourier series appeared as a result of summation over the eigenvalues filling fractions for different cuts, with analogs of $\{a_i\}$ being the filling fractions themselves.}. Consider \eqref{eq:TauGen} at $\epsilon\rightarrow 0$
\eq{
\mc T^{N,0}_j(z)=\sum_{\vec\Lambda\in Q_{N-1}+\omega_j}\exp\left(\frac{(\vec\eta,\vec\Lambda)}\epsilon+\frac1{\epsilon^2}\mc F^{N,0}(\vec a+\epsilon\vec\Lambda,z)+\mc F_1^{N,0}(\vec a+\epsilon\vec\Lambda,z)+o(\epsilon)\right)\,.
}
and, first, find the point $\vec a_\ast=\vec a+\epsilon\vec\Lambda_\ast$, where exponent has maximal value:
\eq{
\eta_i+\frac{\d \mc F^{N,0}(\vec a_*,z)}{\d a_{\ast i}}=0\,.
\label{eq:saddle}
}
Expanding into Taylor series around this point, one gets
\eq{
\mc T^{N,0}_j(z)=\exp\left(\frac{(\vec\eta,\vec a_\ast-\vec a)}{\epsilon^2}+\frac1{\epsilon^2}\mc F^{N,0}(\vec a_\ast,z)+\mc F_1^{N,0}(\vec a_\ast,z)+o(\epsilon)\right)\times
\\
\times
\sum_{\vec\Lambda\in Q_{N-1}}\exp\left(\frac12\frac{\d^2\mc F^{N,0}(\vec a_\ast,z)}{\d a_{\ast i}\d a_{\ast i'}}\left(\Lambda_i+\frac{a_i-a_{\ast i}}{\epsilon}+\omega_j\right)\left(\Lambda_{i'}+\frac{a_{i'}-a_{\ast i'}}{\epsilon}+\omega_j\right)+o(\epsilon)\right)\,.
}
Now we neglect $o(\epsilon)$ terms and rewrite this using the Poisson summation formula. The result can be written as a sum over dual lattice, namely the the weight lattice $P_{N-1}=\{(L_1,\ldots,L_N)\in\mathbb{Q}^N|\sum L_i=0, L_{i}-L_{i-1}\in \mathbb{Z}\}$ and has the form
\eq{
\mc T^{N,0}_j(z)=\exp\left(\frac{(\vec\eta,\vec a_\ast-\vec a)}{\epsilon^2}+\frac1{\epsilon^2}\mc F^{N,0}(\vec a_\ast,z)+\mc F_1^{N,0}(\vec a_\ast,z)\right)
\det\left(-\frac1{2\pi}\frac{\d^2\mc F^{N,0}(\vec a_\ast,z)}{\d a_{\ast i}\d a_{\ast i'}}\right)^{-1/2}
\times\\\times
\sum_{\vec L\in P_{N-1}}\exp\left(2\pi^2\left[\frac{\d^2\mc F^{N,0}(\vec a_\ast,z)}{\d a_{\ast}\otimes\d a_{\ast}}\right]^{-1}_{ii'}L_iL_k+
2\pi i\left(\frac{\vec a-\vec a_{\ast}}{\epsilon}+\omega_j,\vec L\right)\right)\,.
}
Consider now the $z$-dependence of this tau function, and substitute $z=z_0q^m=z_0e^{\epsilon m}$. The value $a_{\ast}$ defined by \eqref{eq:saddle} becomes dependent on $m$ (we keep $\vec{\eta}$ constant) and one gets
\eq{\label{eq:a*m-a*0}
(a_{\ast i'}(m)-a_{\ast i'}(0))\frac{\d^2\mc F^{N,0}(\vec a_\ast,z)}{\d a_{\ast i'}\d a_{\ast i}}=-\epsilon mz_0\frac{\d^2\mc F^{N,0}(\vec a_\ast,z)}{\d z_0 \d a_{\ast i}}\,.
}
This can be rewritten as $\vec a_\ast(m)=\vec a_\ast(0) + \epsilon m U$ for certain $U\in \mathbb{C}^N$ defined by \eqref{eq:a*m-a*0}. Introducing also $Z\in \mathbb{C}^N$ as $\epsilon Z=\vec a-\vec a_\ast(0)$ we get finally
\eq{
\mc T^{N,0}_j\left(z_0q^m\right)=(\ldots)\cdot
e^{\beta N^2 m^2}\cdot\sum_{\vec L\in P_{N-1}} \exp\left(2\pi^2\left[\frac{\d^2\mc F^{N,0}(\vec a_\ast,z)}{\d a_{\ast}\otimes\d a_{\ast}}\right]^{-1}_{ii'}L_iL_{i'}+
2\pi i\left(Z+j\omega_1+mU,\vec L\right)\right)\,,
}
where $\beta$ is a certain combination of $\mc F$ derivatives, and we omit the exponent of linear in $m$ function, which can be removed by gauge transformation, see \eqref{eq:gauge} below. Here we also used $\vec\omega_j-j\vec\omega_1\in Q_{N-1}$.

Therefore, basically we get a  theta-function on Jacobian of genus $g=N-1$ curve. Moduli of this curve are locally parameterized by $z$ and the vector $\vec{\eta}$ which can be identified with ``dual Seiberg-Witten periods'' $\vec{a}_D$. The period matrix is equal to $\dfrac{\partial^2 \mathcal{F}^{N,0}(\vec{a}_\ast,z)}{\partial a_{_\ast i} \partial a_{_\ast i'}}$, this matrix coincides with period matrix of the curve with the Newton polygon $Y^{N,0}$ (see review \cite{M99} and references therein). Such curves are spectral curves $\mathcal{C}^{N,0}$ of $N$-particle relativistic Toda chain. The vector in Jacobian $J(\mathcal{C}^{N,0})$ depends linearly on $j$  and $m$.

We postpone general discussion of this issue, which should be valid for the limit of topological string partition function, constructed for any convex Newton polygon.
We believe that solutions of autonomous discrete flows for any Newton polygon $\Delta$ can be given in terms of the theta functions, see \cite{F}.
Here we just present explicit solutions for all our Toda family curves, using the Fay trisecant identity.

\subsubsection{Hirota bilinear equation}

As we have seen in Sect.~\ref{sec:Todas}, any discrete integrable system of the Toda family can be obtained as a reduction of Hirota difference equation (see e.g.~\cite{Zabrodin_Hir})
\eq{
T_{n,m+1,p}T_{n,m-1,p}=B \cdot T_{n,m,p+1}T_{n,m,p-1}+ C\cdot T_{n+1,m,p}T_{n-1,m,p}\,.
\label{eq:Hirota_general}
}
It is well-known that Hirota equations can be solved using the Fay trisecant identity \cite[eq (45)]{Fay}:
\eqs{
&E(\tilde{x},\tilde{v})E(\tilde{u},\tilde{y})\Theta\lb Z+\tilde{\mc A}(\tilde{x})-\tilde{\mc A}(\tilde{u})\rb\Theta\lb Z+\tilde{\mc A}(\tilde{y})-\tilde{\mc A}(\tilde{v})\rb=\\
=&E(\tilde{x},\tilde{y})E(\tilde{u},\tilde{v})\Theta\lb Z\rb\Theta\lb Z+\tilde{\mc A}(\tilde{x})+\tilde{\mc A}(\tilde{y})-\tilde{\mc A}(\tilde{u})-\tilde{\mc A}(\tilde{v})\rb+\\+
&E(\tilde{x},\tilde{u})E(\tilde{v},\tilde{y})\Theta\lb Z+\tilde{\mc A}(\tilde{x})-\tilde{\mc A}(\tilde{v})\rb\Theta\lb Z+\tilde{\mc A}(\tilde{y})-\tilde{\mc A}(\tilde{u})\rb\,,
\label{Fay_gen}
}
where $\tilde{x}, \tilde{y}, \tilde{u}, \tilde{v} \in \tilde{\mathcal{C}}$ are four points on a universal cover of a curve $\mathcal{C}$ of genus $g$, $E(\tilde{x},\tilde{y})$ is the Prime form, $\mc A: \mathcal{C}\mapsto J(\mathcal{C})$ is the Abel map, $\tilde{\mc A}: \tilde{\mathcal{C}}\mapsto \mathbb{C}^g$ its cover, $\Theta$ is theta function corresponding to $\mathcal{C}$, and $Z$ is an arbitrary vector in $\mathbb{C}^g$. Using \rf{Fay_gen} one can write down general solution of \eqref{eq:Hirota_general} in terms of theta-functions:
\eq{
\label{TTh}
T_{n,m,p}=\Theta\lb Z+nV+mU+pW\rb\,,
}
where we have for three $g$-dimensional vectors
\eq{
\label{UVW}
2U=\tilde{\mc A}(\tilde{x})-\tilde{\mc A}(\tilde{y})-\tilde{\mc A}(\tilde{u})+\tilde{\mc A}(\tilde{v})\,,\\
2V=\tilde{\mc A}(\tilde{x})-\tilde{\mc A}(\tilde{y})+\tilde{\mc A}(\tilde{u})-\tilde{\mc A}(\tilde{v})\,,\\
2W=\tilde{\mc A}(\tilde{x})+\tilde{\mc A}(\tilde{y})-\tilde{\mc A}(\tilde{u})-\tilde{\mc A}(\tilde{v})\,,
}
and
\eq{
B=\frac{E(\tilde{x},\tilde{y})E(\tilde{u},\tilde{v})}{E(\tilde{x},\tilde{v})E(\tilde{u},\tilde{y})}\,,\qquad
C=\frac{E(\tilde{x},\tilde{u})E(\tilde{v},\tilde{y})}{E(\tilde{x},\tilde{v})E(\tilde{u},\tilde{y})}
\,.
}
for the coefficients in \rf{eq:Hirota_general}.

Moreover, equations \eqref{eq:Hirota_general} are invariant under the ``gauge transformation'' of the form
\be
\label{eq:gauge}
T_{n,m,p}\mapsto T_{n,m,p} e^{P(n,m,p)}\,,
\ee
where $P(n,m,p) = \sum_{0\leq i,j,k\leq 1}P_{ijk}n^im^jp^k$ is just a multilinear function of three discrete
arguments. Similar transformation
\begin{equation}\label{eq:T_transformation:quadratic}
T_{n,m,p}\mapsto T_{n,m,p} e^{Q(n,m,p)}
\end{equation}
with quadratic function $Q(n,m,p) = \alpha m^2 + \beta n^2 +\gamma p^2$
preserves the structure of equations \eqref{eq:Hirota_general}, but changes the coefficients $B, C$, and we are going to use this freedom below.

\subsubsection{$Y^{N,k}$ system}

The autonomous version of \eqref{eq:toda_deaut1} can be written as
\eq{\label{eq:toda_autNK}
\tau_{n,m+1}\tau_{n,m-1}=\tau_{n,m}^2+ z_0^{1/N}\cdot\tau_{n+1,m}\tau_{n-1,m}\,.
}
This equation does not depend on $p$, in order to obtain it as a reduction of \eqref{eq:Hirota_general} we impose
\eq{
\tau_{n,m}=T_{n,m,p}.
}
According to the ansatz \eqref{eq:T_transformation:quadratic} this can be achieved if $P(n,m,p)$ does not depend on  $p$ and shift by $W$ does not change theta function. The last condition means that $W$ belongs to the lattice of $A$-cycles, we denote this lattice as $\mathbb Z^g_A$. Therefore, for the projections of $\tilde{x}, \tilde{y}, \tilde{u}, \tilde{v}$ on $\mathcal{C}$ we have a relation $\mc A(x)+\mc A(y) - \mc A(u)-\mc A(v) = 0$, which implies that there exists a function of degree two with zeroes at $x$, $y$ and two poles at $u$, $v$, so the curve $\mathcal{C}$ is hyperelliptic. The hyperelliptic involution
$\sigma$ permutes the points from each pair: $u=\sigma(v)$, $y=\sigma(x)$. It is convenient to choose the base
point of the Abel map $\mc A$ at a branch point of $\mathcal{C}$, then
$\mc A(\sigma(w))=-\mc A(w)$ for any $w \in \mathcal{C}$, and $\mc A(x)+\mc A(y) = \mc A(u)+\mc A(v) = 0$. Then one can choose preimages $\tilde{x}, \tilde{y}, \tilde{u}, \tilde{v} \in \tilde{\mathcal{C}}$ such that $W=0$.

To satisfy periodicity condition $\tau_{n+N,m+k}=\tau_{n,m}$ (see \eqref{eq:boundary_condition1}) it is natural  to use \eqref{eq:T_transformation:quadratic} in the form \begin{equation}\label{eq:tau:theta}
    \tau_{n,m}=e^{\beta(nk-mN)^2} \Theta\lb Z+nV+mU\rb
\end{equation}
and impose that theta function does not change under the shift by $N V+k U$. This means that $N V+k U \in \mathbb Z^g_A$, which leads together with \rf{UVW} and $W=0$ to
%\eq{
%k\lb \tilde{\mc A}(\tilde{x})-\tilde{\mc A}(\tilde{u})\rb
%+N\lb \tilde{\mc A}(\tilde{x})-\tilde{\mc A}(\tilde{v}) \rb=0\,,
%}
%or, equivalently, to
\eq{
N\lb\tilde{\mc A}(\tilde{v})-\tilde{\mc A}(\tilde{x})\rb-k\lb \tilde{\mc A}(\tilde{v})-\tilde{\mc A}(\tilde{y})\rb\in \mathbb Z^g_A\,.
\label{eq:reductionNk:tilde}
}
This leads to the relation in the Jacobian
\eq{
N\lb\mc A(v)-\mc A(x)\rb=k\lb \mc A(v)-\mc A(\sigma(x))\rb\,,
\label{eq:reductionNk}
}
recall that $y=\sigma(x)$. Clearly \eqref{eq:reductionNk:tilde} does not follow from \eqref{eq:reductionNk}, but for given $x,y,u,v$ satisfing \eqref{eq:reductionNk} we can find basis of $A$- and $B$-cycles such that l.h.s. of \eqref{eq:reductionNk:tilde} lies in the lattice of $A$-cycles, so we get a periodicty condition.\footnote{Equivalently, on can go from \eqref{eq:reductionNk} to \eqref{eq:reductionNk:tilde} using arbitrary choice of $A$- and $B$-cycles but generic transformation of the form \eqref{eq:T_transformation:quadratic}.}

In order to get solutions to \eqref{eq:toda_autNK} we also choose
\be
\exp(2N^2\beta)=\frac{E(\tilde{x},\tilde{y})E(\tilde{u},\tilde{v})}{E(\tilde{x},\tilde{v})E(\tilde{u},\tilde{y})}, \quad z_0^{1/N}=\frac{E(\tilde{x},\tilde{u})E(\tilde{v},\tilde{y})}{E(\tilde{x},\tilde{v})E(\tilde{u},\tilde{y})}\exp(2\beta(k^2-N^2))\,.
\ee
Now, to solve Hirota equations \eqref{eq:toda_autNK} we have to find solutions to the linear equations \eqref{eq:reductionNk}, using explicit description of the hyperelliptic curves $\mathcal{C}^{N,k}$, defined by equations \eqref{SC} with the Newton polygons $Y^{N,k}$, see Sect.~\ref{ssec:polygons}. We will see that $z_0$ does not depend on the choice of lifts of points $x,y,u,v$ to $\tilde{\mathcal{C}}^{N,k}$.

\paragraph{$Y^{N,0}$ system.} The curve $\mathcal{C}^{N,0}$ is defined by
\eq{
\lambda^{-1}+\mu^N\lambda+\mu^N+c_{N-1}\mu^{N-1}+\ldots+c_0=0\,,
}
so for the zeroes and poles of the functions $\lambda$ an $\mu$ on $\mathcal{C}^{N,0}$ (for generic coefficients in this equation) one gets:
\begin{itemize}
    \item $\lambda$ has $N$-th order zero at the point $(\lambda,\mu)=(0,\infty)$, where $\mu\to \infty$, $\lambda\sim -\mu^{-N}$;
    \item $\lambda$ has $N$-th order pole at $(\lambda,\mu)=(\infty,0)$, where $\mu\to 0$, $\lambda\sim-\mu^{-N}$;
    \item $\mu$ has simple zeros at the point $(\lambda,\mu)=(-c_0^{-1},0)$, and the point $(\lambda,\mu)=(\infty,0)$, where $\mu\to 0$, $\lambda\sim-c_0\mu^{-N}$;
    \item $\mu$ has simple poles at the point $(\lambda,\mu)=(0,\infty)$, where  $\mu\to \infty$, $\lambda\sim-\mu^{-N}$,  and at the point $(\lambda,\mu)=(-1,\infty)$.
\end{itemize}
Hence, one gets for the divisors of these functions on $\mathcal{C}^{N,0}$
\eq{
(\lambda)=N(0,\infty)-N(\infty,0),\quad (\mu)=(-c_0^{-1},0)+(\infty,0)-(0,\infty)-(-1,\infty)\,,
}
which means that for $k=0$ equation \eqref{eq:reductionNk} has an obvious solution
\eq{
x=(0,\infty),\quad v=(\infty,0)\,,
}
or vice versa.

\paragraph{$Y^{N,k}$ systems for $0< k < N$.} The curve $\mathcal{C}^{N,k}$ is defined here by
\eq{
\lambda^{-1}+\mu^{N-k}\lambda+\mu^N+c_{N-1}\mu^{N-1}+\ldots+c_0=0
}
so that for the zeroes and poles of functions $\lambda$ an $\mu$ one has:
\begin{itemize}
\item $\lambda$ has $N$-th order zero at the point $(\lambda,\mu)=(0,\infty)$, where $\mu\to \infty$, $\lambda\sim -\mu^{-N}$;
\item $\lambda$ has pole of order $N-k$ at the point $(\infty,0)$, where $\mu\to 0$, $\lambda\sim-c_0\mu^{k-N}$, and an extra pole of order $k$, at the point $(\lambda,\mu)=(\infty,\infty)$, where $\mu\to\infty$, $\lambda\sim-\mu^k$;
\item $\mu$ has zeros at the points $(\lambda,\mu)=(-c_0^{-1},0)$, and $(\lambda,\mu)=(\infty,0)$, where now $\mu\to 0$, $\lambda\sim-c_0\mu^{k-N}$;
\item $\mu$ has poles at the points $(\lambda,\mu)=(0,\infty)$, where $\mu\to \infty$, $\lambda\sim -\mu^{-N}$, and $(\lambda,\mu)=(\infty,\infty)$, where $\mu\to\infty$, $\lambda\sim-\mu^k$.
\end{itemize}
Hence, for the divisors on $\mathcal{C}^{N,k}$ one gets
\eq{
(\lambda)=N(0,\infty)-(N-k)(\infty,0)-k(\infty,\infty),\quad (\mu)=(-c_0^{-1},0)+(\infty,0)-(0,\infty)-(\infty,\infty)\,,
}
and therefore \eqref{eq:reductionNk} in this case is solved by
\eq{
v=(\infty,0),\quad x=(0,\infty),\quad \sigma(x)=(\infty,\infty)\,,
}
where the last equality follows from the fact that $\sigma$ permutes two poles of $\mu$.

Now one can also write down explicit expressions for the functions $\lambda$ and $\mu$:
\eq{
\lambda(p)=C_\lambda\frac{E(\tilde p,\tilde x)^N}{E(\tilde p,\tilde v)^{N-k}E(\tilde p,\tilde y)^k}\,,\qquad\mu(p)=C_\mu\frac{E(\tilde p,\tilde u)E(\tilde p,\tilde v)}{E(\tilde p,\tilde x)E(\tilde p,\tilde y)}\,.
}
where the constants $C_\lambda$ and $C_\mu$ are found, requiring near $x$, $y$, $u$, $v$:
\eqs{
p\to x, \lambda\mu^N\to-1: & \qquad C_\lambda C_\mu^N\frac{E(\tilde x,\tilde v)^kE(\tilde x,\tilde u)^N}{E(\tilde x,\tilde y)^{N+k}}=-1\,,\\
p\to y, \lambda\mu^{-k}\to -1: & \qquad C_\lambda C_\mu^{-k}\frac{E(\tilde y,\tilde x)^{N+k}}{E(\tilde y,\tilde v)^NE(\tilde y,\tilde u)^k}=-1\,,\\
p\to u, \lambda=-c_0^{-1}: & \qquad C_\lambda \frac{E(\tilde u,\tilde x)^N}{E(\tilde u,\tilde v)^{N-k}E(\tilde u,\tilde y)^k}=-c_0^{-1}\,.
}
Solving these equations we get
\eq{
%C_\lambda=(-1)^{N+k+1}\left(\frac{E(y,v)^{N^2}E(y,u)^{kN}}{E(x,v)^{k^2}E(x,u)^{kN}}\right)^{1/(k+N)}E(x,y)^{k-N}\,\\
c_0^{-1}=(-1)^{N}
\frac{(E(\tilde v,\tilde y)E(\tilde x,\tilde u))^{\frac{N^2}{k+N}}}{(E(\tilde x,\tilde y)E(\tilde u,\tilde v))^{N-k}(E(\tilde x,\tilde v)E(\tilde u,\tilde y))^{\frac{k^2}{k+N}}}=(-1)^{N}z_0^{\frac{N}{k+N}}.
}
Therefore $z_0$, which is the coefficient in the Hirota equation \eqref{eq:toda_autNK} is a coefficient in the spectral curve $\mathcal{C}^{N,k}$ equation. In particular $z_0$ does not depend on the choice of the lifts $\tilde{x}, \tilde{y}, \tilde{u}, \tilde{v}$ of the points $x,y,u,v \in \mc{C}^{N,k}$.

\paragraph{$Y^{N,N}$ system.} The curve $\mathcal{C}^{N,N}$ is defined by
\eq{\label{eq:CNN_equation}
\lambda^{-1}+\lambda+\mu^N+c_{N-1}\mu^{N-1}+\ldots+c_0=0\,,
}
so that for the zeroes and poles of functions $\lambda$ an $\mu$ one now has:
\begin{itemize}
\item $\lambda$ has $N$-th order zero at the point $(\lambda,\mu)=(0,\infty)$, where $\mu\to \infty$, $\lambda\sim -\mu^{-N}$;
\item $\lambda$ has $N$-th order pole at the point $(\lambda,\mu)=(\infty,\infty)$, where $\mu\to\infty$, $\lambda\sim-\mu^N$;
\item $\mu$ has two zeros at the points $(\lambda,\mu)=(\lambda_\pm,0)$, where $\{\lambda_\pm\}$ are two roots of quadratic equation $\lambda+\lambda^{-1}+c_0=0$;
\item $\mu$ has poles at the point $(\lambda,\mu)=(0,\infty)$, where $\mu\to \infty$, $\lambda\sim -\mu^{-N}$, and the point $(\lambda,\mu)=(\infty,\infty)$, where $\mu\to\infty$, $\lambda\sim-\mu^N$.
\end{itemize}
Hence, the divisors here are
\eq{
(\lambda)=N(0,\infty)-N(\infty,\infty),\quad (\mu)=(\lambda_+,0)+(\lambda_-,0)-(0,\infty)-(\infty,\infty)\,,
}
and solution of \eqref{eq:reductionNk} for $k=N$ is given by
\eq{
x=(0,\infty),\quad \sigma(x)=(\infty,\infty)\,.
}
Note that in this case there is no condition on the point $v$, so it can be chosen arbitrarily. In other words, the zeroes $\{(\lambda_\pm,0)\}$ of the hyperelliptic co-ordinate $\mu$ are not distinguished on $\mathcal{C}^{N,N}$, since constant shift of $\mu$ preserves the form \eqref{eq:CNN_equation}, but moves the zeroes to another pair of points.

\begin{Example} Consider the case $Y^{2,0}$, when equations can be solved by Jacobi theta functions on torus with periods $(\pi,\pi\uptau)$. They satisfy
the following addition formulas (we follow \cite{Whittaker Watson} in notations)
\eq{
\label{WWadd}
\theta_3^2(0)\theta_3(Z+U)\theta_3(Z-U)=\theta_3(U)^2\theta_3(Z)^2+\theta_1^2(U)\theta_1^2(Z)\,,\\
\theta_3^2(0)\theta_1(Z+U)\theta_1(Z-U)=\theta_3(U)^2\theta_1(Z)^2-\theta_1^2(U)\theta_3^2(Z)\,,
}
where vectors $U$ and $Z$ in this $g=1$ case are just $\mathbb{C}$-numbers.
Introducing therefore two functions, which differ by shift $Z\mapsto Z+V$ with $2V=0$
\begin{equation}
\tau_{0,m}=\left(\frac{\theta_3(0)}{\theta_3(U)}\right)^{m^2}\theta_3(Z+mU),\quad
\tau_{1,m}=e^{i\pi/4}\left(\frac{\theta_3(0)}{\theta_3(U)}\right)^{m^2}\theta_1(Z+mU)\,,
\end{equation}
one gets from \rf{WWadd} the solutions of the following bilinear equations (the particular versions of \eqref{eq:toda_autNK})
\eq{
\tau_{1,m+1}\tau_{1,m-1}=\tau_{1,m}^2+z^{1/2}\tau_{0,m}^2\,,
\\
\tau_{0,m+1}\tau_{0,m-1}=\tau_{0,m}^2+z^{1/2}\tau_{1,m}^2
}
with the coefficient
\be
\label{zTh}
z^{1/4}=e^{-i\pi/4}\frac{\theta_1(U)}{\theta_3(U)}\,.
\ee
One can also find all corresponding $x$-variables:
\eq{
x^{+00}_{0,m}=x_{0,m}=z^{1/2}\frac{\tau_{1,m-1}^2}{\tau_{0,m-1}^2}\,,\quad x^{+00}_{1,m}=x_{1,m}=z^{1/2}\frac{\tau_{0,m-1}^2}{\tau_{1,m-1}^2}\,,\\
x_{0,m}^{\times00}=z^{-1/2}\frac{\tau_{0,m+1}^2}{\tau_{1,m+1}^2}\,,\quad x_{1,m}^{\times00}=z^{-1/2}\frac{\tau_{1,m+1}^2}{\tau_{0,m+1}^2}\,,\\
x_{0,m}^{+11}=\frac{z^{1/2}\tau_{1,m}^4}{\tau_{1,m+1}^2\tau_{0,m-1}^2}\,,\quad
x_{1,m}^{+11}=\frac{z^{1/2}\tau_{0,m}^4}{\tau_{0,m+1}^2\tau_{1,m-1}^2}\,,\\
x_{0,m}^{\times11}=\frac{\tau_{0,m+1}^2\tau_{1,m-1}^2}{z^{1/2}\tau_{1,m}^4}\,,\quad
x_{1,m}^{\times11}=\frac{\tau_{1,m+1}^2\tau_{0,m-1}^2}{z^{1/2}\tau_{0,m}^4}\,.
}
Notice that in this case the quiver contains only the vertices of types $(0,0)$, $(\rm I,I)$ and $(\rm -I,-I)$, and effectively all these $x$-variables are expressed in terms of a single function
\eq{
x_{0,m+1}=G_m=\lb\frac{\theta_1(U)}{\theta_3(U)}\frac{\theta_1(Z+mU)}{\theta_3(Z+mU)}\rb^2\,,
\label{eq:toy_solution}
}
which solves
\eq{
G_{m+1}G_{m-1}=\left(\frac{G_m+z}{G_m+1}\right)^2
\label{eq:toy_equation}
}
with the constant $z$ given by \rf{zTh}.
This discrete equation has an integral of motion, the Hamiltonian:
\eq{
H=G_m^{1/2}(G_{m-1}^{1/2}+G_{m-1}^{-1/2})+G_{m}^{-1/2}(G_{m-1}^{1/2}+zG_{m-1}^{-1/2})\,,
}
whose value on solution \rf{eq:toy_solution} is given by
\eq{
H=\frac{2\theta_3(0)^2}{\theta_2(0)\theta_4(0)}\frac{\theta_2(U)\theta_4(U)}{\theta_3(U)^2}\,.
\label{eq:toy_Hamiltonian}
}

\end{Example}

\subsubsection{$L^{1,2N-1,2}$ system}

The autonomous version of equation \eqref{eq:Labc_equation} can be written as
\eq{
\label{eq:Labc_aut}
\tau_{l+N}\tau_{l-N}=\tau_{l+N-1}\tau_{l-N+1}+
z\cdot\tau_{l+N-2}\tau_{l-N+2}\,,
}
and to get it from \eqref{eq:Hirota_general} we identify three above shifts of $l$ with three shifts of different indices of the $\{\tau_l=T_{n,m,p}\}$-variables, e.g.
\eq{
\label{Ttauid}
\tau_{l\pm N}=T_{n,m\pm1,p},\quad \tau_{l\pm (N-1)}=T_{n,m,p\pm 1},\quad \tau_{l\pm (N-2)}=T_{n\pm 1,m,p}\,.
}
It is convenient then to set
\begin{equation}
    T_{n,m,p}=e^{\left((N-2)n+Nm+(N-1)p\right)^2\beta}\Theta(Z+nV+mU+pW)\,,
\end{equation}
so that \rf{eq:Labc_aut}, \rf{Fay_gen} gives for parameters $\beta$ and $z$
\begin{equation}
    e^{-2(2N-1)\beta}=\frac{E(\tilde x,\tilde y)E(\tilde u,\tilde v)}{E(\tilde x,\tilde v)E(\tilde u,\tilde y)}\,,\quad
z=\frac{E(\tilde x,\tilde u)E(\tilde v,\tilde y)}{E(\tilde x,\tilde v)E(\tilde u,\tilde y)}e^{8(N-1)\beta}\,,
\end{equation}
and identifications \rf{Ttauid} due to $T_{n+1,m+1,p} = T_{n,m,p+2}$ and $T_{n-N,m+N,p}=T_{n,m+2,p}$ constrain the vectors
\eq{
U+V-2W\in\mathbb Z^g_A,\quad N(U-V)-2 U\in\mathbb Z^g_A\,.
}
In terms of four points $\tilde x,\tilde y,\tilde u,\tilde v \in \tilde{\mathcal{C}}^{1,2N-1,2}$ formulas \rf{UVW} imply
\eq{
\tilde{\mc A}(\tilde u)+\tilde{\mc A}(\tilde v)-2\tilde{\mc A}(\tilde y)\in\mathbb Z^g_A\,,\quad (N-1)\lb \tilde{\mc A}(\tilde v)-\tilde{\mc A}(\tilde u) \rb-\tilde{\mc A}(\tilde x)+\tilde{\mc A}(\tilde y)\in\mathbb Z^g_A\,.
\label{eq:reductionL:tilde}
}
These two relations can be first projected to Jacobian:
\eq{
\mc A(u)+\mc A( v)-2\mc A( y)=0\,,\quad (N-1)\lb {\mc A}( v)-{\mc A}( u) \rb-{\mc A}( x)+{\mc A}( y)=0\,.
\label{eq:reductionL}
}
The first relation means that on $\mathcal{C}^{1,2N-1,2}$ there is a function of order 2, the hyperelliptic involution $\sigma$ now acts as $u=\sigma(v)$, $y=\sigma(y)$.
We present below points that solve these equations.
After this, as before, one can choose preimages of these points in $\tilde{\mc C}$ such that the l.h.s of the first relation vanishes. Then we choose such $A$- and $B$-cycles that l.h.s. of the second relation lies in the lattice of $A$-cycles, and thus solve \eqref{eq:reductionL:tilde}.

Now we present solution to the linear equations \eqref{eq:reductionL} for the curve $\mathcal{C}^{1,2N-1,2}$ defined by equation
\eq{
\lambda+\lambda^{-1}\mu^{-1}+c\lambda^{-1}+\mu^{N-1}+\ldots+c_0= \lambda+\lambda^{-1}\mu^{-1}+c\lambda^{-1}
+ P(\mu) = 0\,.
}
Consider again zeroes and poles of the functions $\lambda$ and $\mu$ on $\mathcal{C}^{1,2N-1,2}$:
\begin{itemize}
\item $\lambda$ has $(N-1)$-th order zero at the point $(\lambda,\mu)=(0,\infty)$, where $\mu\to \infty$, $\lambda\sim -c \mu^{-N+1}$, and simple zero at the point $(\lambda,\mu)=(0,-c^{-1})$, where $\lambda\sim -(\mu^{-1}+c)/P(-c^{-1})$;
\item $\lambda$ has $(N-1)$-th order pole at the point $(\lambda,\mu)=(\infty,\infty)$, where $\mu\to\infty$, $\lambda\sim-\mu^{N-1}$, and simple pole at the point $(\lambda,\mu)=(\infty,0)$, where $\lambda\to \infty$, $\mu \sim -\lambda^{-2}$;
\item $\mu$ has second order zero at the point $(\lambda,\mu)=(\infty,0)$, where $\lambda\to \infty$, $\mu \sim -\lambda^{-2}$;
\item $\mu$ has simple poles at the point $(\lambda,\mu)=(0,\infty)$, where $\mu\to \infty$, $\lambda\sim -c\mu^{-N+1}$, and at the point $(\lambda,\mu)=(\infty,\infty)$, where $\mu\to\infty$, $\lambda\sim-\mu^{N-1}$.
\end{itemize}
Therefore, for their divisors one gets
\eq{
(\lambda)= (N-1)(0,\infty)+(0,-c^{-1})-(N-1)(\infty,\infty)-(\infty,0)\,,\quad (\mu)=2(\infty,0)-(0,\infty)-(\infty,\infty)\,,
}
which basically coincide with relations \eqref{eq:reductionL}, being therefore solved by
\eq{
v=(0,\infty)\,,\quad u=(\infty,\infty)\,,\quad y=(\infty,0)\,,\quad x=(0,-c^{-1})\,.
}
This example completes the list of solutions of our autonomous systems.

\subsection{4d limit}

Let us now discuss a particular limit of our bilinear equations and their solutions. In terms of Nekrasov functions it corresponds $R\rightarrow 0$, or, in other words, the limit from 5-dimensional to 4-dimensional supersymmetric gauge theory. In terms of $q$-deformed infinite-dimensional algebras this is the conformal limit, reproducing well-known algebras of two-dimensional conformal theories with extended symmetry. The $q$-difference equations in this limit turn into differential equations.

Assume that $q=\exp R$ (i.e. rescale the background parameters to $\epsilon_1=-\epsilon_2=1$) and $z =R^{2N}\mathsf{z}$, then the $R \rightarrow 0$ limit of equation \eqref{eq:Tau:bilin:YNk} (or, equivalently, \eqref{eq:toda_deaut2}) acquires the form
\begin{equation} \label{eq:4d_equation}
    ( \partial_{\log \mathsf{z}})^2 \log \tau_j  =\mathsf{z}^{1/N}\frac{\tau_{j+1} \tau_{j-1} }{\tau_{j} ^2}\,,
    \quad j\in \mathbb{Z}/N\mathbb{Z}\,,
\end{equation}
where $\tau_j (\mathsf{z})$ denotes the 4d limit of $\mathcal{T}_j(z)$.

The limit of the solution \eqref{eq:TauGen} is straightforward, 5d Nekrasov functions become their 4d versions, the double Pochhammer products are replaced by the Barnes $\mathsf{G}$-functions. The result is a generalization of the Painlev\'e-$\mathrm{III}_3$ tau function \cite{GIL1302} to the higher rank case, or degeneration of the four-point higher-rank isomonodromic  tau-function from \cite{GIso}. Equation \eqref{eq:4d_equation} can be viewed as a Toda tau-form of the corresponding isomonodromy deformation problem, for $N=2$ case this has been described in \cite{BS:2016:2}.

There is a special $j$-independent solution of the system \eqref{eq:4d_equation}, in this case
\eq{\label{eq:twist_4d}
\tau_j({\mathsf z})={\mathsf z}^{\frac{N^2-1}{24N}}e^{N^2 {\mathsf z}^{1/N}}\,,
}
which is a 4d limit of the solution \rf{eq:TauN0:twisted}.
In terms of the definition \eqref{eq:TauGen} this solution corresponds to all $\{s_j=1\}$ and special values of the condensates $\{u_j\}$, which correspond to twist field see \cite{BS:2016:2} for $N=2$ case, in the case of generic $N$ the corresponding twist fields were discussed \cite{GM}, in particular the number $\frac{N^2-1}{24N}$ is the dimension of the twist field corresponding to the Coxeter element for $GL(N)$.

Tau functions \eqref{eq:twist_4d} can be also seen as (a special case of) the dual partition functions from \cite{NO:2003}. Recall that the latter were defind as matrix elements of the form $\tau=\langle0|e^{\frac1N J_{1}}{\mathsf z}^{L_0}s^{J_0}e^{\frac1N J_{-1}}|0\rangle$ --- in terms of \eqref{eq:TauGen} it corresponds already to generic values of $\{s_j\}$-parameters, but still fixed $\{u_j\}$-condensates given by those of the twist fields. The relation between dual partition functions and Toda equation was stated in \cite[eq. (5.26)]{NO:2003}. \footnote{It looks that the non-autonomous factor ${\mathsf z}^{1/N}$ was missed there.}

Recall also another form of the equation \eqref{eq:4d_equation}. Denote $\phi_j=\log \tau_j-\log \tau_{j-1}$ and $r=2N {\mathsf z}^{\frac1{2N}}$, then one gets the radial Toda equation
\eq{\label{eq:4d_Toda_radial}
\frac{d^2 \phi_n}{dr^2}+\frac1r\frac{d \phi_n}{dr}=e^{\phi_{n+1}-\phi_n}-e^{\phi_{n}-\phi_{n-1}}\,.
}
In \cite{BGT1} the Fredholm determinant formula was proposed for the solution of equation \eqref{eq:4d_Toda_radial}, conjecturally this formula corresponds to the limit of \eqref{eq:TauGen} with $\{s_j=1\}$ but generic values of the condensates $\{u_j\}$. Analogous Fredholm determinant for solution of corresponding $q$--difference equation with $s=1$ is written in \cite{BGT2}.

\subsection{Quantization and Poisson bracket}

We believe that quantization of all  $Y^{N,k}$ non-autonomous integrable systems is straightforward along the lines, proposed in \cite{BGM}, and similar to the case $Y^{2,0}$ presented there explicitly. We postpone the detailed discussion of this issue, but our conjecture for the solutions of the quantum cluster system is the following: one should replace the partition functions with $q_1q_2=1$ by partition functions of refined topological string theory, and supply this with extra quantization of the  $\{u_j\}$ and $\{s_j\}$ variables. Namely, the product
$q_1q_2=p \neq 1$ becomes the multiplicative quantum Planck constant~\footnote{Possibly related wth quantum gravitational anomaly, we would like to thank N.~Nekrasov for pointing out this issue.}
so that the parameters of solution $u_js_k=p^{\delta_{jk}}s_ku_j$ are no longer commutative. In the quasiclassical limit one should have the statement about the Poisson bracket: $\{u_j,s_k\}=\delta_{jk}u_js_k$, the traces of this relation can be found in Sect.~\ref{ssec:autlim}.

Surprisingly, but analogous quantization of the solutions in autonomous case looks more tricky. To demonstrate it let us present the cluster Poisson bracket
\eq{
\{G_m,G_{m-1}\}=2G_mG_{m-1}
}
rewritten in terms of parameters of solution \eqref{eq:toy_solution}:
\eqs{
&\{\uptau,Z\}=-\frac{2i}{\pi}\frac{\theta_2(U)\theta_4(U)}{\theta_3(0)^4\theta_3(U)^2}\,,\\
&\{U,Z\}=\frac1{2\pi^2}\frac{\frac{\d}{\d y}\lb\theta_2(U)\theta_4(U)\rb}{\theta_3(0)^4\theta_3(U)^2}\,,\\
&\{\uptau,U\}=0\,.
}
where $\uptau$ is modulus of the elliptic curve.
As a byproduct of this computation one finds the Hamiltonian flow, generated by \eqref{eq:toy_Hamiltonian}:
\eq{
\{H,Z\}=\frac{2\theta_3(0)^2}{\theta_2(0)\theta_4(0)}\frac{\theta_1(U)^2}{\theta_3(U)^2}\,.
}
These brackets are highly non-linear, so the problem of solution of quantum autonomous equations looks to be more difficult than for non-autonomous ones (see \cite[Sect. 4]{BGM} for the example of solution of non-autonomous quantum equation).

\begin{Remark}
There are also papers \cite{HM}, \cite{FHM} where the spectra of quantum cluster integrable systems, like relativistic Toda chain, are studied. The conjectural exact quantization conditions are given in terms of topological strings amplitudes. It would be interesting to find a relation between these results and our formulas like \eqref{eq:TauGen} (and their quantum analogs) for the solutions of the discrete flow equations.  
\end{Remark}

\section{Conclusion}

The results of this paper strongly support the main proposal of \cite{BGM}:
\begin{itemize}
  \item Deautonomization of a cluster integrable system, defined by a Newton polygon $\Delta$, leads to $q$-difference equations of the Painlev\'e type, generated by discrete flows, which can be treated as sequences of quiver mutations.
  \item These equations have the tau-form, which usually can be written as a system of Hirota bilinear difference equations.
  \item The tau-functions (solutions of the tau-form) are given by Fourier series of Nekrasov partition functions of 5d supersymmetric gauge theory (with the Seiberg-Witten curve, determined by 
  the polygon $\Delta$). Equivalently one can express tau-functions in terms of partition functions of the topological string on 3d Calabi-Yau (also determined by the same polygon $\Delta$).
\end{itemize}

We have now tested this conjecture on the cluster integrable systems of the Toda family.
This family corresponds to hyperelliptic curves depending on one Casimir, there is a single discrete flow preserving Hamiltonians of the integrable system. This family corresponds to so called $Y^{N,k}$ and $L^{1,2N-1,2}$ geometries. The 5d supersymmetric gauge theories for  $Y^{N,k}$ family is pure $SU(N)$ theory with Chern-Simons term at level $k$. We have presented solutions of the corresponding $q$-difference equations in terms of Nekrasov functions.

Certainly many important questions remained beyond the scope of this paper. We have discussed very briefly
the autonomous limit of the solutions to $q$-difference equations. In principle this procedure should work for any cluster integrable
system, giving rise simultaneously to the extremal ``Seiberg-Witten'' geometry of the corresponding topological string model, as well as to generic solution of a cluster integrable system in terms of the theta-functions.
Also, and perhaps the most important issue is quantization. The results of this paper, in addition to those of 
\cite{BGM}, suggest that there should be a direct procedure in the \emph{non-autonomous} case, coming from straightforward quantization of a cluster variety, and leading to the solutions in terms of refined partition functions of the topological strings, however with many interesting subtleties in the autonomous limit. We are planning to return to these issues elsewhere.

\section*{Acknowledgements}

We would like to thank  B.~Feigin, M.~Finkelberg, V.~Fock, R.~Gonin, A.~Grassi, A.~Hanany, A.~Ilina, T.~Ishibashi, N.~Nekrasov, M.~Semenyakin, R.-K. Seong,  A.~Shchechkin and Y. Zenkevich for stimulating discussions.

The main results of Sect.~\ref{sec:Todas} have been obtained using support of Russian Science Foundation by the RSF grant No. 16-11-10160,
the work of AM has been also partially supported by RFBR grant 18-01-00460A and
the RFBR/JSPS joint project 17-51-50051.
This work has been also funded by the  Russian Academic Excellence Project
'5-100'.
MB and PG are Young Russian Mathematics award winners and would like to thank its sponsors and jury. We would like also to thank the Galileo Galilei Institute for
Theoretical Physics for the hospitality (and the INFN for partial support)
during the program {\em Supersymmetric Quantum Field Theories in the Non-perturbative Regime}, where this paper has been completed and the results were reported.

%\newpage

\footnotesize

\bigskip
\noindent \textsc{Landau Institute for Theoretical Physics, Chernogolovka, Russia,\\
	Center for Advanced Studies, Skoltech, Moscow, Russia,\\
	Laboratory for Mathematical Physics, NRU HSE, Moscow, Russia,\\
	Institute for Information Transmission Problems, Moscow, Russia,\\
	Independent University of Moscow, Moscow, Russia}

\emph{E-mail}:\,\,\textbf{mbersht@gmail.com}\\

\noindent \textsc{Center for Advanced Studies, Skoltech, Moscow, Russia,\\
%	Department of Mathematics and
Laboratory for Mathematical Physics, NRU HSE, Moscow, Russia,\\
	Bogolyubov Institute for Theoretical Physics, Kyiv, Ukraine}

\emph{E-mail}:\,\,\textbf{pasha.145@gmail.com}\\

\noindent \textsc{Center for Advanced Studies, Skoltech, Moscow, Russia,\\
	Department of Mathematics and Laboratory for Mathematical Physics, NRU HSE, Moscow, Russia,\\
	Institute for Theoretical and Experimental Physics, Moscow, Russia\\
	Theory Department of Lebedev Physics Institute, Moscow, Russia}

\emph{E-mail}:\,\,\textbf{andrei.marshakov@gmail.com}

\end{document}